\documentclass[sigconf,screen, nonacm]{acmart}
\title{How Expressive Are Graph Neural Networks in the Presence of Node Identifiers?}

  \author{Arie Soeteman}
\affiliation{%
  \institution{Institute for Logic Language and Computation\\University of Amsterdam}
  \city{Amsterdam}
  \country{the Netherlands}}
\email{a.w.soeteman@uva.nl}

\author{Michael Benedikt}
\affiliation{%
  \institution{Department of Computer Science\\
  University of Oxford}
  \city{Oxford}
  \country{UK}}
\email{michael.benedikt@cs.ox.ac.uk}

\author{Martin Grohe}
\affiliation{%
  \institution{Lehrstuhl Informatik 7\\
  RWTH Aachen University}
  \city{Aachen}
  \country{Germany}}
\email{grohe@informatik.rwth-aachen.de}

\author{Balder ten Cate}
\affiliation{%
  \institution{Institute for Logic Language and Computation\\
  University of Amsterdam}
  \city{Amsterdam}
  \country{the Netherlands}}
\email{b.d.tencate@uva.nl}

\usepackage{graphicx} % Required for inserting images
\usepackage{amsthm,amsmath, amsfonts}
\usepackage{geometry}
\usepackage{enumitem}
\usepackage{float}
\setlist{leftmargin=5.5mm,partopsep=0pt,topsep=0pt,itemsep=2pt,parsep=2pt} 
\usepackage{parskip} % This messes with the author names
\usepackage{algorithm}
\usepackage[noend]{algpseudocode}
\usepackage{bm}
\usepackage[font=small,labelfont=bf,textfont=normalfont]{caption}
\usepackage{thm-restate}

\usepackage{xspace}
\usepackage{xcolor}
\usepackage{makecell}
\usepackage{bbm}
\usepackage{mathtools}
\usepackage{stmaryrd}
\usepackage{hyperref}
\usepackage{booktabs}
\usepackage[normalem]{ulem}

\usepackage{tikz-cd}
\usetikzlibrary{positioning}
\usetikzlibrary{shapes.geometric}
\usetikzlibrary{patterns}

\newcommand{\kw}[1]{{\ensuremath{\mathsf{#1}}}\xspace} %michael: generic keyword

\newtheorem{dummy}{Dummy}[section]
\newtheorem{theorem}[dummy]{Theorem}

\newtheorem{proposition}[theorem]{Proposition}

\newtheorem{definition}[theorem]{Definition}

\newtheorem{example}[theorem]{Example}
\newtheorem{lemma}[theorem]{Lemma}
\newtheorem{claim}[theorem]{Claim}

\newcommand{\tdash}{\text{-}}
\newcommand{\agnn}{{\mathcal G}} %better choice?
\newcommand{\altgnn}{{\mathcal H}} %better choice?

\newcommand{\agraph}{G} %a colored graph
\newcommand{\altgraph}{H} %another colored graph
\newcommand{\apgraph}{G^u}
\newcommand{\altpgraph}{H^v}

\newcommand{\rneighborhood}{\upharpoonright
 r}
 
\newcommand{\agg}{\kw{agg}}
\newcommand{\com}{\kw{com}}

\usepackage{todonotes}

\newcommand{\balder}[1]{}
\newcommand{\arie}[1]{}
\newcommand{\martin}[1]{}
\newcommand{\michael}[1]{}

\newcommand{\group}{{\mathcal S}}
\newcommand{\combineclass}{\mathcal C}

\newcommand{\uddl}{\kw{LDDL}}

\newcommand{\false}{\kw{false}}
\newcommand{\WGMLtrue}{\WGML(\top)}
\newcommand{\WGMLmodal}{\WGML(\kw{modal})}

\newcommand{\reals}{{\mathbb R}}
\newcommand{\relu}{\kw{ReLU}}
\newcommand{\sigmoid}{\kw{sigmoid}}

\newcommand{\feature}{\kw{Feature}}
\newcommand{\localsum}{\kw{LocalSum}}
\newcommand{\localmax}{\kw{LocalMax}}
\newcommand{\localmin}{\kw{LocalMin}}
\newcommand{\globalsum}{\kw{GlobalSum}}

\newcommand{\globallocalsum}{\kw{GlobalLocalSum}}

\newcommand{\localmaxF}{(\localmax)--($F$)\xspace}
\newcommand{\localsumF}{(\localsum)--($F$)\xspace}

\newcommand{\unarysemilinear}{\ensuremath{\text{Semilin}_1}}

\newcommand{\localsumsemilinear}{\localsum--Semilinear\xspace}

\newcommand{\localsumcontinuous}{\localsum--Continuous\xspace}
\newcommand{\localsumrelu}{\localsum--\FFN(\ReLU)\xspace}
\newcommand{\localsumsigmoid}{\localsum--\FFN(\ReLU, \sigmoid)\xspace}
\newcommand{\globalsumsemilinear}{\globalsum--Semilinear\xspace}

\newcommand{\localmaxsigmoid}{\localmax--\FFN(\relu,\sigmoid)\xspace}
\newcommand{\localmaxcontinuous}{\localmax--Continuous\xspace}
\newcommand{\localmaxrelu}{\localmax--\FFN(\ReLU)\xspace}

\newcommand{\localmaxsemilinear}{\localmax--Semilinear\xspace}
\newcommand{\localmaxunarysemilinear}{\localmax--\FFN(Semilin$_1$)\xspace}

\newcommand{\ifPos}{\kw{ifPos}}
\newcommand{\ReLU}{\kw{ReLU}}
\newcommand{\Heavi}{\kw{H}}
\newcommand{\FFN}{{\kw{FFN}}}
\newcommand{\FFNs}{{\kw{FFNs}}}

\newcommand{\policy}[2]{\textup{\ensuremath{(#1/#2)}}}

\newcommand{\GNN}{\textsc{GNN}\xspace}
\newcommand{\GNNs}{\textsc{GNN}s\xspace}
\newcommand{\GML}{\textsc{GML}\xspace}
\newcommand{\WGML}{\textsc{WGML}\xspace}
\newcommand{\ordinvfo}{\textsc{OrdInvFO} \xspace}
\newcommand{\ML}{\textsc{ML}\xspace}

\newcommand{\HASH}{\textsc{hash}\xspace}

\newcommand{\step}{\kw{step}}
\newcommand{\test}{\kw{test}}
\newcommand{\stay}{\kw{stay}}

\newcommand{\adom}{\kw{Adom}}

\newcommand{\myeat}[1]{}

\newcommand{\calL}{\mathcal{L}}

\newcommand{\calQ}{\mathcal{Q}}

\newcommand{\col}{\textup{col}\xspace}
\newcommand{\emb}{\textup{emb}\xspace}
\newcommand{\key}{\textup{key}\xspace}
\newcommand{\val}{\textup{val}\xspace}

\newcommand{\lab}{\textup{lab}}
\newcommand{\refine}{\textup{refine}}

\newcommand{\multiset}[1]{\{\!\!\{#1\}\!\!\}}
\newcommand{\multisets}[1]{\mathcal{M}(#1)}

\newcommand{\widebox}[1]{%
  \fbox{%
    \parbox{0.50\textwidth}{\centering #1}%
  }%
}

\newcommand{\sem}[1]{\left\llbracket#1\right\rrbracket}
\DeclareMathOperator{\colr}{cr}

\newcommand{\Nat}{{\mathbb N}}

\newcommand{\fbisim}{\ensuremath{\xrightarrow{\text{bisim}}}}
\newcommand{\covers}{\ensuremath{\xrightarrow{\text{cov}}}}

\newcommand{\FO}{\textsc{FO}}
\newcommand{\FOC}{\textsc{FO+C}}
\newcommand{\ord}{\operatorname{ord}}
\newcommand{\ordinvfoc}{\textsc{OrdInvFO+C} \xspace}

\newcommand{\num}{\textup{num}\xspace}

\begin{document}
\sloppy % Enable math line breaking, I think this is better than the overflows we have otherwise

\begin{abstract}
Graph neural networks (GNNs) are a widely used class of machine learning models for graph-structured data, based on local aggregation over neighbors. GNNs have close connections to logic. In particular, their expressive power is linked to that of modal logics and bounded-variable logics with counting. 
In many practical scenarios, graphs processed by GNNs have node features that act as unique identifiers. In this work, we study how such identifiers affect the expressive power of GNNs.
We initiate a study of the \emph{key-invariant} expressive power of GNNs, inspired by the notion of order-invariant definability in finite model theory: \emph{which node queries that depend only on the underlying graph structure
can GNNs express on graphs with unique node identifiers?}
We provide answers for various classes of GNNs with local max- or sum-aggregation.

\end{abstract}

{
  \setlength{\parskip}{0pt}
\maketitle
}% \setlength{\parindent}{0pt}

\section{Introduction} \label{sec:intro}
Graph Neural Networks (\GNNs), and specifically message-passing neural networks \cite{scarselli2008graph,gilmer2017neural} are deep learning models for graph data. Since \GNNs act natively on graphs, their predictions are invariant under isomorphisms. \GNNs play a key role in machine learning on graphs and have been applied successfully in numerous domains such as molecular property prediction \cite{besharatifard2024review}, system monitoring \cite{huang2024link} and knowledge graph completion \cite{zhu2021neural,cucala2022explainable}. \GNNs are defined in a finite sequence of layers. Each layer computes a real-valued \emph{embedding vector} for each node $u$ of the graph, by taking the embedding vectors produced by the previous layer as input, applying an \emph{aggregation function} to the vectors of the neighbors of $u$, and applying a \emph{combination function} to combine the result of the aggregation with the current embedding vector of $u$. Typical aggregation functions include maximum and sum. The most common choice for combination function is a feedforward neural network. The initial embedding vector of each node is an encoding of the node features, and the embeddings produced by the final layer are used to compute an output, such as a Boolean value in the case of node classification.

Logic and descriptive complexity have played a central role in the study 
of \GNNs over the last years. For example,  
indistinguishability by GNNs with sum-aggregation has been characterized
in terms of the \emph{color refinement} test (which can be equivalently cast in terms of  \emph{graded bisimulations}) \cite{howpowerful,goneural}, and, 
when viewed as a query language over node labeled graphs,
\GNNs with sum-aggregation subsume \emph{graded modal logic}
(GML) \cite{barceloetallogical} and are contained in \emph{first-order logic with counting} (FO+C) \cite{grohedescriptivegnn}. 
Subsequently, the field has seen numerous logical characterizations of \GNN with specific aggregation or combination functions (e.g., \cite{schonherr2025logical,barcelo2025logical}), or with more expressive message-passing architectures (e.g., \cite{goneural, geerts2022expressiveness, soeteman2025logical, hauke2025aggregate}). Connections to logic have also been
used, for instance, to obtain decision procedures for static analysis problems involving \GNNs~\cite{benedikt2024Decidability}.

One conclusion arising from this line of work is that
the basic \GNN architecture, as we described it above, is limited in its expressive power: there are
many natural queries that cannot be expressed by a \GNN, which limits the range of classifiers that 
can potentially be learned when using \GNNs as graph learning architecture.
Several approaches for addressing this shortcoming have been explored. In particular, \GNNs can approximate every function over bounded size graphs, if these graphs are augmented with randomly generated node features \cite{Abboud2021Surprising}. This result builds on the observation that with high probability random node features are unique identifiers that distinguish all nodes in a graph.
% \emph{individualizes} a graph with high probability, generating unique features for each node. 
Note however that adding random node features means that 
 \GNNs are no longer guaranteed to be invariant for isomorphisms of the original graph (they are isomorphism-invariant only in expectation). Other 
approaches have been developed for generating unique node identifiers~\cite{dasoulas2020coloring,franks2023systematic, pellizzoni2024expressivity} but they suffer from the same problem. 

This raises a natural question, which we investigate here:
\begin{itemize}
\item[(*)]
\emph{Which isomorphism invariant properties can \GNNs express on graphs with unique node identifiers?}
\end{itemize}

We have motivated this question based on the 
study of randomized \GNNs, but it is of broader
interest, since unique node identifiers
are often present in practical applications.
We give two examples. 

Firstly, geometric \GNNs work on graphs where node vectors have been enriched with geometric attributes, such as coordinates in Euclidean space or distances between nodes. Assuming that no two nodes occupy the same point in space, coordinates are unique node identifiers. In this geometric setting, isomorphic graphs with different geometric attributes yield different predictions, but it is generally desirable that the output of the \GNN is invariant (or equivariant) under group transformations, such as translations or rotations in Euclidean space (see \cite{han2025survey} for an overview). Answers to (*) provide a lower bounds for the expressiveness of these models by assessing the expressiveness of \GNNs invariant for \emph{all} bijections on the geometric attributes, instead of only the bijections from a specific group (cf.~also~Section~\ref{sec:groups}).

Secondly, inspired by transformer architectures, a recent trend in graph machine learning is the use of positional encodings to enhance the expressiveness of graph neural networks. Positional encodings are additional node features that ideally reflect the structural role a node plays in a graph. Various approaches to obtain such positional encodings have been proposed, the most important of which are based on random walks \cite{RampasekGDLWB22, grotschla2024benchmarking}, local substructures \cite{YouGYL21,BouritsasFZB23}, and eigenvectors \cite{dwivedi2020generalization,KreuzerBHLT21}. There are also methods for learning effective positional encodings~\cite{canturk2024graphpositional,HuangL0YZJ024,KanatsoulisCJLR25}. The node features provided by positional encodings are not necessarily unique node identifiers, though often they are for ``typical graphs''. %\michael{Citation/pointer to something more precise about typicality?} 
Furthermore, they are not necessarily isomorphism invariant, which raises 
the question which isomorphism-invariant properties can be expressed using such positional encodings.
We do not pursue this here, but our results provide a starting point for investigating this question for different types of positional encodings.

% Secondly, graph transformers generate positional encodings from the Laplacian of a graph \cite{dwivedi2020generalization}. \red{Expand a bit on the relevant methods here, and give a bit of explanation on what these positional encodings achieve.} 

A completely different motivation for our research comes from work in finite model theory on order-invariant logics (see Section~\ref{para:order-inv} for the definition). It was realized early that over finite structures, order-invariant first-order logic is more expressive than plain first-order logic~\cite{Gurevich84}. The expressiveness of order-invariant logics has been an active research topic. For example, it has been proved that order-invariant first-order logic can only express (weakly) local properties \cite{GroheS00}, and interesting connections to descriptive complexity and proof complexity have been established \cite{ChenF12}.

\begin{figure}
\begin{center}
\small
\begin{tikzcd}[column sep=5mm]
& \textbf{Arbitrary} & \\
\textbf{Semilinear} \arrow[ur] && \textbf{Continuous} \arrow[ul] \\
\text{\FFN($\ReLU,\Heavi$)} \arrow[u] && \text{\FFN(\ReLU,$\sigmoid$)} \arrow[u] \\
& \textbf{\FFN(\ReLU)} \arrow[ul] \arrow[ur] 
\end{tikzcd}
\end{center}
\caption{Combination function classes ordered by increasing expressive power. Here $\sigmoid$ denotes the activation $\sigmoid(x) = \frac{1}{1+e^{-x}}$, and $\Heavi$ is the Heaviside step function.}
\label{fig:combination_fns_hierarchy}
\end{figure}
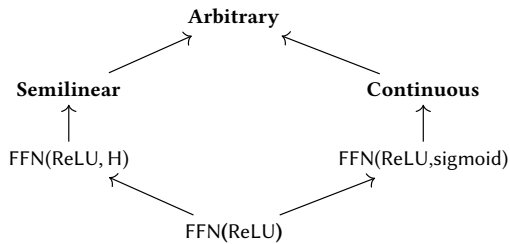

\subsection*{Contributions}
Inspired by the notion of order-invariance in finite model theory, we define a \emph{key-invariant} \GNN
to be a \GNN that operates on graphs including
a key (i.e., a real-valued node feature acting as 
a unique node identifier), such that
the predictions of the \GNN depend only on the
underlying graph and not on the keys
(so long as each node has a different key).
This allows us to rephrase the question (*) as 
asking what is the expressive power of key-invariant \GNNs.

We study the expressibility of \emph{node queries} (i.e., queries that take as input a pointed graph and that return true/false) by key-invariant \GNNs using
either \localmax or \localsum aggregation, with different types of combination functions---$\ReLU$-FFNs,  continuous functions, semilinear functions and arbitrary functions (cf.~Figure~\ref{fig:combination_fns_hierarchy}).
In each setting, we study 
the expressible node queries from several different angles:

\begin{enumerate}
    \item Closure properties, i.e., invariance for suitable variants of bisimulation and/or color refinement;
    \item Lower bounds and upper bounds in terms of variants of graded modal logic (GML), order-invariant FO, and first-order logic with counting (FO+C);
    \item Collapse and non-collapse results,  addressing the question when key-invariant GNNs can be reduced to key-oblivious GNNs;
    \item Specific case studies, such as $\calQ_{\text{even}}$ (the node query that tests whether the given node has an even number of neighbors) and
    isomorphism type queries $\calQ_{\apgraph}$ (i.e., node queries that test whether the connected component of an input pointed graph is isomorphic to a specific graph $\apgraph$). 
\end{enumerate}

\begin{table*}[t]
\scriptsize
\centering
\renewcommand{\arraystretch}{1.05}
% Left table
\begin{minipage}{0.48\textwidth}
\centering
\caption{Logic bounds, closure properties and collapse results for the expressive power of key-invariant \localmax \GNNs.}
\begin{tabular}{p{0.3\textwidth} p{0.27\textwidth} p{0.27\textwidth}}
\toprule
Combination fns & Subsumes & Subsumed by \\
\midrule
\FFN(\relu) 
& \makecell[l]{%
  $\WGMLtrue$ (Thm.~\ref{thm:WGML encoded by localmax}(a))
}& \\
\addlinespace
Continuous & \makecell[l]{%
  $\WGMLmodal$\\(Thm.~\ref{thm:WGML encoded by localmax}(b))
}&
\makecell[l]{%
  Queries closed under inverse\\
  functional bisimulations\\
  (Thm.~\ref{thm:max yes-open closed under inverse fbisim})
}\\
\addlinespace
\makecell[l]{Continuous with \\ \policy{>0}{<0} acceptance}
& \multicolumn{2}{c}{%
  \widebox{Key-oblivious \localmaxrelu = Modal Logic (Thm.~\ref{thm:collapseyesnolocalmax})}
} \\
\addlinespace
Semilinear 
& \makecell[l]{%
  \uddl (Thm.~\ref{thm:uddllower});\\
  Isomorphism types\\(Thm.~\ref{thm:localmaxsemilinear recognizes up to isomorphism})
}& \\
Arbitrary && \ordinvfo (Thm.~\ref{thm:orderinvariant}) \\
\addlinespace
\bottomrule
\end{tabular}
\label{tab:localmxsummary}
\vspace{0.05cm}
\end{minipage}
\hfill
% Right table
\begin{minipage}{0.48\textwidth}
\centering
\caption{Logic bounds, closure properties and collapse results for the expressive power of key-invariant \localsum \GNNs.}
\begin{tabular}{ p{0.3\textwidth} p{0.27\textwidth} p{0.27\textwidth}}
\toprule
Combination fns & Subsumes & Subsumed by \\
\midrule
\FFN(\relu)  &&
\makecell[l]{\ordinvfoc (Thm \ref{thm:SumReluinOrdInvFOC})} \\
\addlinespace
\makecell[l]{\FFN(\relu) with \\ \policy{\geq 1}{\leq 0} acceptance}
& \multicolumn{2}{c}{%
  \widebox{Key-oblivious \localsumrelu
  (Thm.~\ref{thm:localsumgapcollapse})}
} \\
\addlinespace
Continuous  & 
& Queries closed under inverse coverings
  (Thm.~\ref{thm:closed-under-coverings}) \\
\addlinespace
\makecell[l]{Continuous with \\  \policy{>0}{<0} acceptance}
& \multicolumn{2}{c}{%
  \widebox{Key-oblivious \localsumcontinuous\\ (all strongly local CR-invariant queries)
  (Thm.~\ref{thm:>0/<0 localsumcontinuous collapses to key-oblivious})}
} \\
\addlinespace
Semilinear + $(\cdot)^2$ 
& \mbox{Isomorphism types} (Thm.~\ref{thm:iso test with local sum}) & \\
\addlinespace
Arbitrary 
& \multicolumn{2}{c}{%
  \widebox{All strongly local queries (Thm.~\ref{thm:localsumarbunlimited})}
} \\
\bottomrule
\end{tabular}
\label{tab:localsumsummary}
\vspace{0.4cm}
\end{minipage}
\end{table*}

Our main results are summarized in Tables~\ref{tab:localmxsummary} and~\ref{tab:localsumsummary}. 
Based on these results, the relative key-invariant
expressive power of the various classes of 
\GNNs is depicted in
Figures~\ref{fig:max GNNs hierarchy} and~\ref{fig:sum-hierarchy}.
We highlight a few interesting conclusions that can be drawn from these results:
\begin{itemize}
    \item Whereas without keys every \localmax \GNN is equivalent to one with \FFN(\ReLU) combination functions, in the key-invariant setting different combination classes generate a strict hierarchy  (cf.~Figure~\ref{fig:max GNNs hierarchy}).

    \item Whereas without keys every \localsum \GNN is equivalent to one with continuous combination functions, in the key-invariant setting discontinuous functions increase expressive power (cf.~Figure~\ref{fig:sum-hierarchy}).
    
    \item Key-invariant \localmax \GNNs with arbitrary combination functions are subsumed by order-invariant FO. Key-invariant \localsum \GNNs with arbitrary combination functions, on the other hand, are expressively complete: they can express \emph{all} strongly local node queries. 
    %\item We show that key-invariant \localsum \GNNs with \FFN(\ReLU) combination functions are subsumed by order-invariant FO+C.
    \item The expressive power of key-invariant \localmax and \localsum \GNNs with \emph{continuous} combination functions strongly depends on the acceptance policy. 
    If we require the \GNN to output a positive on yes-instances and a negative value on no-instances (and  never output $0$), key-invariant \GNNs with continuous combination functions are no more expressive than ordinary \GNNs on unkeyed graphs. Without this restriction on the output, keys \emph{do} increase expressive power.
\end{itemize}
%\michael{I am not sure that we should emphasize the big table as the big takeaway although we should include it. Instead or in addition, we could emphasize
%the connections: of local max to order invariant FO, a new dynamic logic \uddl, variants of GML, and functional bisimulations; for local sum to CR-invariance, coverings, and FO+C}

% Moved this to the discussion section
%\balder{It we are still worried about getting logicians interested (although I think we are ok), perhaps we can list here some of the core tools that we 
% draw on: counting logics; modal and dynamic logics; bisimulations, coverings and tree unravellings; 
% collapse results in embedded finite model theory; moment-injective encodings; ...}

\subsection*{Outline}
%After a brief discussion of related work in Section~\ref{sec:related},
Section~\ref{sec:prelims} defines \GNNs over keyed graphs and introduces the notion of key-invariance.
Section~\ref{sec:oblivious} reviews relevant (mostly) known results for key-oblivious \GNNs, which link their expressive power to logics.
Section~\ref{sec:max} and Section~\ref{sec:sum} contain our results for key-invariant \localmax \GNNs and \localsum \GNNs, respectively. Finally,
Section~\ref{sec:discussion} briefly discusses a number of further related aspects.

\begin{figure}
    \centering
    \includegraphics[width=0.95\linewidth]{pictures/Invariant_GNN_MAX_circles.jpg}
    \caption{Relative expressive power of key-invariant \localmax \GNNs with varying acceptance policies and combination functions. Key-invariant \policy{>0}{<0} \localmaxcontinuous \GNNs collapse to key-oblivious \localmax \GNNs. $\Diamond^{\geq 2} \top$ and $\Diamond^{\geq 2} p$ are formulas in \WGML (Definition \ref{def:weaklygradedml}), $\Diamond^{=1} \top$ is a formula in \GML (section \ref{subsec:Logics and key-oblivious}), and $\langle \step ; \step\rangle^{=1}$ is a formula in \uddl (section \ref{sec:uddl}). $\calQ_{C^u_3}$ is the isomorphism type of a triangle with a distinguished node $u$ and uniform labeling. It remains open for key-invariant \localmax \GNNs whether semilinear combination functions have the same expressive power as arbitrary combination functions, and whether they subsume the expressive power of continuous combination functions.}
    \label{fig:max GNNs hierarchy}
\end{figure}
\begin{figure}
    \centering
\includegraphics[width=0.95\linewidth]{pictures/Invariant_GNN_SUM_circles.jpg}
    % \arie{Can we find a single non-expressible query to place in $(>0/<0)$ Continuous $\setminus$ semilinear. We don't know whether `even neighbors' is in semilinear.}
    \caption{Relative expressive power of key-invariant \localsum \GNNs with varying acceptance policies and combination functions. Key-invariant \policy{\geq 1}{\leq 0} \localsumrelu \GNNs collapse to key-oblivious \localsumrelu \GNNs, and key-invariant \policy{>0}{<0} \localsumcontinuous \GNNs collapse to key-oblivious \localsumcontinuous \GNNs. $\calQ_{G^u_{\triangle p}}$ is the isomorphism type of a rooted triangle with a single $p$ node (see Example \ref{ex:localsum} and Theorem \ref{thm:iso test with local sum}). It remains open whether key-invariant \localsumrelu \GNNs are more expressive than key-oblivious \localsumrelu \GNNs.}
    \label{fig:sum-hierarchy}
\end{figure}

\section{Preliminaries}
\label{sec:prelims}

\subsection{Data model}
Fix a finite set $\Pi=\{p_1, \ldots, p_n\}$ of
atomic binary node features.
A \emph{graph} is a triple $\agraph=(V,E,\lab)$ where $V$ is a set of nodes, $E \subseteq V\times V$ is an irreflexive and symmetric edge relation and $\lab: V \to \{0,1\}^{|\Pi|}$ is a labeling function where $\lab(v)$ is called the label of $v$. 
We often use $V(\agraph), E(\agraph)$ to denote the vertices and edges of $\agraph$. A \emph{valued graph} is a tuple $(\agraph,\val)$ with a value function $\val: V(\agraph) \to \reals$ that assigns a real value to each node. A \emph{keyed graph} is a valued graph where the value function is injective, in this case we also refer to $\val$ as a keying and refer to it as `$\key$'. We call $(\agraph,\val)$ a valued extension of $\agraph$, or a keyed extension when $\val$ is a keying and often write $\agraph_\val, \agraph_\key$ for valued and keyed extension of $\agraph$. A \emph{pointed graph} $\apgraph$ is a graph with a distinguished node $u$, and likewise for pointed valued graphs $\apgraph_\val$ and pointed keyed graphs $\apgraph_\key$. An isomorphism between pointed graphs $\apgraph, \altpgraph$ is a bijection $f: V(\agraph) \to V(\altgraph)$ that preserves labels and edges both ways, and where $f(u)=v$. 

A \emph{node query} $\calQ$ is an isomorphism closed class of pointed graphs. We denote its complement by $\calQ^c$. We will investigate expressive power in terms of expressible node queries. We give two examples of node queries that we use throughout this work. $\calQ_{\text{even}}$ contains all pointed graphs $\apgraph$ where $u$ has an even number of neighbors. $\calQ_{\apgraph}$ is the isomorphism type of $\apgraph$---the query that contains pointed graph $\altpgraph$ if and only if the connected component of $v$ in $\altpgraph$ is isomorphic to $\apgraph$.

\subsection{GNNs}
\label{sec:GNN-prelims}
In what follows, we will use the notation $\multiset{x_1, \ldots, x_n}$ for multisets, and $\oplus$ for concatenation. 
We will denote by $\multisets{X}$ the set of all finite multisets of elements from $X$.

An \emph{embedded} graph $(\agraph,\emb)$ is a graph with an \emph{embedding function}
$\emb: V(\agraph) \to \reals^D$ for some $D>0$. A $(D,D')\tdash\GNN$ consists of a finite sequence of layers $(\calL_1, \dots, \calL_\ell)$. Each layer is a pair $\calL_i = (\agg_i, \com_i)$ consisting of an aggregation function $\agg_i: \multisets{\reals^{D_{i}}} \to \reals^{D_{i}}$ and a combination function $\com_i: \reals^{2 D_{i}} \to \reals^{D_{i+1}}$, where $D=D_1, D'=D_{\ell+1}$. A layer $\calL = (\agg, \com)$ computes a function from embedded graphs to embeddings defined by:
\begin{align*}
    \calL(\agraph,\emb)(u) &= \com(\emb(u), \agg(\multiset{\emb(v) \,\mid (u,v) \in E(G)}))
\end{align*}
We sometimes abuse notation and write $\calL(\agraph,\emb)$ for $(\agraph, \calL(\agraph,\emb))$. Accordingly, a \GNN $\agnn = (\calL_1, \dots, \calL_\ell)$ computes function $\agnn(\agraph,\emb) := \calL_\ell \circ \dots \circ \calL_1(\agraph,\emb)$. Given a valued pointed graph $\apgraph_\val$ we write $\agnn(\apgraph_\val)$ for $\agnn(\agraph_\val,\emb_G)(u)$ where for all $v \in V(\agraph)$, $\emb_G(v)$ is the $n+1$-vector $\lab(v) \oplus(\val(v))$. 
In other words, $\agnn(\apgraph_\val)$ is the embedding for node $u$ obtained by applying $\agnn$ to the graph with the initial embedding defined by the labels and values of $\agraph_\val$.

\paragraph{\GNN node-classifiers}
A \GNN node-classifier $\agnn$ is a $(|\Pi|+1,1)$-\GNN (where $\Pi$ is the set of propositions) which we interpret as accepting a pointed valued graph $\apgraph_\val$ if $\agnn(\apgraph_\val)>0$ and rejecting $\apgraph_\val$ otherwise. We will also refer to this as the
``$\policy{>0}{\leq 0}$ acceptance policy'', and we define several other, more restrictive acceptance policies below. See also Section~\ref{sec:policies} for 
further discussion. 

\paragraph{Key invariance} A \GNN node-classifier is \emph{key-invariant} if it classifies all keyed extensions of a pointed graph equally. Such a key-invariant classifier thus defines a node query $\calQ_\agnn$ that contains exactly the pointed graphs of which the keyed extensions are accepted by $\agnn$. A \GNN $\agnn$ expresses a node query $\calQ$ if $\calQ = \calQ_\agnn$. 

\paragraph{Key-oblivious} A \emph{key-oblivious \GNN node-classifier} is defined similarly as above but on graphs without values, i.e. it is a $(|\Pi|,1)$\tdash\GNN that uses as initial embedding $\emb(v)=\lab(v)$. Every key-oblivious \GNN can be viewed as a key-invariant \GNN that simply ignores the node values. Hence, the expressive power of key-invariant \GNNs is lower bounded by that of key-oblivious \GNNs.

We will consider several further restrictions of the 
default $\policy{>0}{\leq 0}$ acceptance policy of \GNN node-classifiers. A $\policy{>0}{<0}$ \GNN node-classifier never outputs $0$ on keyed graphs, and a $\policy{\geq1}{
\leq0}$ \GNN node-classifier never outputs values in the range $(0,1)$ on keyed graphs. 

\begin{example}\label{ex:gnn}
   Consider the node query \[\calQ_{\Diamond p_i} = \{\apgraph \mid \text{there is a $v$ such that $E(u,v)$ and $\lab(v)_i=1$}\}\] 
   consisting of all
   pointed graphs $\apgraph$ 
   in which the distinguished node $u$ has at least one 
   neighbor whose label includes $p_i$. This  query is expressible 
   by the key-oblivious \GNN consisting of a single layer $\calL_1=(\agg_1,\com_1)$, where
   \begin{itemize}
       \item $\agg_1:\multisets{\reals^D}\to\reals^D$ (with $D=|\Pi|+1$)
       returns the coordinate-wise maximum of a given multiset of vectors.
       \item $\com_1:\reals^D\to \reals$ is the $i$-th projection.
   \end{itemize}
   It can be verified that, on input $\apgraph$, this GNN returns $1$ if $\apgraph\in\calQ_{\Diamond p_i}$ and 0 otherwise. Thus, it is in fact a 
   $\policy{\geq 1}{\leq 0}$ GNN node-classifier.
\end{example}

\paragraph{Aggregation and Combination functions}

We investigate the expressiveness of key-invariant \GNNs where either all aggregation functions are coordinate-wise maximum (\localmax \GNNs) or all aggregation functions are coordinate-wise sum (\localsum \GNNs).  Given a class of real-valued functions $F: \reals^n \to \reals^m$ for $n,m \in \mathbb{N}$, we write \localmaxF and \localsumF to refer to \localmax \GNNs, respectively \localsum \GNNs, containing only combination functions in $F$. When no $F$ is specified, it is to be understood that all real valued functions are allowed as combination functions.

We consider several classes of combination functions, shown in Figure \ref{fig:combination_fns_hierarchy}.
We denote by \FFN($\mathcal{A}$),
where $\mathcal{A}$ is a set of real-valued unary functions $f:\reals\to\reals$, the set of
feedforward networks (of arbitrary finite depth and with rationals weights and biases) using activation
functions from $\mathcal{A}$. Formally, we can 
think of \FFN($\mathcal{A}$) as consisting of all functions $f:\reals^n\to \reals^m$ 
definable as compositions of affine functions with rational coefficients 
and unary functions in $\mathcal{A}$. In practice, $\mathcal{A}$ will typically consist of ReLU, the sigmoid function $\sigmoid(x) = \frac{1}{1+e^{-x}}$ 
or the Heaviside step
function
\[
    \Heavi(x) = \begin{cases}
      1 & \text{if $x\geq 0$} \\
      0 & \text{otherwise}
      \end{cases}
    \]
Thus,
for example, \FFN(\ReLU) refers to the familiar class of feed forward networks with ReLU activations. 
We also consider more general classes of combination functions, such as all continuous functions, all functions, and  \emph{semilinear functions}:
any function $f:\reals^n\to \reals^m$ that is definable
by a first-order formula (without parameters)
$\phi(x_1, \ldots, x_n,y)$ over the
structure $(\reals,+,<,0,1)$, i.e., such that $(\reals,+,<,0,1)\models\phi(a_1,\ldots,a_n,b)$ iff $b=f(a_1,\ldots,a_n)$.

\begin{restatable}{lemma}{lemSemilinearOne}
    The semilinear functions are precisely all real-valued functions that can be obtained through composition 
    from (i) affine functions with rational coefficients, 
    and (ii) the ternary function $\ifPos$, where
    \[\ifPos(x,y,z) = 
      \begin{cases}
         y & \text{if $x>0$} \\
         z & \text{otherwise}
      \end{cases}\]
\end{restatable}
For \emph{unary} semilinear functions, there is another
characterization that will be helpful. Let \unarysemilinear~denote the set of all unary semilinear functions $f:\reals\to\reals$. We note that this includes $ReLU$ and $\Heavi$ but not, for example, $\sigmoid$.

\begin{restatable}{lemma}{lemSemilinearTwo}
\label{lem:unary semilinear = relu,H}
    Every unary semilinear function is
    definable by a FFN using $\relu$ and 
    $\Heavi$ as activation functions. Consequently, \FFN(\relu,\Heavi) is expressively equivalent to \FFN(\unarysemilinear).
\end{restatable}

\section{Key-Oblivious \GNNs}
\label{sec:oblivious}

In this section, we review known logical and combinatorial characterisations of the 
expressive power of key-oblivious \GNNs, 
which will provide the backdrop for the 
results in Section~\ref{sec:max} and 
Section~\ref{sec:sum}. We also show that
the node query $\calQ_{\text{even}}$ is not
expressible by a key-oblivious \localsumrelu \GNN, 
which will be helpful as this query will serve as a 
running example in subsequent sections.

\subsection{Color Refinement and Bisimulation} \label{subsec:colorrefbisimulation}

The distinguishing power of key-oblivious \localsum \GNNs
with arbitrary combination functions can be characterized in 
terms of the \emph{color refinement} algorithm, which originates in the study of the graph isomorphism problem.
We briefly recall it. For convenience,
we will identify colors with natural numbers.
A \emph{node coloring}
for a graph $\agraph$ is then a
map $\col:V(G)\to\Nat$. 
Fix some injective map $\HASH:2^\Pi \uplus (\Nat\times\multisets{\Nat})\to\Nat$, where
$2^\Pi$ is the set of all node labels.
For a graph $\agraph$, by the 
\emph{initial coloring} of $\agraph$ we
will mean the coloring $\col_\agraph$ given by
$\col_\agraph(u)=\HASH(\lab(u))$.
The color refinement algorithm (\refine) takes as input a graph $\agraph$, a node coloring $\col$ and an integer $d\geq 0$, and it produces a 
new node coloring for the same graph,
% denoted
% $\WL(G,\col,d)$,
as follows:
\begin{algorithm}[H]
\begin{algorithmic}[1]
\Function{\refine}{$\agraph,\col, d$}
    \State $\col^0 = \col$
    \For{$i = 1, \dots, d$}
             \State $\col^{i} := \{ u \mapsto \text{\HASH} (\col^{i-1}(u),$ 
             
             \State \hspace{22mm} $\multiset{\col^{i-1}(v) | (u,v) \in E(G)})\mid u\in V\}$
    \EndFor
    \State \Return $\col^d$
\EndFunction
\end{algorithmic}
\end{algorithm}
We will denote by $CR(\apgraph)$ the infinite
sequence $(c_1,c_2,\ldots)\in\Nat^\Nat$
where $c_i=\refine(\agraph,\col_G,i)(u)$, where $\col_G$ is the initial coloring. That is, 
$c_i$ is the color of $u$ after $i$ rounds of color refinement. We also denote  by
$CR^{(i)}(\apgraph)$ the color $\refine(\agraph,\col_G,i)(u)$.
\begin{definition}[CR-invariant]
\label{def:CR invariant queries}
    Node query $\calQ$ is CR-invariant if for all pointed graphs $\apgraph$ and $\altpgraph$,
    $CR(\apgraph)=CR(\altpgraph)$ implies that $\apgraph \in\calQ$ iff
    $\altpgraph \in\calQ$.
    We say that 
    $\calQ$ is \emph{r-round CR-invariant}
    if  for all pointed graphs $\apgraph$ and $\altpgraph$, $CR^{(r)}(\apgraph)=CR^{(r)}(\altpgraph)$ implies that $\apgraph\in\calQ$ iff $\altpgraph\in\calQ$.
\end{definition}

For a pointed graph $\apgraph$ and $r>0$, we denote by $\apgraph\rneighborhood$ the pointed graph $H^u$, where
$H$ is the induced subgraph of $G$ containing nodes
reachable from $u$ in at most $r$ steps.
\begin{definition}[Strongly Local]     
\label{def:local queries}
Node query $\calQ$ is strongly local if there is an $r \geq 0$,
such that, for all pointed graphs $\apgraph$,
$\apgraph\in\calQ$ iff $\apgraph \upharpoonright
 r\in\calQ$.
\end{definition}

\begin{restatable}{theorem}{thmOblivLocalSumCR}
    \label{thm:localsumcontinuous expresses all local CR-invariant queries}
    Let $\calQ$ be a node query. The following are equivalent:
    \begin{enumerate}
        \item $\calQ$ is expressed by a key-oblivious \localsum \GNN;
        \item $\calQ$ is expressed by a key-oblivious   \localsumcontinuous \GNN  with \policy{\geq 1}{\leq 0} acceptance policy;
        %\michael{Did we say somewhere what the default is when we omit the combination functions?}
        \item $\calQ$ is strongly local and CR invariant;
        \item  $\calQ$ is $r$-round CR-invariant for some $r\geq 0$.
    \end{enumerate}
\end{restatable}
 The equivalence between (1) and (4) is well known \cite{howpowerful,goneural,wagstaff2022universal}. We show the equivalence with (2) and (3), since it will be relevant for us later, using arguments based on \cite{amir2023neural} and \cite{Otto04,Otto12}.

An analogous result holds for \localmax \GNNs, using bisimulations instead of color refinement.
Bisimulations can be defined in several equivalent ways. One is algorithmic, namely by modifying color refinement so that in each round we hash the \emph{set} of neighbor colors instead of the multiset. The more standard definition, however, is as follows:

\begin{definition}[Bisimulation]
Given pointed graphs $\apgraph,\altpgraph$, a bisimulation is a relation $B \subseteq V(\agraph) \times V(\altgraph)$ such that $B(u,v)$, and whenever $B(u_0,v_0)$ for any $u_0\in V(\agraph), v_0 \in V(\altgraph)$ then 
\begin{enumerate}
    \item $u_0$ and $v_0$ have the same label;
    \item For each neighbor $u_1$ of $u_0$ in $\agraph$ there is a neighbor $v_1$ of $v_0$ in $\altgraph$ such that $B(u_1,v_1)$;
    \item For each neighbor $v_1$ of $v_0$ in $\altgraph$ there is a neighbor $u_1$ of $u_0$ in $\agraph$ such that $B(u_1,v_1)$.
\end{enumerate}
\end{definition}
\begin{restatable}{proposition}{propbisimulationlocalmaxobliv}[from~\cite{schonherr2025logical,bernardowalega}]
\label{prop:bisimulation-localmaxobliv}
    Let $\calQ$ be a node query, then the following are equivalent:
    \begin{enumerate}
        \item 
    $\calQ$ is expressed by a key-oblivious \localmax \GNN;
    \item $\calQ$ is expressed by a key-oblivious \localmaxrelu \GNN with \policy{\geq 1}{\leq 0} acceptance policy;
    \item 
    $\calQ$ is strongly local and closed under bisimulation.
    \end{enumerate}
\end{restatable}

\subsection{Logics and Logical Characterizations}
\label{subsec:Logics and key-oblivious}

We now review some uniform bounds on the queries expressible by \GNNs in terms of logics.
%michael{I gather that here we only mention logics that are related to prior work, saving (e.g.) order-invariant FO for later}
%michael{Does it make sense to wait on the local sum-related logics until later? The way we do it now, we have propositions and corollaries that do not resonate until we get to local-sum}
%\balder{It would make sense to have all the results that are just about key-oblivious GNNs  here in the preliminaries, so they don't disrupt the story line later on, but I wouldn't insist. }

\paragraph{Graded modal logic (\GML)}
The formulas of \GML are given by
the recursive grammar $\phi ::= p \mid \top \mid \phi\land\psi\mid \neg\phi\ \mid \Diamond^{\geq k}\phi$, where $p\in \Pi$ and $k\in\mathbb{N}$. \emph{Satisfaction} of such a formula at a node $u$ in a graph $\agraph$ (denoted: 
$\apgraph\models\phi$) is defined inductively, as usual, where $\apgraph\models p_i$ iff $\lab(u)_i=1$, the Boolean operators have the standard interpretation, and $\apgraph\models\Diamond^{\geq k}\phi$
iff $|\{v\mid (u,v)\in E \text{ and } \agraph^v\models\phi\}|\geq k$.
We use $\Diamond\phi$ as
shorthand for $\Diamond^{\geq 1}\phi$, $\Box\phi$ as shorthand for $\neg\Diamond\neg\phi$, $\Diamond^{=k}\phi$ as shorthand for $\Diamond^{\geq k} (\phi) \wedge \neg \Diamond^{\geq k+1} (\phi)$ and $\Diamond^{\leq k} \phi$ as shorthand for $\neg \Diamond^{\geq k+1} \phi$.
Every \GML-formula $\phi$ gives rise to a
node query $\calQ_\phi := \{\apgraph\mid \apgraph\models\phi\}$. The fragment of \GML where
$\Diamond^{\geq k}$ is allowed only for $k=1$ is known as 
\emph{modal logic (ML)}. We say a query $\calQ$ is expressible in a logic, or is in a logic, if it is equivalent to $\calQ_\phi$ for some $\phi$ in the logic. 
\begin{theorem}[\cite{schonherr2025logical,bernardowalega}]
\label{thm:ML-localmaxobliv}
   Let $\calQ$ be a node query. The following are equivalent:
    \begin{enumerate}
        \item $\calQ$ is expressible by a key-oblivious \localmax \GNN;
        \item $\calQ$ is expressible by a key-oblivious \localmaxrelu\ \GNN with \policy{\geq 1}{\leq 0} acceptance policy;
        \item $\calQ$ is expressible in modal logic. 
        \end{enumerate}
\end{theorem}

\begin{theorem}[\cite{barceloetallogical}] ~
\label{thm:GMLloweroblivious}
\begin{enumerate}
    \item Every node query expressible in $\GML$ is 
     expressible by a key-oblivious \localsumrelu \GNN;
     \item The converse holds for first-order node queries, i.e., if a node query is definable
in first-order logic and expressed by a key-oblivious \localsum \GNN, then it is expressible in \GML. \footnote{In fact as observed in \cite{ahvonen2024logical}, it follows from results in~\cite{ElberfeldGroheTantau2016} that this holds true even when first-order logic is replaced by monadic second-order logic.}
\end{enumerate}
\end{theorem}
An immediate consequence of the above results is that every node query expressible by a key-oblivious \localmax GNN is also expressible by a \localsum GNN.
The latter result also shows that for key-oblivious \localsum \GNNs, when restricted to first-order logic, arbitrary combination functions are %precisely 
as expressive as \FFN(\ReLU). Beyond first-order logic the situation is different. While \localsumrelu \GNNs can express queries not definable in first-order logic such as ``there are more $p$-neighbors than $q$-neighbors'', they are strictly less expressive than \localsumcontinuous \GNNs. Recall that $\calQ_{\text{even}}$ contains the pointed graph in which the distinguished node has an even number of neighbors. 
\begin{restatable}{proposition}{PropQevenSumvsRelu}
\label{prop:Q_even sum cont vs sum relu}
$\calQ_{\text{even}}$ is expressed by a key-oblivious \localsumcontinuous \GNN, but not by a key-oblivious \localsumrelu \GNN.    
\end{restatable}
Although this inexpressibility result for $\calQ_{\text{even}}$ is in line with general expectations based on recent papers  on zero-one laws for GNNs \cite{AdamDay2023,AdamDay2024}, it does not follow from those results.

% \martin{It may be worth mentioning that "Having more red neighbors than blue neighbors" is a query that is expressible by a \localsumrelu \GNN, but not in first-order logic (and hence not in graded modal logic).}
Theorem~\ref{thm:GMLloweroblivious} shows that \GML constitutes a \emph{lower bound} on the expressive power of key-oblivious \localsumrelu \GNNs.
In a similar way, the logic we will review next constitutes an \emph{upper bound} for the same
class of \GNNs.

\paragraph{First-Order Logic with Counting (FO+C)}
First-order logic with counting extends the standard first-order logic of graphs by the ability to count elements of definable sets and perform basic arithmetic on the counts. We sketch the main parts of the definition and refer the reader to \cite{grohedescriptivegnn} for details. The logic $\FOC$ has formulas and terms. 
\emph{Formulas}
are formed by the usual rules of first-order logic in the language of graphs, using \emph{node variables}, and admitting inequalities $\theta\le\eta$ between terms as additional atomic formulas. \emph{Terms} take values over the natural numbers. They are formed from 
\emph{number variables}, the constants $0,1,\ord$, and \emph{counting terms} of the form
\begin{equation}
  \label{eq:counting-term}
    \#(x_1,\ldots,x_k,y_1<\theta_1,\ldots,y_\ell<\theta_\ell).\psi, 
\end{equation}
where the $x_i$ are node variables, the $y_i$ are number
variables, $\psi$ is a formula, and the $\theta_i$ are
terms. Terms can be combined using addition and multiplication.

To define the semantics of the logic over graphs, we
interpret formulas and terms over the 2-sorted expansion
of a graph $\agraph$ by the standard model $(\mathbb N,+,\cdot,0,1)$ of arithmetic. Node variables range
over $V(\agraph)$ and number variables range over $\mathbb N$. We inductively
define a Boolean value for each formula and a numerical value in $\mathbb N$
for each term. The constant $\ord$ is interpreted by $|V(\agraph)|$. The value of a counting term \eqref{eq:counting-term} is the
number of tuples $(v_1,\ldots,v_k,i_1,\ldots,i_\ell)\in
V(\agraph)^k\times\Nat^\ell$ such that for all $j$, $i_j$ is smaller than
the value of the term $\theta_j$ and $\psi$ holds under the assignment
$x_i\mapsto v_i,y_j\mapsto i_j$. 
Note that the bounding terms $\theta_j$ in the counting construct ensure that the count is a finite number. 
Now the inductive definition of the values of all terms and formulas is straightforward.

% The formulas and number terms of FO+C over a relational signature $\sigma$ are generated by mutual recursion:
% \[
% \begin{array}{ll}
% \phi ::= & R(v_1,\ldots,v_n) \mid v_1=v_2 \mid \theta\leq \theta' \mid \neg\phi \mid \phi\land\phi' \\
% \theta ::= & x \mid 0\mid 1  \mid \theta +\theta' \mid \theta\cdot \theta' \\& \mid \#(v_1,\ldots, v_k, x_1\leq \theta_1,\ldots,x_\ell\leq \theta_\ell).\phi
% \end{array}
% \]
% where $R\in\sigma$ is an $n$-ary relation symbol, $v_i$ are node variables ranging over vertices, $x,x_i$ are number variables ranging over natural numbers, and, in the last clause, $k+\ell\geq 1$.

% The two-variable fragment FO$^2$+C consists of the formulas and number terms that only use two node variables, but unlimited number variables. 

Every FO+C-formula $\phi(u)$ with a single free node variable and no free number variables defines a node query $\calQ_\phi$. 
For example, $\calQ_{\text{even}}$ is defined by:
$$
\phi(u)\coloneqq \exists (x<\ord) . \big(2\cdot x=\#(u').E(u,u')\big)
$$
where $\exists (x<\ord)$ is short for $\#(x<\ord)\cdots >0$.

% the
% %FO$^2$+C-formula
% FO+C-formula
% \[\phi(v) = \Big(\# (x\leq \ord). \exists v'.(E(v,v')\land \# (v).E(v',v)=x)\Big)\geq 3\] expresses that there are at least 3 numbers $n$ for which $v$ has a successor of degree $n$. 
It is not difficult to see that every GML-formula can be expressed in FO+C. 
%FO$^2$+C.
\begin{theorem}[\cite{grohedescriptivegnn}]
    Every node query expressed by a key-oblivious \localsumrelu \GNN is in 
    FO+C.
    % FO$^2$+C.
\end{theorem}
In fact, a stronger result is proved in \cite{grohedescriptivegnn} that applies to GNNs with numerical inputs and outputs, which we will later use to upper bound the expressiveness of key-invariant \localsumrelu \GNNs.
% (Theorem \ref{thm:SumReluinOrdInvFOC}).

%\balder{Above phrasing is a bit misleading, given that this will only happen in the appendix}

% \begin{figure}
% \begin{center}
% \small
% \begin{tikzcd}[column sep=-1mm]
% \textbf{$\policy{\geq0}{<0}$} & \bm{\kern 0.04em \cap} \arrow["\equiv", sloped, phantom, d] & \textbf{$\policy{>0}{\leq0}$}\\  
% &\textbf{$\policy{>0}{<0}$} \arrow[ul] \arrow[ur] 
% % \arrow["\equiv", u, phantom]
% % \arrow["\equiv", u, phantom, rotate=45]
% \\
% & \textbf{$(\geq 1/\leq 0)$} \arrow[u]
% \end{tikzcd}\end{center}
% \caption{Acceptance policies ordered by increasing expressive power
% }
% \label{fig:acceptance_policies_hierarchy}
% \end{figure}

\section{Local Max GNNs}\label{sec:max}

Recall that the node queries expressible by key-oblivious \localmax \GNNs correspond precisely to the formulas of  modal logic, which can count over neighbors only up to $1$, whereas key-oblivious \localsum \GNNs can count up to any fixed natural number. The following example shows that key-invariant \localmax \GNNs have the capability to count up to $2$. 

\begin{example}
\label{ex:diamond2top}
   The node query defined by the GML-formula $\Diamond^{\geq 2}\top$ can be expressed by a
   key-invariant \localmaxrelu \GNN as follows:
   the \GNN first applies one round of message passing to compute, for each node $v$, the minimal and maximal key values among the neighbors of $v$. It then outputs their absolute difference. It follows from Proposition~\ref{prop:bisimulation-localmaxobliv}
   that the same node query cannot be computed by a key-oblivious \localmaxrelu \GNN. 
\end{example}

The above example shows that key-invariant
\localmaxrelu GNNs are more expressive than
key-oblivious \localmax \GNNs with arbitrary combination functions. On the
other hand, the added expressive power is 
limited: as we will see later (Theorem~\ref{thm:diamond geq2 p not expressed by localmaxifpos}) even the node query $\Diamond^{\geq 2}p$ already cannot be expressed by a key-invariant \localmaxrelu GNN,
although it can be expressed if we allow a broader class of combination functions.

In the following subsections, 
we investigate the expressive
power of key-invariant \localmax \GNNs for different classes of combination functions.

\subsection{Continuous Combination Functions}

Example~\ref{ex:diamond2top} shows that
key-invariant \localmaxcontinuous \GNNs can express node queries that are not closed under bisimulation. We find an expressiveness upper bound by further restricting to functional bisimulations (also known in the modal logic literature as bounded morphisms).
\begin{definition}[Functional bisimulations]
Given pointed graphs $\apgraph,\altpgraph$, a functional bisimulation
$B: \apgraph\fbisim \altpgraph$
is a bisimulation $B\subseteq V(\agraph)\times V(\altgraph)$ that is functional, i.e., such that for every $u \in V(\agraph)$ there is exactly one $v \in V(\altgraph)$ such that $B(u,v)$. Likewise for pointed graphs. 
% A node query $\calQ$ is closed under functional bisimulations if $(G_1,v_1) \fbisim (G_2,v_2)$ and $Q(G_1,v_1)=1$ implies $Q(G_2,v_2)=1$.
%$\calQ$ is closed under inverse functional bisimulations if $(G_1,v_1)\fbisim(G_2,v_2)$ and $\calQ(G_2,v_2)=1$ implies $\calQ(G_1,v_1)=1$.
%\balder{We could avoid this definition and instead talk about $Q^c$ being closed under functional bisimulations.}
\end{definition}
% \balder{Alternative equivalent definition: a functional bisimulation is a homomorphism $h:(G,v)\to (H,u)$ with the additional property the $h$ maps the neighbors of a node $w\in V(G)$ surjectively (but not necessarily bijectively) onto the neighbors of $h(w)$. This makes the definition more parallel to that of coverings below, which requires bijectiveness.}
%\begin{theorem} \label{thm:yes-closed closed under fbisim}
%   If $\calQ$ is expressed by a key-invariant \policy{\geq0}{<0}\localmaxcontinuous \GNN, then $\calQ$ is closed under functional bisimulations
%\end{theorem}
\begin{restatable}{theorem}{ThmMaxClosedUnderFbisim}
\label{thm:max yes-open closed under inverse fbisim}
   If a key-invariant %\policy{>0}{\leq0}
   \localmaxcontinuous \GNN expresses $\calQ$, then $\calQ^c$ is closed under functional bisimulations.
\end{restatable}

As a first application of this theorem, 
we obtain that the node
property $\Diamond^{\leq 1}\top$ 
%\balder{We didn't introduce $\Diamond^{\leq}$ only $\Diamond^{\geq}$}
is
not expressible by a key-invariant \localmaxcontinuous \GNN: let 
 $\apgraph$ be a pointed graph consisting
 of a single edge $(u,u')$ and let $\altpgraph$ be a pointed graph consisting of two edges $(v,v'), (v,v'')$. Then $\apgraph\models\Diamond^{\leq 1}\top$, 
$\altpgraph\not\models\Diamond^{\leq 1}\top$,
and 
$\altpgraph\fbisim\apgraph$.

In combination with Example~\ref{ex:diamond2top}, this also shows that the class of queries expressible by \localmaxcontinuous \GNNs is not closed under complementation. This is a consequence of our asymmetric \policy{>0}{\le0} acceptance condition
(see also Section~\ref{sec:policies}).

As a second application, consider any pointed graph $\apgraph$, and let $\altgraph$ be the graph obtained from $\agraph$ by adding an extra node $v$ with the same label and neighbors as $u$. The identity map extended by sending $v$ to $u$ is a functional bisimulation from $\altpgraph$ to $\apgraph$. It follows immediately that \emph{no isomorphism type over connected graphs is definable by a \localmaxcontinuous \GNN}, except for the isomorphism type of a single isolated node:
\begin{restatable}{corollary}{CorIsoCollapseMaxContinuous}
\label{cor:Max iso continuous collapse}
Let $\apgraph$ be a pointed connected graph with isomorphism type $\calQ_{\apgraph}$. Then the following are equivalent:
\begin{enumerate}
    \item $\calQ_{\apgraph}$ is expressed by a key-invariant \localmaxcontinuous \GNN;
    \item $\calQ_{\apgraph}$ is expressed by a key-oblivious \localmaxcontinuous \GNN;
    \item $\apgraph$ is a single isolated node.
\end{enumerate}
\end{restatable}
% \balder{For example, consider any pointed graph $(G,v)$ and let $G'$ be the graph obtained from $G$ by adding an extra node $u$ with the same color and the same neighbors as $v$. The identity map extended by sending $u$ to $v$ is a functional bisimulation from $(G',v)$ to $(G,v)$. it follows immediately that no isomorphism type is definable by a continuous local max GNN, except for the isomorphism type of single-node graphs.}

%\begin{corollary}
%\label{cor:continuous_max_implies_bm_or_ibm_closed}
%    If $\calQ$ is expressed by a key-invariant \localmaxcontinuous \GNN it is either closed under functional bisimulations or inverse functional bisimulations. This holds for all acceptance policies in figure~\ref{fig:acceptance_policies_hierarchy}.
%\end{corollary}

As a third application, 
recall that a $\policy{>0}{<0}$ \GNN node-classifier never outputs $0$. Theorem~\ref{thm:max yes-open closed under inverse fbisim} can be used to show that, in the case of \localmaxcontinuous, this restriction collapses key-invariant classifiers to key-oblivious classifiers:
\begin{restatable}{theorem}{ThmMaxCollapse}
\label{thm:collapseyesnolocalmax}
   Let $\calQ$ be a node query. The following are equivalent:
   \begin{enumerate}
   \item $\calQ$ is expressed by a key-invariant \policy{>0}{<0} \localmaxcontinuous \GNN;
    \item Both $\calQ$ and $\calQ^c$ are expressed by key-invariant 
      % \policy{>0}{\leq 0)}
 \localmaxcontinuous \GNNs;
    \item $\calQ$ is expressed by a key-oblivious \localmaxrelu  \GNN;
    \item $
    \calQ $ is expressible in Modal Logic.
   \end{enumerate}
\end{restatable}

The above results give limitations on the expressive power of key-invariant \localmaxcontinuous \GNNs. In contrast, the next theorem showcases the added expressive power of key-invariant \localmaxrelu and, more generally, \localmaxcontinuous
\GNNs, compared to key-oblivious ones. We introduce two extensions of modal logic:

\begin{definition}[Weakly Graded Modal Logic] \label{def:weaklygradedml}
\begin{align*}
    % &\qquad\ML \,\,:  p \,\mid\, \neg \phi \,\mid\, \phi \vee \psi \,\mid\, \Diamond \phi\\
    &\WGMLtrue: 
    \begin{cases}
    \phi ::= \chi \mid  \phi_1 \vee \phi_2 \mid \phi_1 \wedge \phi_2 \mid \Diamond \phi \mid \Diamond^{\geq 2} \top\\
    \chi ::= \ML
    \end{cases}
    \\
    &\WGMLmodal: 
    \begin{cases}
    \phi ::= \chi \mid  \phi_1 \vee \phi_2\mid \phi_1 \wedge \phi_2 \mid \Diamond \phi \mid \Diamond^{\geq 2} \chi\\
    \chi ::= \ML
    \end{cases}
    % \\
    % &\WGMLall: \ML \,\mid\,
    % \Diamond^{\geq 2} \phi\\
\end{align*}
\end{definition}
where $\apgraph \models \Diamond^{\geq2} \phi$ if and only if $u$ has at least two neighbors that satisfy $\phi$. In other words, $\WGMLtrue$ and $\WGMLmodal$ extend \ML with $\Diamond^{\geq 2} \top$ and $\Diamond^{\geq 2} \chi$, respectively, where $\chi$ is allowed to be a modal formula, and where this counting modality may not occur under a negation.
\begin{restatable}{theorem}{LocalMaxEncodesWGML}
\label{thm:WGML encoded by localmax}
\leavevmode\par
\begin{itemize}
    \item[a]
    For every $\phi \in \WGMLtrue$, $\calQ_\phi$ is expressed by a key-invariant \localmaxrelu \GNN.
    
    \item[b]
    For every $\phi \in \WGMLmodal$, $\calQ_\phi$ is expressed by a key-invariant \localmaxsigmoid \GNN.
\end{itemize}
\end{restatable}
The first result is optimal in the sense that $\calQ_{\Diamond^{\geq2}p}$ is not expressible by a \localmaxrelu \GNN: this will follow from Theorem \ref{thm:diamond geq2 p not expressed by localmaxifpos}. The proof of the second result also implies that \localmaxsigmoid \GNNs can use keys to count up to $2$ over nodes at a fixed distance $k$ for any $k \geq 1$. Key-invariant \localmaxsigmoid \GNNs thus express queries, such as ``there are at least $2$ nodes reachable in two steps'', that are not CR-invariant and hence cannot be expressed by any key-oblivious \localmax or \localsum \GNN. \footnote{Note that this query distinguishes the a cycle of length 4 from a cycle of length 2.}
% The second result shows that key-invariant \localmaxsigmoid \GNNs express queries that are not expressed by any key-oblivious \localsum \GNN. For example, $\calQ_{\Diamond^{\geq2}_2 p}$ is not CR-invariant.
\subsection{Arbitrary Combination Functions}

We now examine what happens when \localmax \GNNs use discontinuous combination functions. \localmaxsemilinear \GNNs can express step functions as well as all connected isomorphism types. We further upper bound the expressiveness of \localmax \GNNs with arbitrary combination functions by order-invariant FO.

\subsubsection{\uddl and \localmaxsemilinear \GNNs}
\label{sec:uddl}
We introduce a new modal-style logic, \uddl; We will prove that every 
\uddl formula is expressible by a key-invariant \localmaxsemilinear \GNN.

\uddl (Local-Determinism Dynamic Logic)
is a  variant of propositional dynamic logic (PDL) interpreted over graphs.
There are two kinds of terms. \emph{Formulas} define node properties, and \emph{Program expressions} define binary relations, interpreted as ways to transition from one node to another in the graph.

\begin{tabular}{l}
    Formulas: \\ $\phi ::= p \mid \top \mid \phi_1\land \phi_2 \mid \neg\phi \mid \langle \pi\rangle\phi \mid  [\pi]\phi \mid \langle\pi\rangle^{=1}\phi$ \\ \\
    Program expressions: \\
$\pi ::= \step \mid \test(\phi) \mid \pi_1;\pi_2 \mid \pi_1\cup \pi_2$ 
\end{tabular}

Here $p$ is a proposition in in $\Pi$. The semantics are given in Table~\ref{tab:uddl}.
We write $\apgraph\models\phi$ if $u\in \sem{\phi}_{\agraph}$ and we also write $\agraph,(u,v)\models\pi$ if $(u,v)\in \sem{\pi}_{G}$. 
We use $\stay$ as shorthand for $\test(\top)$. Each \uddl formula defines a node query--the set of pointed graphs on which it is satisfied.

\begin{table}
\caption{Semantics of \uddl, where $\agraph=(V,E,\lab)$.}
\[\begin{array}{lll}
\sem{p_i}_{G} &=& \{u\in V\mid \lab(u)_i=1\}\\
\sem{\top}_{G} &=& V\\
\sem{\phi_1\land\phi_2}_{G} &=& \sem{\phi_1}_{G}\cap\sem{\phi_2}_{G}\\
\sem{\neg\phi}_{G} &=& 
V\setminus\sem{\phi}_{G} \\ \
\sem{\langle \pi\rangle\phi}_{G} &=& \{u\in V\mid\exists v\in\sem{\phi}_{G} \text{ s.t. }  (u,v)\in\sem{\pi}_{G}\}\\
\sem{[\pi]\phi}_{G} &=& 
\{u\in V\mid\forall v,(u,v)\in\sem{\pi}_{G}\Rightarrow v\in\sem{\phi}_{G}\}\\
\sem{\langle \pi\rangle^{=1}\phi}_{G} &=&
\{u\in V\mid \exists^{=1} v\in \sem{\phi}_{G} \text{ s.t. } (u,v)\in \sem{\pi}_{G}\}\\ \\
\sem{\step}_{G} &=& E \\
\sem{\test(\phi)}_{G} &=& \{(u,u)\mid u\in \sem{\phi}_{G}\} \\
\sem{\pi_1;\pi_2}_{G} &=& \sem{\pi_1}_{G} \circ \sem{\pi_2}_{G} \\
\sem{\pi_1\cup\pi_2}_{G} &=& \sem{\pi_1}_{G}\cup \sem{\pi_2}_{G}
\end{array}\]
\label{tab:uddl}
\end{table}

In other words, $\uddl$ is star-free PDL interpreted over graphs, extended with a restricted counting construct that can be applied to a complex program, and that allows only to distinguish zero, one and more-than-one.
$\uddl$ strictly subsumes weakly graded modal logic, and contains formulas outside of graded modal logic: 

\begin{example} \label{ex:uddl}
The \uddl formula $\langle step;step\rangle^{=1}p$ 
expresses that 
there is exactly one $p$-labeled node reachable in two steps from the distinguished node.
% The formula $\phi_1=\langle\step\rangle^{=1} \greencolor$  checks that
% there is a unique Green neighbor. The formula $\phi_2=\langle\step;\test(\redcolor);\step\rangle^{=1} \greencolor$  checks that
% there is a unique Green node reachable in two steps by going through a Red node. Let $\phi_3$ be defined as
% \begin{align*}
%  \langle \step \cup (\step;\test(\redcolor);\step)\rangle^{=1} \greencolor
% \end{align*}
% $\phi_3$ expresses that there is a unique Green node that is reachable either via a single step or two-steps via a Red node.
% Thus the conjunction of $\phi_1$, $\phi_2$, and $\phi_3$ implies  that the distinguished node  participates in a triangle, with one of the other vertices being Red the other Green, the 
% Green node being the sole one reachable via a red-green walk.
\end{example}

Our first main result regarding key-invariant \localmaxsemilinear \GNNs is that they express all queries in \uddl.

\begin{restatable}{theorem}{thmuddllower} \label{thm:uddllower} 
Let $\phi$ be a $\uddl$ formula. Then $\calQ_\phi$ is expressed by a key-invariant \localmaxsemilinear \GNN.
\end{restatable}

The above encoding of $\uddl$ requires the full power of semilinear combination functions: 
$\FFNs$ with unary semilinear activation functions (equivalently, \FFN(\relu,\Heavi)) do not suffice:
\begin{restatable}{theorem}{thmMaxContUpperBound}
\label{thm:diamond geq2 p not expressed by localmaxifpos}
    There does not exist a key-invariant \localmaxunarysemilinear \GNN that expresses $Q_{\Diamond^{\geq 2} p}$. 
\end{restatable}
% \balder{Could phrase this in terms of \FFN(\unarysemilinear).}
% Hence there also does not exist a \localmaxrelu \GNN that expresses this query.
Similarly, it follows from Theorem~\ref{thm:max yes-open closed under inverse fbisim} that \emph{continuous} combination functions are not enough to express \uddl, since $\calQ^c_{\Diamond^{=1}p}$ is not closed under functional bisimulation. 

The following result attests to the considerable expressive power of $\uddl$, compared to, say, GML:

\begin{restatable}{theorem}{thmundecidable}\label{thm:undecidableuddl} Satisfiability is undecidable for $\uddl$. %, thus the satisfiability problem for key-invariant \localmaxsemilinear \GNNs is undecidable.
\end{restatable}
% \michael{I think the "thus..." isn't well-posed, since we don't have a syntax for key-invariant localmax semilinear GNNs. When I first proposed this, I had a corollary that satisfiability of general \localmaxsemilinear GNNs is undecidable -- but this is out of scope. We could just leave it here without
% the "thus..."?}
The idea of the proof is that one can enforce a tiling structure, with distinct edge and node coloring patterns corresponding to vertical
and horizontal axes. Using the ``unique witness modality'' $\langle \pi\rangle^{=1}\phi$ we can enforce grid-likeness of the graph structure,
for example that moving horizontally and then vertically gets you to the same node as moving vertically and then horizontally.
We will make further use of this in Section~\ref{sec:testing} to show undecidability of testing key-invariance for \localmaxsemilinear \GNNs.

\subsubsection{Isomorphism Types}

Besides subsuming \uddl node queries, 
key-invariant \localmaxsemilinear \GNNs are expressive enough to express all isomorphism types of connected pointed graphs.

\begin{restatable}{theorem}{ThmMaxIsoTypes}
\label{thm:localmaxsemilinear recognizes up to isomorphism}
    Let $\apgraph$ be a connected pointed graph. The isomorphism type $\calQ_{\apgraph}$ is expressed by a key-invariant \localmaxsemilinear \GNN.
\end{restatable}

This result is still uniform in the sense that a single \GNN separates a pointed graph from all other graphs. The proof proceeds by showing that, for every natural number $N$, there is a key-invariant \localmaxsemilinear \GNN that accepts a pointed graph precisely if the 
number of nodes reachable from the distinguished 
node is exactly $N$, and that computes the list of $N$ keys in descending order.

\subsubsection{An expressive upper bound: order-invariant FO}
\label{para:order-inv}

We now show that the expressive power of 
key-invariant \localmax \GNNs with arbitrary combination functions is upper-bounded by that of order-invariant FO.

% \begin{proposition} \label{prop:orderinvariant} For any key-invariant local-max GNN $\agnn$, regardless of the combination functions,
% there is an order-invariant first-order formula $\phi$ in the language of graphs with an additional linear order on the nodes
% such that $\agnn$ and $\phi$ define the same query (over keyed graphs).
% \end{proposition}
% \begin{definition}[Order-Invariant FO] \label{def:orderinvfo}

 Let an FO($<$)-formula be any first-order formula built up from atomic formulas of the form
 $E(x,y)$, $x<y$ or $p(x)$ with $p\in\Pi$, using the connectives of first-order logic (i.e., Boolean operators and
existential and universal quantification). 
Such formulas are naturally interpreted over
\emph{ordered graphs} $(\agraph,<)$, i.e., graphs expanded with a linear order.
A formula $\phi(x)$ in this logic is \emph{order-invariant} if, for all
pointed graphs $\apgraph$ and linear orders $<,<'$
on $V(\agraph)$,
$(\agraph,<)\models\phi(u)$ iff $(\agraph,<')\models\phi(u)$.
Each order-invariant FO($<$)-formula $\phi(x)$  naturally
defines a node query (namely, the 
set of pointed graphs that, when expanded with an arbitrary linear order, satisfy $\phi$).
We let $\ordinvfo$ be the collection of node queries defined by order-invariant FO($<$)-formulas in this way.
\ordinvfo is strictly more expressive than plain first-order logic, as was first shown by Gurevich in the 1980s
\cite{Gurevich84}, cf.~also~\cite{Schweikardt13:invariant}.
% \martin{I believe the reference where this appeared first is \cite{Gurevich84}.}
% \michael{If you are not sure I think it is safer, and maybe more helpful to readers, to point to Nicole's article.
% There she writes: ``These examples go back to Gurevich (who
% did not publish this example; but it can be found in the textbooks [1,14]), Potthoff [18],
% and Otto [17].}

\begin{restatable}{theorem}{thmorderinvariantfo} \label{thm:orderinvariant} Let $\calQ$ be a node query expressed by a key-invariant \localmax \GNN. Then $\calQ$ is in \ordinvfo.
\end{restatable}

Note that this result applies to \localmax \GNNs with arbitrary, not necessarily semilinear, combination functions.

For example, recall that 
$\calQ_{\text{even}}$ is the node query that tests if the number of
neighbors of the distinguished node is even. It follows from known results in finite model theory that this node query is 
not $\ordinvfo$-definable. Hence, it is 
also not expressible by any key-invariant \localmax GNN.

The proof of Theorem~\ref{thm:orderinvariant}
uses techniques from embedded finite model theory. By identifying each node with its key, 
we can view a keyed graph as an embedded finite structure. More precisely, fix a set $F$ of 
real-valued functions. A keyed graph $G_\key$ can be represented by an infinite structure
$(\reals,<, (f)_{f\in F}, V,E,p_1,\ldots, p_n)$, where $V$ is a finite subset of $\mathbb{R}$, $E\subseteq V^2$, and $p_i\subseteq V$. 
Structures of this form are known as embedded finite structures, since they consist of a 
finite structure embedded within some fixed infinite structure (in this case, $(\reals,<, (f)_{f\in F})$).
Every \localmax \GNN using combination functions from $F$ can be inductively translated to
a first-order formula interpreted over embedded finite structures of the above kind, and, in fact, to an \emph{active-domain} formula, i.e., one in which all quantifiers are restricted to range over the vertex set $V$. Key-invariance of a \GNN amounts to \emph{genericity} of the corresponding active-domain first-order formula in the usual sense from database theory \cite{AHV}: it expresses a query whose output depends only on the isomorphism type of the embedded graph, not on how it is embedded in the reals. The \emph{locally-generic collapse} theorem from \cite{collapsejacm} implies that every generic active-domain first-order formula can be rewritten to an equivalent \ordinvfo formula.
%\balder{Thanks! Rewrote it a bit (in particular, to avoid talking about feature expressions here)}

\section{Local Sum GNNs}\label{sec:sum}

We give an example showing that key-invariant \localsumsigmoid \GNNs are more expressive than key-oblivious \localsum \GNNs with arbitrary combination functions:
% \localsum \GNNs can use keys to distinguish nodes that are reachable via unique paths:
\begin{example}
\label{ex:localsum}
    Let $G_{\triangle p}^u$ be a pointed 3-cycle in which precisely one node $u'$ is labeled $p$,  where $u\neq u'$. We can construct a key-invariant
    \localsumsigmoid \GNN node-classifier that 
    accepts a connected pointed graph precisely if it is \emph{not} isomorphic to  $G_{\triangle p}^u$. To see this, let $\phi$ be the GML-formula
    \[\Diamond^{=2}\top\land\Diamond(\neg p\land\Diamond^{=2}\top\land\Diamond p\land\Diamond\neg p)\land\Diamond(p\land\Diamond^{=2}\top\land\Box\neg p)\]
    A connected pointed graph is isomorphic to $G_{\triangle p}^u$ iff 
    it satisfies $\phi$, and, furthermore,
    the unique $p$-neighbor of $u$ is equal to the unique 
    $p$-node that can be reached from $u$ in two steps. 

    The \GNN proceeds as follows. Let $i$ be the index of $p$ in the labeling. Each node sends a normalized key in the range $(0,1)$ if it satisfies $p$, and $0$ otherwise, using $\relu(\sigmoid(\key(v)) - \lab(v)_i)$. This is the only place where $\sigmoid$ is used. 
    Then, the \GNN 
    computes $x+|y-z|$ where
    \begin{itemize}
    \item $x$ is 0 if $u$ satisfies $\phi$
        and 1 otherwise (using~Theorem~\ref{thm:GMLloweroblivious}),
    \item 
        $y=\sum_{\substack{(u,v)\in E\\ \lab(v)_i=1}}\sigmoid(\key(v))$, 
        
        which can be computed with one round of \ \localsum aggregation,
        %, by subtracting $1$ from the key if $p$ is false and then applying $\relu$).
    \item $z=\sum_{\substack{(u,w)\in E, (w,v)\in E\\ \lab(v)_i=1}}\sigmoid(\key(v))$, 

    which
    can be computed with two rounds of \ \localsum aggregation.
    \end{itemize}
    %It is now easy to see that $x+|y-z|=0$ if and only if the input pointed graph is isomorphic to $G_{\triangle p}^u$.

    See also Appendix~\ref{appendix:prelims} for a  more detailed specification of the \GNN.
\end{example}
To provide some more intuition for this example, 
let us call a node in a pointed graph \emph{uniquely addressable} if one can specify a walk consisting of \GML formulas that deterministically navigates from the distinguished node to this node (see Definition \ref{def:uniqely addressable}). Key-invariant \localsumsigmoid \GNNs can test if two uniquely addressable nodes are non-identical. In the example above, this capacity is used to separate two uniquely addressable nodes that satisfy $p$. We will make this more precise in Proposition \ref{prop:sum continuous unique paths}. 
% The above example shows that key-invariant \localsumsigmoid \GNNs are more expressive than key-oblivious \localsum \GNNs with arbitrary combination functions. 

As in the case for \localmax, we will see that this added expressive power is limited. For example, key-invariant \localsumcontinuous \GNNs express exactly the same isomorphism types as key-oblivious \localsumcontinuous \GNNs (Theorem \ref{thm:sum-isomorphism-collapse}). In the remainder of this section, we investigate the expressive power of \localsum \GNNs, again for continuous as well as discontinuous combination functions.

\subsection{Continuous Combination Functions}

Theorem \ref{thm:max yes-open closed under inverse fbisim} showed that queries expressed by \localmaxcontinuous \GNNs are closed under a strengthening of bisimulation. We now upper bound the expressiveness of \localsumcontinuous via a strengthening of CR-equivalence. 
Subsequently we show that key-invariant \localsumcontinuous \GNNs that never output $0$ collapse to key-oblivious \localsumcontinuous \GNNs, analogous to the case for \localmaxcontinuous \GNNs
(Theorem \ref{thm:collapseyesnolocalmax}).
% \subsubsection{Upper bounds}

% In the \localmax case, 
% Theorem \ref{thm:max yes-open closed under inverse fbisim} showed that queries expressed by \localmaxcontinuous \GNNs are closed under a strengthening of bisimulation. We find that queries expressed by \localsumcontinuous \GNNs are closed under a strengthening of cr-invariance:

\begin{definition}[Coverings]
Given pointed graphs $\apgraph, \altpgraph$, a homomorphism $h:\apgraph\to \altpgraph$ is a map from $V(\agraph)$ to $V(\altgraph)$ satisfying 
\begin{enumerate}
    \item $(u_0,u_1)\in E(\agraph)\Rightarrow (h(u_0),h(u_1))\in E(\altgraph)$;
    \item
$\lab_{\altgraph}(h(u_0))=\lab_{\agraph}(u_0)$;
\item $h(u)=v$.
\end{enumerate}
A \emph{covering} $h:\apgraph\covers \altpgraph$ is a homomorphism with the additional property that, for each $u_0 \in V(\agraph)$, 
$h$ bijectively maps the neighbors of $u_0$ onto the neighbors of
$h(u_0)$.
% A node query $\calQ$ is closed under coverings if $(G_1,v_1) \covers (G_2,v_2)$ and $Q(G_1,v_1)=1$ implies $Q(G_2,v_2)=1$.
%$\calQ$ is closed under inverse coverings if $(G_1,v_1) \covers (G_2,v_2)$ and $\calQ(G_2,v_2)=1$ implies $\calQ(G_1,v_1)=1$.
\end{definition}

It is not difficult to see that, when $\apgraph\covers\altpgraph$,  $CR(\apgraph)=CR(\altpgraph)$. Thus
the covering relation can be viewed
as a strengthening of CR-equivalence.
Coverings are a well studied notion in graph theory (cf.~e.g.,~\cite{kratochvil2025graphcoverssemicoversstronger}). They can be equivalently described as 
functional graded bisimulations, i.e., functions whose graph is a graded bisimulation.

% \begin{theorem}
% \label{thmr:continuous_sum_yes-closed_implies_cov_closed}
% Let $\calQ$ be a node query decided by an invariant \policy{\geq0}{<0} \localsumcontinuous\ \GNN. Then $\calQ$ is closed under coverings.
% \end{theorem}

\begin{restatable}{theorem}{ThmSumClosedUnderCoverings}
\label{thm:closed-under-coverings}
If a key-invariant %\policy{>0}{\leq0} 
\localsumcontinuous\ \GNN expresses $\calQ$, then $\calQ^c$ is closed under coverings.
\end{restatable}
As a first application, consider the
pointed graphs $C^u_3$ and $C^v_6$ where
$C_n$ is the cycle of length $n$. 
Since $C^v_6\covers C^u_3$, every 
key-invariant \localsumcontinuous \GNN
that accepts $C^u_3$ must also accept
$C^v_6$. Thus, for example, the 
query ``the distinguished node lies on a 3-cycle'' is not expressible by
a key-invariant \localsumcontinuous \GNN.
This also shows that  the isomorphism
type of $C^u_3$ is not expressible by
a key-invariant \localsumcontinuous \GNN.

The same argument applies to every connected pointed graph with a cycle, so that their isomorphism type   is not expressible with a key-invariant \localsumcontinuous \GNN. On the other hand, the isomorphism type of every connected \emph{acyclic} graph is expressible by key-oblivious \localsumcontinuous \GNNs. Thus, as a second application, we obtain the following result:

\begin{restatable}{theorem}{thmSumIsomorphismCollapse}
\label{thm:sum-isomorphism-collapse}
   Let $\apgraph$ be a connected pointed graph with isomorphism type $\calQ_{\apgraph}$. The following are equivalent:
   \begin{enumerate}
       \item $\calQ_{\apgraph}$ is expressed by a key-invariant \localsumcontinuous \GNN;
       \item $\calQ_{\apgraph}$ is expressed by a key-oblivious \localsumcontinuous \GNN;
       \item $\apgraph$ is acyclic.
   \end{enumerate}
\end{restatable}
Note however that there do exist \emph{complements} of isomorphism types that are expressed by a key-invariant, but not a key-oblivious \localsumcontinuous \GNN, as is shown by Example \ref{ex:localsum}.

As a third application, 
recall that a $\policy{>0}{<0}$ \GNN node-classifier never outputs $0$. 
Theorem~\ref{thm:closed-under-coverings}
can be used to show that the expressive power of  key-invariant
$\policy{>0}{<0}$ \localsumcontinuous \GNNs
collapses to that of key-oblivious \localsumcontinuous \GNNs.

\begin{restatable}{theorem}{thmgreaterthanzeroslashlessthanzerolocalsumcontinuouscollapsestokeyoblivious}
\label{thm:>0/<0 localsumcontinuous collapses to key-oblivious}
    Let $\calQ$ be a node query. The following are equivalent:
    \begin{enumerate}
        \item $\calQ$ is expressed by a key-invariant \policy{>0}{<0} 
        \localsumcontinuous\ \GNN;
        \item $\calQ$ and $\calQ^c$ are both expressed by a  key-invariant  
        \localsumcontinuous\ \GNN;
        %$\calQ$ is expressed by a key-invariant \policy{\geq0}{<0} \localsumcontinuous\ \GNN and by a key-invariant \policy{>0}{\leq 0} \localsumcontinuous\ \GNN.        
        \item $\calQ$ is expressed by a key-oblivious \localsumcontinuous\ \GNN;
        
        \item $\calQ$ is strongly local (Def. \ref{def:local queries}) and CR-invariant (Def. \ref{def:CR invariant queries}).
    \end{enumerate}
   % Here, $\calQ$ is local if for some $r>0$ pointed graphs with isomorphic radius $r$ induced subgraphs are not separated by $\calQ$. $\calQ$ is CR-invariant if pointed graphs inseparable by CR are not separated by $\calQ$.
\end{restatable}
In the case of \policy{\geq 1}{\leq 0}
we get an even stronger form of collapse:

\begin{restatable}{theorem}{thmlocalsumgapcollapse} \label{thm:localsumgapcollapse}
Let $\calQ$ be a node query expressed by a key-invariant \policy{\geq1}{\leq0} \localsumcontinuous\ \GNN. Then $\calQ$  is expressed by a key-oblivious \localsumcontinuous\ \GNN with the same combination functions.
\end{restatable}
Indeed, the proof shows that, in the case of
key-invariant \policy{\geq1}{\leq0} \localsumcontinuous\ \GNNs, all keys can be
replaced by a fixed value (say, $0$) without 
affecting the predictions of the \GNN.

The above results give limitations in expressive power for
key-invariant \localsumcontinuous \GNNs.
 Next, we will showcase a concrete class of queries that \emph{are} expressible by key-invariant \localsumcontinuous \GNNs, expanding on the intuition given in Example \ref{ex:localsum}.

\begin{definition}[Uniquely addressable]
\label{def:uniqely addressable}
Let $n \geq 1$, let $\phi_1, \dots \phi_n$ be \GML formulas and let $\apgraph$ be a pointed graph. A node $u' \in V(\agraph)$ is \emph{uniquely addressable} by $\phi_1,\dots,\phi_n$ if there is a single walk $u,u_1, \dots u_{n}$ in $\agraph$ where each $u_i$ satisfies $\phi_i$, and if $u'=u_n$.     
\end{definition}

We will denote by $\calQ_{\langle\phi_1,\dots,\phi_n\rangle\neq\langle\psi_1,\dots,\psi_m\rangle}$ the node query that consists of all pointed graphs $\apgraph$ containing nodes $u_\phi,u_\psi \in V(\agraph)$ uniquely addressable by $\phi_1,\dots,\phi_n$ and $\psi_1,\dots,\psi_m$ respectively such that $u_\phi \neq u_\psi$. We call such a query \emph{unique address separating}.
% \arie{The iso-types are a bit tricky: for every graph with unique paths to nodes $u_1, \dots u_k$ 
% where for all CR-equivalent graphs there are no $v_1, \dots v_k$ such that all the `unique addressing' paths to $u_i$ are also `unique addressing' paths to $v_i$, we can express the co-isomorphism type with \localsumcontinuous.
% For example, if all nodes are uniquely addressable in a single way, every CR-equivalent graph has the same uniquely addressing paths (since we can enforce this with GML). Then also every CR-equivalent graph has $v_1, \dots v_k$ such that all unique paths to $u_i$ are also to $v_i$, hence
% \localsumcontinuous can't express any co-iso types beyond key-oblivious. With \localsumsemilinear we can do more.}
%%%%%%%%%%%%%
\begin{restatable}{proposition}{LocalSumContUniquePaths}
\label{prop:sum continuous unique paths}
Every unique address separating query $\calQ_{\langle\phi_1,\dots,\phi_n\rangle\neq\langle\psi_1,\dots,\psi_m\rangle}$ specified by \GML-formulas is expressible by a \localsumsigmoid \GNN.
\end{restatable}
These queries are not always
expressible by a key-oblivious \localsum \GNNs. As an example consider $\calQ_{\langle p \rangle \neq \langle \top, p \rangle}$. Now, if $\agraph$ is a $3$-cycle with a single $p$ node, $\altgraph$ is a $6$-cycle with two $p$ nodes at distance $3$ from each other, $u$ is a non-$p$ node in $\agraph$ and $v$ is a non-$p$ node in $\altgraph$, then $\apgraph \not\in \calQ_{\langle p \rangle \neq \langle \top, p \rangle}, \altpgraph \in \calQ_{\langle p \rangle \neq \langle \top, p \rangle}$ while $CR(\apgraph)=CR(\altpgraph)$. Note that the \GNN in Example \ref{ex:localsum} expresses $\calQ_{\neg\phi} \cup \calQ_{\langle p \rangle \neq \langle \top, p \rangle}$, where $\phi$ is a \GML formula.

\subsubsection{\localsumrelu \GNNs}

Above, we showed that keys increase the expressive power of \localsum \GNNs with \FFN(\relu,\sigmoid) combination functions, and with continuous combination functions in general. It remains open whether the same is true for \localsumrelu \GNNs. In Section \ref{sec:restricted} we argue that in the restricted setting where keys are from a bounded real interval, key-invariant \localsumrelu \GNNs do express queries that are not expressed by any key-oblivious \localsum \GNN.

One can easily separate key-invariant \localsum \GNNs with \FFN(\relu) combination functions from \localsumcontinuous \GNNs. There are only countably many classifiers in the former class, since we require the parameters to be rationals, whereas the latter class contains uncountably many queries.
Theorem~\ref{thm:localsumgapcollapse} yields a concrete separating example for \policy{\geq1}{\leq0} \localsumrelu\ \GNNs: it implies 
 (together with Proposition~\ref{prop:Q_even sum cont vs sum relu})
that $\calQ_{\text{even}}$ is not expressible by a a key-invariant \policy{\geq1}{\leq0} \localsumrelu\ \GNN, whereas
it is expressed by a key-invariant \policy{\geq1}{\leq0} \localsumcontinuous\ \GNN (in fact, by a key-oblivious one), by Theorem~\ref{thm:localsumcontinuous expresses all local CR-invariant queries}. 

We now provide a logical upper bound for the expressiveness of key-invariant \localsumrelu \GNNs. Order-invariant $\FOC$ (\ordinvfoc) is defined in the same way as order-invariant
$\FO$ (see Section~\ref{para:order-inv}). 
It was shown in \cite{Barrington1990} that \ordinvfoc captures the circuit complexity class TC$^0$.

\begin{restatable}{theorem}{ThmSumReluinOrdInvFOC}
\label{thm:SumReluinOrdInvFOC}
  Let $\calQ$ be a node query expressed by a key-invariant \localsumrelu \GNN. Then $\calQ$ is in \ordinvfoc.
\end{restatable}

The theorem essentially follows from \cite[Theorem~5.1, Remark~5.3]{grohedescriptivegnn}.

\subsection{Arbitrary Combination Functions}

\subsubsection{Expressive completeness}

If we allow arbitrary (not-necessarily continuous) combination functions, 
key-invariant \localsum \GNNs become expressively complete:

\begin{restatable}{theorem}{thmlocalsumarbunlimited} Every strongly local node query is expressed by a key-invariant \localsum \GNN.
\label{thm:localsumarbunlimited}
\end{restatable}

The proof uses a 
moment-injective function $f:\reals\to\reals$, i.e., a function
such that the map $\widehat{f}(\multiset{x_1, \ldots, x_m})=\sum_{i=1\ldots m} f(x_i)$ from 
finite multisets of reals to reals is injective. By applying such a function $f$ to the embedding vector of each node before every round of \localsum aggregation, one obtains a maximally-distinguishing \GNN (cf.~\cite{amir2023neural}).

We can further use this to show that 
there are queries expressible by a key-invariant \localsum \GNN but not with either continuous or semilinear combination functions:
\begin{restatable}{proposition}{propSumArbitraryBeyondContandSemilin} \label{prop:sumarbitrarybeyondcontandsemilin}
There exist node queries that are expressed by a key-invariant \localsum \GNN, but not by a key-invariant \localsumcontinuous or key-invariant \localsumsemilinear \GNN.  
\end{restatable}
This result follows from the observations that \localsumcontinuous \GNNs are limited in expressive power by Theorem \ref{thm:closed-under-coverings} and that there are only countably many queries expressed by \localsumsemilinear \GNNs, while neither of these restrictions applies to \localsum \GNNs with arbitrary combination functions.

\subsubsection{\localsumsemilinear \GNNs and isomorphism types}

Recall that unique address separating queries $\calQ_{\langle\phi_1,\dots,\phi_n\rangle\neq\langle\psi_1,\dots,\psi_m\rangle}$ are expressible by key-invariant \localsumsigmoid \GNNs (Proposition \ref{prop:sum continuous unique paths}). The same can be achieved with semilinear combination functions:
\begin{restatable}{proposition}{LocalSumSemiUniquePaths}
\label{prop:sum semilinear unique paths}
  Every unique address separating query $\calQ_{\langle\phi_1,\dots,\phi_n\rangle\neq\langle\psi_1,\dots,\psi_m\rangle}$ specified by \GML-formulas is expressible by a \localsumsemilinear \GNN.
\end{restatable}

This shows that key-invariant \localsumsemilinear \GNNs express queries beyond the reach of any key-oblivious \localsum \GNN. The following result in combination with Theorem~\ref{thm:sum-isomorphism-collapse} further separates key-invariant \localsum \GNNs with semilinear and continuous combination functions.
\begin{restatable}{proposition}{propSumSemilinExpressesTriangleIso}
\label{prop:SumSemilinExpressesTriangleCoIso}
    Let $\apgraph_{\triangle p}$ be the pointed $3$-cycle with a single $p$ node from example \ref{ex:localsum}. The isomorphism type $\calQ_{\apgraph_{\triangle p}}$ is expressed by a \localsumsemilinear \GNN.
\end{restatable}
In addition, every isomorphism type of connected pointed graphs is expressible using semilinear combination functions \emph{with squaring}:
\begin{restatable}{theorem}{ThmSumSemilinearSquareIsotypes}
\label{thm:iso test with local sum}
    Let $\apgraph$ be a connected pointed graph. Its isomorphism type $\calQ_{\apgraph}$ is expressed by a key-invariant \localsum \GNN with combination functions composed of semilinear functions and the squaring operation $(\cdot)^2: \reals \to \reals$.
\end{restatable}
It is worth contrasting this result with 
the one for \localmax GNNs (Theorem~\ref{thm:localmaxsemilinear recognizes up to isomorphism}). There we were able to express isomorphism types using semilinear combination functions. In the \localsum case, the proof is substantially more complex. One of the ingredients in the proof of Theorem~\ref{thm:iso test with local sum} is the Cauchy-Schwarz inequality, which tells us 
that the sum and sum-of-squares of a multiset of reals together form an encoding that allows one to detect whether the multiset contains at most one nonzero value. 

\section{Further Aspects}
\label{sec:discussion}

\subsection{Other Acceptance Policies}
\label{sec:policies}

Our definition of \GNN node-classifiers included an acceptance policy 
\policy{>0}{\leq 0}, where the other acceptance policies added further restrictions. 
One could instead consider a $\policy{\geq 0}{<0}$
acceptance policy. The latter could be called \emph{yes-closed} \GNN node-classifiers (as opposed to \emph{yes-open}). A yes-closed \GNN node classifier for $\calQ$ can be 
equivalently viewed as a yes-open GNN node classifier for $\calQ^c$, simply by modifying the last layer to turn the final value $x$ into its negative $-x$.  
Thus, all results carry over to the yes-closed setting when one replaces $\calQ$ by $\calQ^c$. 

Key-invariant yes-open \localsumcontinuous \GNNs express the same isomorphism types as key-oblivious \localsumcontinuous \GNNs (Theorem \ref{thm:sum-isomorphism-collapse}).  Example~\ref{ex:localsum} shows that the same is not true with yes-closed acceptance, since by complementation it provides
a non-trivial example of an isomorphism type that is expressed by a key-invariant yes-closed \localsumsigmoid \GNN. We also provide a negative example of this capacity to express isomorphism types. Let $n \geq 3$ and let $C_n^u$
be the cycle of length $n$
(with uniform labeling, and where the distinguished node $u$ is any node on the cycle). 
Then $C_{4n}^u\covers C_{2n}^u\covers C_n^u$. It follows
by Theorem~\ref{thm:closed-under-coverings} that the isomorphism type of $C_{2n}^u$ is not expressed by a \localsumcontinuous GNN, regardless of whether one uses a yes-open or yes-closed acceptance policy.
 
 \subsection{Composite Keys}
 \label{sec:composite}

In the introduction, we mentioned that in geometric graph learning the input graph often comes equipped with uniquely identifying node features, such as Euclidean coordinates. In this setting,  keys may be composite, consisting of more than one node feature (e.g., the $x$- and $y$-coordinates). It is easy to see that all our \emph{upper-bound} results on the expressive power of key-invariant \GNNs remain valid in this setting. Indeed, if a query is expressed by a key-invariant \GNN $\agnn$ given composite keys, it can also be expressed with non-composite keys by first translating each key $k$ to the composite key $\langle k, \ldots,k\rangle$ and then running $\agnn$. 
Most of our \emph{lower-bound} results on expressive power of key-invariant \GNNs extend to the composite-key setting as well, although this requires careful inspection of the proofs.

\subsection{Invariance for Specific Symmetry Groups}
\label{sec:groups}

Let $\group$ be a group of permutations on the  reals. $\group$ induces an equivalence relation $\equiv_\group$ on pointed  keyed graphs where two pointed  keyed graphs $\apgraph$ and $\altpgraph$ are equivalent if there exists an isomorphism $f$ between the  underlying pointed  graphs, and a $g \in \group$ so that for each $w \in V(G): g(\val(w)) = \val(f(w))$.
We say that a node query is \emph{$\group$-invariant} if it returns the same value on keyed graph inputs that are equivalent under $\equiv_\group$.
Key-invariance can be understood equivalently as 
$\group_{\text{full}}$-invariance, where $\group_{\text{full}}$ is the  full symmetry group (i.e., the group of all bijections).
However, other symmetry groups are also natural to consider, such as the symmetry group $\group_<$ of all order-preserving bijections, or the symmetry group $\group_{\text{cont}}$ of all continuous bijections.  
Note that $\group' \subset \group$ implies that every $\group$-invariant query is also $\group'$-invariant. In particular, key-invariant queries are $\group$-invariant for \emph{every} symmetry group $\group$. 

In the context of geometric graph learning, there has been considerable effort on the design of \GNN that are invariant, or  equivariant, for different symmetry groups~\cite{han2025survey},
but we are not aware of any work that specifically investigates the leverage such \GNNs can get from having unique node identifiers.

We present one result along these lines, as it follows immediately from one of our proofs. Let $\group_<$ be the symmetry group that consists of all order-preserving bijections on $\reals$. 
Each $\group_<$-invariant \GNN naturally defines a node query over ordered graphs (i.e., graphs augmented with a linear order on the vertex set).
Let FO(<) be the language of first-order logic over
ordered graphs. The proof of Theorem~\ref{thm:orderinvariant} then shows:

\begin{theorem}
    Every node query expressible by an 
    $\group_<$-invariant \localmax \GNN
    is FO(<)-definable.
\end{theorem}

To illustrate this, consider the \localmax \GNN node classifier that 
 applies one round of aggregation to (i) collect
the minimum key $k$ across all neighbors of the distinguished node that satisfy $p_1$, (ii) collect
the maximum key $k'$ across all neighbors of the distinguished node that satisfy $p_2$, and (iii) outputs $k'-k$. This \GNN, which can be implemented using only \FFN(\relu,\sigmoid) combination functions, is $\group_<$-invariant, and the node query over ordered graph that it expresses is defined by 
the FO(<) formula $\phi(x)=\exists y,z(E(x,y)\land E(x,z)\land p_1(y)\land p_2(z)\land y<z)$. Here, \sigmoid is used to normalize the keys so that they belong to the interval $(0,1)$, as in Example \ref{ex:localsum}.

\subsection{Restricted Key Spaces}
 \label{sec:restricted}

We have assumed nothing about the node identifiers other than that they are unique. In practice, node identifiers may satisfy further structural properties, depending on the way they are generated. For example, keys can be fixed precision rational numbers, or they can be derived from Laplacian eigenvectors, as is common for positional encodings. This allows for a relaxation of ``key invariance'', requiring the \GNN to be invariant only over keyed graphs that satisfy such additional properties--which may increase the expressive power of key-invariant \GNNs. We illustrate this with two examples. First, for fixed precision rational keys the following result applies:
\begin{restatable}{theorem}{epsilonapart} \label{thm:epsilonapart}
    Every strongly local node query can be expressed on $\epsilon$-apart keyed graphs by a key-invariant \localsumcontinuous \GNN with policy $\geq 1/\leq0$.
\end{restatable}
Under this assumption, key-invariant \localsumcontinuous \GNNs are
thus as powerful as key-invariant \localsum \GNNs with arbitrary combination functions, which is not true with real keys (Theorems \ref{thm:closed-under-coverings} and \ref{thm:localsumarbunlimited}). Intuitively, this follows because
on $\epsilon$-apart keyed graphs
we can replace
arbitrary combination functions by continuous ones. As a second example, if we restrict keys to a bounded interval of reals, Example~\ref{ex:localsum} can be adapted to show that key-invariant 
\localsumrelu \GNNs express node queries that 
cannot be expressed by any key-oblivious \localsum \GNN (it suffices to simply omit the initial application of $\sigmoid$).

\subsection{Global Aggregation}

So far, we have only considered \GNNs with local aggregation.
It remains future work to extend our results to GNNs with global aggregation or global readout. We present here a single result for \globallocalsum \GNNs, in which layers $\calL = (\agg, \com)$ compute a function from embedded graphs to embeddings defined by:
\begin{align*}
    \calL(\agraph,\emb)(u) = ~& \com\big(\emb(u),\agg(\multiset{\emb(v) \,\mid (u,v) \in E(G)}),\\ & \hspace{18mm} \agg(\multiset{\emb(v) \,\mid v \in V(G)})\big)
\end{align*}
With semilinear combination functions these models can express isomorphism types, now without requiring the squaring function, and for pointed graphs that are not necessarily connected:
\begin{restatable}{theorem}{ThmIsoWithGlobalSum}
\label{thm:iso test with global sum}
    Let $\apgraph$ be a pointed graph. The query $\calQ_{\apgraph}$ that tests isomorphism with $\apgraph$ is expressed by a key-invariant \globalsumsemilinear \GNN.
\end{restatable}
While for \localsum we require an iterative procedure that applies $(\cdot)^2$ to determine the size of a graph, \globallocalsum \GNNs can do this with a single global aggregation where every node sends $1$.

\subsection{Testing Key-invariance}
\label{sec:testing}

We have focused on the expressive power
of key-invariant \GNNs, but another natural question is whether one can effectively test if a \GNN is key-invariant. Our results regarding $\uddl$ imply the following undecidability result:

\begin{restatable}{theorem}{thmUndecidabilityKeyInvariance} 
\label{thm:undecidable-invariance}
It is undecidable whether a given \localmaxsemilinear \GNN is key-invariant.
\end{restatable}

We leave it open whether key-invariance is decidable for \localmaxrelu \GNNs. Note that \localsumrelu \GNNs are known to have an undecidable satisfiability problem \cite{usicalp24}, which, by the same argument as for Theorem~\ref{thm:undecidable-invariance}, implies undecidability for testing key-invariance as well.

\section{Conclusion}

We proposed a notion of \emph{key-invariance} for \GNNs and established an initial set of results
on the expressive power of key-invariant \localmax and \localsum \GNNs across various classes of combination functions and acceptance policies,
as summarized in Tables~\ref{tab:localmxsummary}, and~\ref{tab:localsumsummary}, and Figures~\ref{fig:max GNNs hierarchy} and ~\ref{fig:sum-hierarchy}. 
Many interesting questions still remain open. To name a few:
\begin{itemize}
    \item Are key-invariant \localsumrelu \GNNs more expressive than key-oblivious ones?
    \item Are key-invariant \localmax \GNNs subsumed by FO? (we only showed subsumption by $\ordinvfo$)
    \item Are key-invariant \localsumsemilinear \GNNs subsumed by $\ordinvfoc$?
    \item Is every key-invariant \localmax \GNN equivalent to one with semilinear combination functions?
\end{itemize}

In addition, we identified several directions in Section~\ref{sec:discussion} that remain to be explored, such as key-invariance for specific positional encodings and key-invariance under specific symmetry groups. 

Our analysis draws on a range of techniques from logic and functional analysis, including collapse results from embedded finite model theory; functional bisimulations and coverings; modal and dynamic logics; and moment-injective real-valued functions. We expect that a systematic study of invariant \GNNs will deepen our understanding of their behavior in practical settings, to which the methods presented here provide a starting point.

This work is very much in line with research on order-invariant logics and on generic definability over embedded models, which has played an important role in the development of descriptive complexity and finite model theory. We view invariant definability as an important perspective on the interface between logic and computation.

\begin{acks}
Martin Grohe is funded by the European Union (ERC, SymSim, 101054974).
\end{acks}

\bibliographystyle{ACM-reference-format}
\bibliography{references}

\newpage
\appendix
\onecolumn

\section{Preliminaries to the appendix}
\label{appendix:prelims}
\subsection{Feature Expressions}

Throughout the appendix, 
we will freely work with an alternative presentation of GNNs which is more flexible than the standard one.
Fix a set $F$ of real-valued functions $f:\reals^n\to\reals$ A
\localmaxF feature expression is any expression built up inductively as follows:
\begin{enumerate}
    \item For each $p\in\Pi$, $\feature_p$ is an atomic feature expression. In addition, $\val$ is an atomic feature expression.
    \item If $\feature_1, \ldots, \feature_n$ are feature expressions and $f\in F$ is an $n$-ary function, then 
    $f(\feature_1,\ldots,\feature_n)$ is a feature expression.
    \item If $\feature$ is a feature expression, then
    $\localmax(\feature)$ is a feature expressions.
\end{enumerate}
Analogously, a \localsumF feature expression is built up as above but using $\localsum$ instead of $\localmax$. By the \emph{aggregation depth} of a feature expression we will mean the maximal nesting of \localmax and/or \localsum operators.

The semantics of a feature expression, denoted  $\sem{\feature}$, is a function from pointed valued graphs to real numbers. Formally:
\[\begin{array}{ll}
\sem{p_i}_{\agraph}(u) &=\lab(u)_i \\
\sem{\val}_{\agraph}(u) &=\val(u) \\
\sem{f(\feature_1, \ldots, \feature_n)}_{\agraph}(u) &=f(\sem{\feature_1}_{\agraph}(u),\ldots,\sem{\feature_n}_{\agraph}(u)) \\
\sem{\localmax(\feature)}_{\agraph}(u) &=\max_{(u,v)\in E(G)}\sem{\feature}_{G}(v) \\
\sem{\localsum(\feature)}_{\agraph}(u) &=\sum_{(u,v)\in E(G)}\sem{\feature}_{G}(v) \\
\end{array}\]
where, as elsewhere, we use the convention that $\max\emptyset=0$.
Note that in feature expressions, \localmax and \localsum operate on individual real-valued features, whereas in \GNNs they are applied coordinate-wise to vectors of reals.

Feature expressions over a function class $F$ correspond in a precise way to \GNNs that use functions from $F$ as combination functions. In order to state the correspondence precisely we must address a small discrepancy: for a set $F$ of real valued functions $f:\reals^n\to\reals$, let us denote by $\widehat{F}$ the set of all functions $g:\reals^n\to\reals^m$ for which there are functions $f_1,\ldots,f_m\in F$ such that $g(\vec{x})=\langle f_1(\vec{x}),\ldots,f_m(\vec{x})\rangle$. We note that all combination function classes we consider in the paper are of the form $\widehat{F}$.

\begin{proposition}
    Let $F$ be any set of real valued functions that includes all linear functions with rational coefficients.
    
    For every $\localmax$ feature expression $\feature$ using functions from $F$, there is a 
        $(|\Pi|+1,1)$ \localmax--$\widehat{F}$ \GNN $\agnn$, such that for all pointed valued graphs $\apgraph$, \[ \agnn(G,\emb_G)(u) = \langle \sem{\feature}_{\agraph}(u)\rangle.\]

        Conversely, for every  $(|\Pi|+1,D)$ \localmax--$\widehat{F}$ \GNN $\agnn$, there are $\localmax$ feature expressions $\feature_1,\ldots,\feature_D$ using functions from $F$, such that for all pointed valued graphs $\apgraph$, \[ \agnn(G,\emb_G)(u) = \langle \sem{\feature_1}_{\agraph}(u),\ldots,\sem{\feature_1}_{\agraph}(u)\rangle .\]
    
    The same applies in the \localsum case.
\end{proposition}
We omit the proof as it is a straightforward induction.

    In particular, this shows that 
    \localmax--$\widehat{F}$ \GNN node classifiers (which, by definition, have output dimension 1) stand in a two-way correspondence with 
    \localmax--$F$ feature expressions; and likewise in the \localsum case.
Thus, when reasoning about \GNN node classifiers
we can equivalently reason in terms of feature expressions. This can make things easier. For 
instance the \GNNs that were described informally as examples in the paper can be specified precisely and succinctly with feature expressions:

\begin{example} (Examples revisited)
\begin{itemize}
    \item
The \GNN described in Example~\ref{ex:gnn} is 
specified
by the expression
$\localmax(p_i)$.
\item
The \GNN described in Example~\ref{ex:diamond2top}
is specified
by 
the feature expression
\[|\localmax(\val)-\localmin(\val)|\]
where $\localmin(x)$ is short for $-\localmax(-x)$
and where $|x|$ is short for $\relu(x)+\relu(-x)$.
\item
The \GNN described in Example~\ref{ex:localsum}
is specified 
by the feature expression $\feature_x+|\feature_y-\feature_z|$ where $\feature_x$ is given by Theorem~\ref{thm:GMLloweroblivious} (and the above proposition), $\feature_y=\localsum(\relu(\sigmoid(\val) - 1 +\feature_p)$, and $\feature_z=\localsum(\localsum(\relu(\sigmoid(\val) - 1 +\feature_p))$.
\end{itemize}
\end{example}

In what follows, we will use \GNNs and feature expressions interchangeably.

\section{Proofs for Section~\ref{sec:prelims}}

\lemSemilinearOne*

\begin{proof}
    A well known equivalent definition of the semilinear functions is as  the functions $f:\reals^n\to \reals$ for which there is a finite 
  partition of $\reals^n$ definable by
  linear inequalities with rational coefficients, such that, within each   
  part of the partition,  $f$
  coincides with a linear function with rational coefficients. 
The  lemma follows quite immediately from this equivalent definition of semilinear functions.
\end{proof}

\lemSemilinearTwo*
\begin{proof}
      Let $f:\reals\to \reals$ be a unary semilinear function. Then there is a partition of the real line, definable by
  linear inequalities with rational coefficients, into finitely many intervals with rational end points such that $f$ behaves linearly in each interval. 
   From there it's easy to see that $f$ is 
  definable by a FFN with rational parameters using \relu and \Heavi. For example, consider the unary semilinear function
  \[ f(x)=\begin{cases} 
     f_1(x)&\text{if $x< a$}\\
     f_2(x)&\text{if $a\leq x\leq b$}\\
     f_3(x)&\text{if $x>b$}
     \end{cases}
  \]
  where $f_1,f_2,f_3$ are linear functions and $a,b \in \mathbb{Q}$. Then $f$ has an alternate definition as %michael: previous version was convoluted to parse
  $f(x) = f^{(-\infty,a)}(x) + f^{[a,b]}(x) + f^{(b,\infty)}(x)$ with:
  \begin{itemize}
      \item $f^{(-\infty,a)} = f_1(a -\relu(a-x)) - \Heavi(x-a)f_1(a)$
      \item $f^{[a,b]} = f_2(b -\relu(b-a-\relu(x-a))) - (1-\Heavi(x-a))f_2(a) - (1-\Heavi(b-x))f_2(b)$
      \item $f^{(b,\infty)} = f_3( b -\relu(x-b) ) - \Heavi(b-x)f_3(b)$
  \end{itemize}  
  Note that $f^{(-\infty,a)}(x)$ is equal to $f_1(x)$ for $x$ in the interval $(-\infty,a)$ and is equal to zero for all other inputs; and likewise for the other functions above. In each case,  $\relu$ is used (possibly twice) to clip the input of the function to the required interval, and $H$ is used to correct the offset outside of this interval.
\end{proof}
\section{Proofs for Section~\ref{sec:oblivious}}

\subsection{Proofs of Theorem~\ref{thm:localsumcontinuous expresses all local CR-invariant queries} and Proposition~\ref{prop:bisimulation-localmaxobliv}: Combinatorial Characterizations of Key-oblivious \GNNs}

The proof of Theorem \ref{thm:localsumcontinuous expresses all local CR-invariant queries} requires several auxiliary results. %michael: added

\begin{lemma} \label{lem:bijective}
    Let $\mathcal{M}(\mathbb{N})$ be the 
    set of all finite multisets of natural numbers. 
    There are continuous functions $f,g:\mathbb{R}\to\mathbb{R}$ such that the function $\multiset{n_1, \ldots, n_k} \mapsto 
        g(\Sigma_i f(n_i))$ is a bijection from $\mathcal{M}(\mathbb{N})$ to $\mathbb{N}$. 
\end{lemma}

\begin{proof}
    Let $f(x)=e^x$. It is known that $f$ is moment-injective on the natural numbers, meaning that  
    the map $\multiset{n_1, \ldots, n_k}\mapsto\Sigma_i f(n_i)$  injectively maps finite multisets of natural numbers to reals.
    Furthermore, it can be shown that $S=\{\sum_{i} e^{n_i}\mid \multiset{n_1, \dots, n_k} \text{ a finite multiset of natural numbers}\}$ is a discrete countably infinite set, i.e. for every $x \in S$ there is $\epsilon >0$ so that $(x-\epsilon, x+\epsilon)\cap S = \{x\}$. To see this, let $x \in S$, there exists some $j \in \mathbb{N}$ such that $x < e^j$. Then $(0, e^j]\cap S$ is a finite set, hence we let $\epsilon = min_{y \in (0, e^j]\cap S} \frac{|y-x|}{2}$. 
    Since $S$ is discrete we can apply a continuous bijection from $S$ to $\mathbb{N}$.
\end{proof}
\begin{proposition}
\label{prop:Oblivious continuous sum implements WL}
    For each $n\geq 0$, there is a key-oblivious local sum GNN $\feature_n$ with continuous combination functions, such that
    in each (unkeyed) graph $G$ and each vertex $v$, $\feature_n(v)\in\mathbb{N}$ 
    is the color assigned to $v$ by the CR-algorithm after $n$ rounds.
\end{proposition}

\begin{proof}
    By induction on $n$.
    The case for $n=0$ is
    simply because there are only finitely many local types. For $n+1$, recall
    that $col_{n+1}(v)=HASH(col_n(v),\multiset{col_n(u)\mid vEu})$. We apply the above lemma twice: let
    \[f'= g(f(\feature_n)+f(g(\localsum(f(\feature_n)))\]
    By construction and induction hypothesis, 
    $f'$ outputs a distinct natural number for each possible value
    of $col_{n+1}$. It remains only to map these natural numbers back to the actual corresponding values of $col_{n+1}$. This can be done by a continuous function, because every bijection on the natural numbers can be extended to a continuous function over the reals. 
\end{proof}

\begin{theorem}[{\cite{Otto04}, cf.~\cite[Proposition 3.7]{Otto12}}]
\label{thm:highgirthcover}
Let $G$ be an undirected graph and $k\ge 1$. Then there is a finite cover of $G$ of girth at least $k$.
\end{theorem}

Since \cite{Otto04,Otto12} uses a slightly different notion of covers, we include a self-contained proof below. It uses the following lemma. 

\begin{lemma}[\cite{Alon1995}]
Let $k,d\ge 1$. There is a finite group $\Gamma$ with a generating set $\Delta$ of order $|\Delta|=d$ such that

\begin{enumerate}
\item $\delta=\delta^{-1}$ for all $\delta\in\Delta$;
\item for all $\ell\ge 1$, $\delta_1,\ldots,\delta_\ell\in\Delta$ with $\delta_i\neq\delta_{i+1}$ for all $i$,
if $\delta_1\ldots\delta_\ell=1$ then $\ell\ge k$.
\end{enumerate}
\end{lemma}

\begin{proof}
    Take a finite $d$-regular tree $T$ that has a node $r$ (the
``root'') of distance at least $k/2$ from all leaves. color its edges
with $d$ colors $1,\ldots,d$ such that each node is either a leaf or is
incident with exactly one edge of color $i$. The group we shall
construct is a subgroup of the symmetric group on $V(T)$. For
$i=1,\ldots,d$, let $\delta_i$ be the permutation of $V(T)$  that maps
each node $v$ to its unique $i$-neighbor, or fixes $v$ if it is a
leaf without an $i$-neighbor. Let $\ell\ge 1$ and $1\le
i_1,\ldots,i_\ell\le d$ such that $i_j\neq i_{j+1}$ for all $j$ and
$\delta:=\delta_{i_1}\ldots\delta_{j_\ell}=1$. Then $r^\delta=r$ for the
root $r$. To achieve this, $\delta$ must first take $r$ to a leaf and
then back to $r$. This implies that $\ell\ge k$, which proves the lemma.
\end{proof}

Recall that the \emph{girth} of an undirected graph is the
size of the smallest cycle (or $\infty$ if the graph is acyclic). By a \emph{cover} of a graph $G$ we will mean a graph $H$ such that $H\covers G$.

\begin{proof}[Proof of Theorem~\ref{thm:highgirthcover}]
Take an arbitrary edge-coloring of $G=(V,E)$ such that no vertex is
incident to two edges of the same color. Say, the colors are $1,\ldots,d$,
and let $E_i$ be the set of all edges of color $i$. Let $\Gamma$ be a group
with generating set $\Delta=\{\delta_1,\ldots,\delta_d\}$ as in the lemma.

We define a graph $H$ as the product of $G$ with the Cayley graph of
$\Gamma$ in the following way: the vertex set is $V(H):=V\times\Gamma$
and the edge set is
$$
E(H):=\{ (v,\gamma)(w,\gamma\delta_i) | i\in[d],vw\in E_i \}.
$$
Then $H$ is an undirected graph, because
$\delta_i=\delta_i^{-1}$ for all $i$. Moreover, the projection
$\pi:(v,\gamma)\mapsto v$ is a covering map:
\begin{itemize}
\item it is surjective;
\item it is a homomorphism, because if $(v,\gamma)(w,\gamma')\in E(H)$ then
  for some $i\in[d]$ it holds that $vw\in E_i\subseteq E$;
\item  it is locally bijective, %michael: not defined, but I guess it is pretty clear
because if $v\in V$ with neighbors
  $w_1,\ldots,w_\ell$ such that $vw_i\in E_{j_i}$ and
  $\gamma\in\Gamma$ then the neighbors of $(v,\gamma)$ are $(w_1,\gamma\delta_{j_1}),\ldots,(w_\ell,\gamma\delta_{j_\ell})$.
\end{itemize}

Finally, the girth of $H$ is at least $k$. To see this, suppose that
$(v_1,\gamma_1),\ldots,(v_\ell,\gamma_\ell)$ is a cycle in $H$. For
every $i\in[\ell]$, let $j_i$ be the color of the edge $v_iv_{i+1}$,
where we take $(v_{\ell+1},\gamma_{\ell+1}):=(v_1,\gamma_1)$. Then
$j_i\neq j_{i+1}$, and we have
$\gamma_1=\gamma_1\delta_{j_1}\ldots\delta_{j_\ell}$, which implies
$\delta_{j_1}\ldots\delta_{j_\ell}=1$ and hence $\ell\ge k$.  
\end{proof}

\begin{lemma}
\label{lem:same root color implies isomorphic pointed trees}
Let $r\geq 0$, and let $\apgraph, \altpgraph$ be pointed trees such that $CR^{(r)}(\apgraph) = CR^{(r)}(\altpgraph)$. Then $\apgraph \rneighborhood$ is isomorphic to $\altpgraph \rneighborhood$.    
\end{lemma}
\begin{proof}

We choose inclusion-wise maximal sets $X\subseteq V(\agraph)$ and $X'\subseteq V(\altgraph)$ such that there is a bijective mapping $h:X\to X'$ satisfying the following conditions.
\begin{enumerate}
    \item $u \in X, v \in X'$ and $h(u)=v$.
    \item $h$ is an isomorphism from the induced subgraphs of $X$ and $X'$.
    \item If $h(x)=x'$ then $dist(u,x)=dist(v,x')$, and if $dist(u,x)=i$ then $CR^{(r-i)}(\agraph^x)=CR^{(r-i)}(\altgraph^{x'})$.
    \item $X$, $X'$ are downward closed, that is, if $x\in X$ and $y$ is the parent of $x$ then $y\in X$, and similarly for $X'$.
\end{enumerate}
Such sets $X,X'$ exists, because $X=\{u\},X'=\{v\}$ with $h:u\mapsto v$ satisfy conditions (1)--(4). We claim that $\apgraph \rneighborhood \subseteq X$ and $\altpgraph \rneighborhood \subseteq X'$.

Suppose for contradiction that the claim is false. Let $i\le r$ be the minimum such that either there is a vertex $y\in V_G\setminus X$ with $dist(u,y)=i$ or there is a vertex $y'\in V_H\setminus X'$ with $dist(v,y')=i$. By symmetry, we may assume the former, so let  $y\in V_G\setminus X$ with $dist(u,y)=i$. Let $x$ be the parent of $y$. Then $x\in X$. Let $x':=h(x)$, let $d:=CR^{r-i}(\agraph^y)$ and let $y_1,..,y_k$ be a list of all neighbors of $x$ in $\agraph$ such that $CR^{r-i}(\agraph^{y_i})=d$. Assume that $y=y_k$ and for some $\ell\in[k]$ we have $y_1,..y_{\ell-1}\in X$ and $y_\ell,...,y_k\not\in X$. Since $CR^{r-i+1}(\agraph^x)=CR^{r-i+1}(\altgraph^{x'})$, there are distinct neighbors $y_1',...,y_k'$ of $x'$ in $\altgraph$ such that $CR^{r-i}(\altgraph^{y_i'})=d$. Among these neighbors are $h(y_1),..,h(y_{\ell-1})$. Without loss of generality, we may assume that $h(y_j)=y_j'$ for $j\in[\ell-1]$. Then $y_{\ell}',\ldots,y_k'\not\in X'$, because if $y_j'\in X'$ then $h^{-1}(y_j')$ is a neighbor of $x$ in $X$ and hence among $y_1,...,y_{\ell-1}$. In particular, $y_k'$ is a neighbor of $x'$ in $V(H)\setminus X'$ with  $CR^{r-i}(\altgraph^{y'_k})=d=CR^{r-i}(\agraph^{y_k})$.
Since $X'$ is downward closed, $y_k'$ is not the parent of $x'$ in $\altpgraph$ and hence $dist(v,y_k')=dist(v,x')+1=dist(u,x)+1=dist(u,y)=i$.

Thus we can extend $X,X',h$ by adding $y_k,y_k'$. This contradicts the maximality of $X,X'$.
\end{proof}

We are now ready to prove Theorem \ref{thm:localsumcontinuous expresses all local CR-invariant queries}.

\thmOblivLocalSumCR*

\begin{proof}
    (4) to (2) is by Proposition \ref{prop:Oblivious continuous sum implements WL}: we can compute $col_n$. 
    Finally, for any subset $X\subseteq \mathbb{N}$ there is a
    continuous function $f$  over the reals such that $f(n)=1$ for
    $n\in X$ and $f(n)=0$ for
    $n\in\mathbb{N}\setminus X$.

    (2) to (1) is immediate.

    (1) to (3):  strong locality is immediate (where the radius is the number of layers of the \GNN); CR-invariance was proven in \cite{howpowerful,goneural}.

    (3) to (4): 
    Let $\calQ$ be strongly local with radius $r$, and CR-invariant. Furthermore, 
    let $\apgraph$ and $\altpgraph$ be pointed graphs such that
$CR^{(r)}\apgraph=CR^{(r)}\altpgraph$ and $\altpgraph\in\calQ$.
It is our task to show that $\apgraph\in\calQ$.
Let $\widehat{G}^u$ and $\widehat{H}^v$ be covers of $\apgraph$ and $\altpgraph$ of girth greater than $r$ (as provided by Theorem~\ref{thm:highgirthcover}, where the node features of each node are inherited through the covering projection).
Now consider the subgraphs
$\widehat{G}^u\upharpoonright r$ and $\widehat{H}^v\upharpoonright r$. By construction, these are acyclic, and $CR^{(r)}(\widehat{G}^u \rneighborhood) = CR^{(r)}(\widehat{H}^v \rneighborhood)$. Hence, by Lemma \ref{lem:same root color implies isomorphic pointed trees} they are isomorphic. 

Thus, we have the following diagram.
    
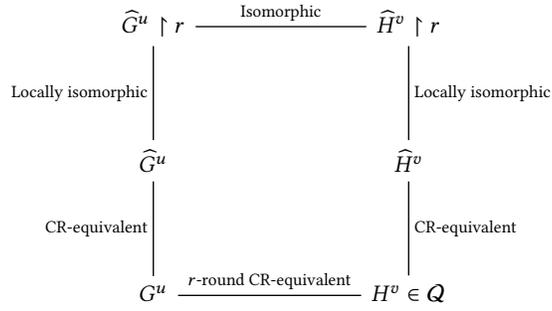
\[
\begin{tikzcd}[row sep=4em, column sep=7em]
\widehat{G}^u\upharpoonright r \arrow[r, no head, "\text{Isomorphic}"] & \widehat{H}^v\upharpoonright r \\
\widehat{G}^u \arrow[u, no head, "\text{Locally isomorphic}"] & \widehat{H}^v  \arrow[u, no head, "\text{Locally isomorphic}"']\\
\apgraph \arrow[u, no head, "\text{CR-equivalent}"] \arrow[r, no head, "\text{$r$-round CR-equivalent}"] & \altpgraph \in \calQ \arrow[u, no head, "\text{CR-equivalent}"']
\end{tikzcd}
\]
Chasing the diagram, we obtain that $\apgraph\in\calQ$.
\end{proof}

\propbisimulationlocalmaxobliv*

\begin{proof}
    The implication from (2) to (1) is immediate. 
    
    From (1) to (3), it is clear that every node query expressed by a \GNN with local aggregation
    is strongly local (where the locality radius is determined by the number of layers). The 
    fact that \localmax \GNNs are invariant for bisimulations is shown in the cited papers. 
    
    The direction (3) to (2) is not proven in the cited papers, but follows by standard arguments from the modal logic literature. We describe the argument in some detail. Let $r$ be the radius with respect to which $\calQ$ is strongly local.
    We define an \emph{$r$-bisimulation} between graphs $\apgraph$ and $\altpgraph$ to be a tuple
    $(Z_0, Z_1, \ldots, Z_r)$ with
     $Z_0 \subseteq Z_1 \dots \subseteq Z_r \subseteq V(G)\times V(H)$ so that $(u,v) \in Z_0$, and for all $(u_0,v_0) \in Z_i$:
    \begin{itemize}
        \item $\lab_G(u_0)=\lab_H(v_0)$;
        \item (for $i<r$) For each edge $(u_0,u_1)\in E(G)$ there is an edge $(v_0,v_1)\in E(H)$ with $(u_1,v_1)\in Z_{i+1}$;
        \item (for $i<r$) For each edge $(v_0,v_1)\in E(H)$ there is an edge $(u_0,u_1)\in E(G)$ with $(u_1,v_1)\in Z_{i+1}$.
    \end{itemize}
    A well-known inductive argument shows that,
    for every pointed graph $\apgraph$ and $r\geq 0$,
    there is a modal formula $\phi_{\apgraph}$ of modal depth $r$ that completely describes $\apgraph$ up to $r$-bisimulation, i.e., such that $\altpgraph\models\phi_{\apgraph}$ if and only if there is an $r$-bisimulation between $\apgraph$ and $\altpgraph$ \cite{blackburn2001m}. Furthermore, it is also well known that there are only finitely many modal formulas of modal depth $r$ up to equivalence. Let $\psi$ now
    be the (finite) disjunction of $\phi_{\apgraph}$ for all $\apgraph\in\calQ$. We will show that $\psi$ defines $\calQ$ --- it then follows by results in the papers cited above that $\psi$ can be expressed by a \localmax \GNN. Clearly, 
    every $\apgraph\in\calQ$ satisfies $\psi$. Conversely, let $\apgraph\models\psi$. By construction of $\psi$, then there is a 
    pointed graph $\altpgraph\in\calQ$ such that
    $\apgraph$ and $\altpgraph$ are $r$-bisimilar.
    As shown in \cite[Lemma 24]{Otto04}, from $G$ and $H$ one can 
    construct graphs $\widehat{G}$ and $\widehat{H}$, respectively, such
    that the $r$-bisimulation between $\apgraph$ and $\altpgraph$ can be ``upgraded'' to a proper bisimulation as in the
    following diagram:

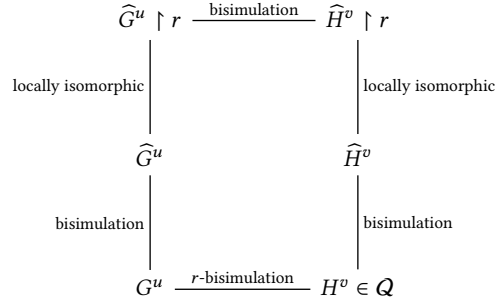
\[
\begin{tikzcd}[row sep=4em, column sep=5em]
\widehat{G}^u\rneighborhood \arrow[r, no head, "\text{bisimulation}"] & \widehat{H}^v\rneighborhood \\
\widehat{G}^u \arrow[u, no head, "\text{locally isomorphic}"] & \widehat{H}^v  \arrow[u, no head, "\text{locally isomorphic}"']\\
\apgraph \arrow[u, no head, "\text{bisimulation}"] \arrow[r, no head, "\text{$r$-bisimulation}"] & \altpgraph \in \calQ \arrow[u, no head, "\text{bisimulation}"']
\end{tikzcd}
\]
Intuitively, $\widehat{G}^u$ and $\widehat{H}^v$ are  obtained by unravelling $\apgraph$ and $\altpgraph$
up to depth $r$ and attaching disjoint copies of the original graph to the leaves. This suffices to ensure that $\widehat{G}^u\rneighborhood$ and $\widehat{H}^v\upharpoonright$ are acyclic so that $\widehat{G}^u\rneighborhood$ is bisimilar to $\widehat{H}^v\upharpoonright$.
We note that the construction is described in \cite{Otto04} for directed graphs, but it can be adapted for undirected graphs simply by taking the symmetric closure of $\widehat{G}$ and $\widehat{H}$.
Chasing the diagram we now obtain that $\apgraph\in\calQ$
\end{proof}

\subsection{Proof of Proposition~\ref{prop:Q_even sum cont vs sum relu}: $\calQ_{\text{even}}$ separates \localsumcontinuous from \localsumrelu}
\PropQevenSumvsRelu*

\begin{proof}
A key-oblivious \localsumcontinuous \GNN expresses $\calQ_{\text{even}}$ by counting the neighbors with $\localsum(1)$, and applying a continuous function that maps natural numbers to $1$ if they are even and $0$ if they are odd. Now let $\agnn$ be a \localsumrelu \GNN and let $\{G_k : k \in \mathbb{N}\}$ be a set of graphs where $V(G_k) = \{u_k, u^1_k, \dots u^k_k\}$, the only edges are $(u_k, u^i_k)$ for $1 \leq i \leq k$, and all nodes have the empty labeling. We show that the output feature $\feature_\agnn$ computed by $\agnn$ is eventually polynomial in $k$,
meaning that there exist polynomials $f_{\feature_\agnn}, g_{\feature_\agnn}$ and $S \in \mathbb{N}$ such that when $k \geq S$:
    \begin{enumerate}
        \item $\sem{\feature_\agnn}_{G_k}(u_k) = f_{\feature_\agnn}(k)$.
        \item $\sem{\feature_\agnn}_{G_k}(u^i_k) = g_{\feature_\agnn}(k)$ for $1 \leq i \leq k$. 
    \end{enumerate}    
Assuming we have proven this, then since $\calQ_{\text{even}}$ is not eventually polynomial it is not expressed by $\agnn$. 

To prove the claim we reason by induction on the feature expression. The input features are constant and hence polynomial. Assume sub-expression $\feature$ of $\feature_\agnn$ is polynomial for $k \geq S$, i.e., for $k \geq S$ and $i \leq k$ $\sem{\feature}_{\agraph_k}(u_k) = f_\feature(k)$ and $\sem{\feature}_{\agraph_k}(u^i_k) = g_\feature(k)$. Now suppose $\feature' = \localsum(\feature)$. Then:
    \begin{align*}
        \sem{\feature'}_{G_k}(u_k) &= \sum^k_{i=1} \left(\sem{\feature}_{G_k}(u^i_k)\right) = k \cdot g_\feature(k)\\
       \text{for $1 \leq i \leq k$} \sem{\feature'}_{G_k}(u^i_k) &= \sem{\feature}_{G_k}(u_k) = f_\feature(k)
    \end{align*}
    Thus $\feature'$ is polynomial for $k \geq S$.

    Now let $\feature'$ be computed by a \FFN(\relu) from eventually polynomial features $\feature_1, \dots \feature_m$. For sufficiently large $k$, all input features are polynomial in $k$. Every pre-activation is then a linear combination of polynomials, which is a polynomial $h(k)$. Since $h$ has finitely many roots, beyond some threshold of $k$ each \relu gives either output $0$ or output $h(k)$, so that the outputs of the first layer are polynomial. By induction over the network depth, $\feature'$ is eventually polynomial.
\end{proof}

\section{Proofs for section \ref{sec:max}}

\subsection{Proofs of Theorem \ref{thm:max yes-open closed under inverse fbisim}, Corollary \ref{cor:Max iso continuous collapse} and Theorem \ref{thm:collapseyesnolocalmax}: Upper Bounds on \localmaxcontinuous \GNNs}

The following lemma is derived by straightforward induction over the feature expression:
\begin{lemma}
\label{lem:ValuedBisimLimitsMax}
Let $\agnn$ be a \localmax \GNN and let $B$ be a bisimulation between $\apgraph_\val$ and $\altpgraph_\val$ that preserves values. Then for all $(u,v) \in B$, $\agnn(\agraph_\val)(u) = \agnn(\altgraph_\val)(v)$.
\end{lemma}

\ThmMaxClosedUnderFbisim*
\begin{proof}
    Let $\agnn$ be a \localmaxcontinuous \GNN that decides $\calQ$. Let $\apgraph \in \calQ^c$, let $B$ be a functional bisimulation from $\apgraph$ to $\altpgraph$ and assume $\altpgraph \not\in \calQ^c$. We derive a contradiction.

    Take a keyed extension $\altpgraph_\key$ and let $\apgraph_\val$ be the valued extension obtained by assigning to each node in $\apgraph$ the key of its image under $B$. $\apgraph_\val$ need not be keyed since $B$ need not be injective. Since $B$ is a bisimulation between $\apgraph_\val$ and $\altpgraph_\key$ that preserves values, by Lemma~\ref{lem:ValuedBisimLimitsMax} they obtain the same output value from $\agnn$. This must be a positive value since $\altpgraph \in \calQ$. By continuity, it follows that for sufficiently small $\epsilon$, $\agnn$ outputs a positive value for each $\apgraph_{\val'}$ where the values are perturbed by at most $\epsilon$. Among such $\apgraph_{\val'}$ we can find a keyed extension of $\apgraph$ that is accepted by $\agnn$, contradicting $\apgraph \in \calQ^c$.
\end{proof}

\CorIsoCollapseMaxContinuous*
\begin{proof}
    (3) to (2) and (2) to (1) are immediate. For (1) to (3), let $\apgraph$ be a non-singleton pointed graph and let $\altgraph$ be the graph obtained from $\agraph$ by adding an extra node $v$ with the same label and neighbors as $u$. The identity map extended by sending $v$ to $u$ is a functional bisimulation from $\altpgraph$ to $\apgraph$. By Theorem~\ref{thm:max yes-open closed under inverse fbisim}, $\calQ_{\apgraph}$ is not expressed by a key-invariant \localmaxcontinuous \GNN.
\end{proof}

\begin{lemma}
\label{lem:fbisim both ways implies bisim}
Let $\calQ$ be a node query such that $\calQ$ and $\calQ^c$ are closed under functional bisimulations. Then $\calQ$ is closed under bisimulations.
\end{lemma}
\begin{proof}
Suppose $\apgraph \in \calQ$ and let $B$ be a bisimulation between $\apgraph$ and $\altpgraph$. We show $\altpgraph \in \calQ$. We use a construction from modal logic known as \emph{bisimulation products}. Let $\agraph \times \altgraph$ be the direct product of $\agraph$ and $\altgraph$ with nodes $V(\agraph) \times V(\altgraph)$, edge set $\{((u_1,u_2),(v_1,v_2)) | (u_1,u_2) \in V(\agraph) \text{ and} (v_2,v_2) \in V(\altgraph)\}$ and labeling $\lab((u_1,v_1)) = \lab_\agraph(u_1) \cap \lab_\altgraph(v_1)$, where $\lab_\agraph, \lab_\altgraph$ are the labelings of $\agraph, \altgraph$. Observe that $B$ is a subset of the domain of $G\times H$. Let $K$ be the subgraph of $G \times H$ induced by the set $B$. Finally, for $i\in\{1,2\}$,
    let $B_i=\{((w_1,w_2),w_i)\mid (w_1,w_2)\in B\}$. In other
    words, $B_1$ and $B_2$ are the graphs of the natural projections.

    Now $B_1$ is a functional bisimulation from $K^{(u,v)}$ to $\apgraph$ and $B_2$ is a functional bisimulation from $K^{(u,v)}$ to $\altpgraph$. Since $\calQ$ and $\calQ^c$ are both closed under functional bisimulation and $\apgraph \in \calQ$, it follows that $\altpgraph \in \calQ$.
\end{proof}

\ThmMaxCollapse*
\begin{proof}
    By Theorem \ref{thm:ML-localmaxobliv} (3) and (4) are equivalent. It remains to show that (1), (2) and (3) are equivalent. The directions from (3) to (1) and (1) to (2) are immediate. Now assume (2). By Theorem \ref{thm:max yes-open closed under inverse fbisim} $\calQ$ and $\calQ^c$ are closed under functional bisimulation. By Lemma \ref{lem:fbisim both ways implies bisim} then $\calQ$ is closed under bisimulation. $\calQ$ is also strongly local since it is expressed by a key-invariant \localmaxcontinuous \GNN. By Proposition~\ref{prop:bisimulation-localmaxobliv} then $\calQ$ is expressed by a key-oblivious \policy{>0}{<0} \localmaxrelu \GNN.
\end{proof}

\subsection{Proof of Theorem \ref{thm:WGML encoded by localmax}: \localmaxcontinuous \GNNs express \WGML}
\begin{lemma}
\label{lem:localmaxrelu continues formula construction}
    Let $\phi, \psi$ be formula's in \WGMLmodal, and suppose \GNNs $\agnn_\phi, \agnn_\psi$ express $\calQ_\phi, \calQ_\psi$. Let $\xi$ be either one of $\phi \vee \psi,  \phi \wedge \psi, \Diamond \phi$. There then exist a \GNN expressing $\calQ_\xi$ that uses combination and aggregation functions from $\agnn_\phi, \agnn_\psi$, as well as max aggregation and \FFN(\relu) combination.
\end{lemma}
\begin{proof}
    \begin{align*}
        \feature_{\phi \vee \psi} &= \relu(\feature_\phi) + \relu(\feature_\psi)\\
        \feature_{\phi \wedge \psi} &= \min(\feature_\phi, \feature_\psi)\\
        \feature_{\Diamond \phi} &= \localmax(\feature_\phi)
    \end{align*}
    where $\min(x,y)= x - \relu(x-y)$.
\end{proof}

\begin{lemma}
\label{lem:MLexpressedbyMaxReLU1/0}
    Let $\phi \in ML$, then $\calQ_\phi$ is expressed by a key-oblivious \policy{\geq 1}{\leq 0} \localmaxrelu \GNN.
\end{lemma}
\begin{proof}
    This follows from \cite{schonherr2025logical}(Theorem 4).
\end{proof}

\LocalMaxEncodesWGML*
\begin{proof}
    Let $\apgraph \models \Diamond^{\geq 2}_r \phi$ if and only if $u$ has at least two nodes at distance $r$ that satisfy $\phi$. We show that a \localmaxrelu \GNN expresses $\calQ_{\Diamond^{\geq 2}_r \top}$ and a \localmaxsigmoid \GNN expresses $\calQ_{\Diamond^{\geq 2}_r \phi}$ for a modal formula $\phi$ and $r \geq 1$. The case $r=1$ is sufficient for the theorem. By Lemma \ref{lem:MLexpressedbyMaxReLU1/0} every modal formula is expressed by a key-invariant \policy{\geq 1}{\leq 0} \localmaxrelu \GNN.

    Now for (a), let $\phi = \Diamond^{\geq 2}_r \top$ for some $r \geq 1$. Then $\calQ_\phi$ is expressed by the following feature:
    \begin{align*}
        \feature_\phi = \localmax_r(\feature_\val) + \localmax_r(-\feature_\val)
    \end{align*}
    where $\localmax_r$ denotes $r$ consecutive applications of $\localmax$. Now if node $u$ in  $\apgraph$ has no $r$-hop neighbors, $\sem{\feature_\phi}_\agraph(u) = 0+0$. If $u$ has $1$ $r$-hop neighbor with key $k_v$, $\sem{\feature_\phi}_\agraph(u) = k_v-k_v=0$. However, if $u$ has at least $2$ $r$-hop neighbors, $\sem{\feature_\phi}_\agraph(u)$ is the difference between the largest and smallest key of these nodes.
    
    For (b), let $r \geq 1$ and let $\phi$ be a modal formula.    
We show $\calQ_{\Diamond^{\geq 2}\phi}$ is expressed by a feature that uses $\FFN(\relu,\sigmoid)$ combination functions:
        \begin{align*}
        \feature_{\val,b} &= \underbrace{\sigmoid(\feature_\val)}_{\text{bounded in}(0,1)}\\
        \feature_{\Diamond^{\geq 2}\phi} &=  \localmax_r(\underbrace{\relu(\feature_{\val,b} - 1 + \feature_{\phi})}_{\feature_{\val,b} \text{ for $\phi$ nodes and $0$ for non-$\phi$ nodes}})
        + \localmax_r(\underbrace{\relu(-\feature_{\val,b}+\feature_{\phi})-1}_{-\feature_{\val,b}\text{ for $\phi$ nodes and $-1$ for non-$\phi$ nodes}})
    \end{align*}
where again $\localmax_r$ represents $r$ consecutive applications of $\localmax$. If node $u$ in $\apgraph$ has no $r$-hop neighbors that satisfy $\phi$, $\sem{\feature_{\Diamond^{\geq 2}\phi}}_\agraph(u) = 0-1$. If $u$ has $1$ such $r$-hop neighbor with key $k_v$, $\sem{\feature_{\Diamond^{\geq 2}\phi}}_\agraph(u) = \sigmoid(k_v)-\sigmoid(k_v)=0$. If $u$ has at least $2$ $r$-hop neighbors that satisfy $\phi$, $\sem{\feature_{\Diamond^{\geq 2}\phi}}_\agraph(u)$ is the difference (after applying $\sigmoid$) between the largest and smallest key of these nodes. 

For both (a) and (b), the result now follows by applying Lemma~\ref{lem:localmaxrelu continues formula construction}.
\end{proof}

\subsection{Proof of Theorem \ref{thm:uddllower}: $\uddl$ Formulas can be translated to Semilinear {\localmax} GNNs }

Recall the statement:
\thmuddllower*

\begin{proof}

Recall the grammar of $\uddl$

\begin{tabular}{l}
    Formulas: \\ $\phi ::= p \mid \top \mid \phi_1\land \phi_2 \mid \neg\phi \mid \langle \pi\rangle\phi \mid  [\pi]\phi \mid \langle\pi\rangle^{=1}\phi$ \\ \\
    Program expressions: \\
$\pi ::= \step \mid \test(\phi) \mid \pi_1;\pi_2 \mid \pi_1\cup \pi_2$ 
\end{tabular}

To reduce the number of cases in the inductive translation, we will make some
simplifying assumptions on the input formula. First of all,
using the distribution laws
\[\begin{array}{lll}
(\pi_1\cup\pi_2);\pi_3 &\equiv& \pi_1;\pi_3\cup \pi_2;\pi_3\\
\pi_1;(\pi_2\cup\pi_3) &\equiv& \pi_1;\pi_2\cup \pi_1;\pi_3\\
\end{array}\]
together with the associativity of the composition operator $;$
and the fact that $\stay$ acts as an identity for
composition,
we can 
ensure that all program expressions occurring in the input formula are
unions of program expressions of the form $\pi_1;\cdots;\pi_n;\stay$
with each $\pi_i$ of the form $\test(\phi)$ or $\step$.
In addition, we can rewrite $[\pi]\phi$ to $\neg\langle\pi\rangle\neg\phi$ so that we can
assume the $[\cdot]\cdot$ operator does not occur in the formula. Finally, 
we have that $\langle\pi\rangle\phi$
and $\langle\pi\rangle^{=1}\phi$ are
equivalent to 
$\langle\pi;\test(\phi)\rangle\top$
and $\langle\pi;\test(\phi)\rangle^{=1}\top$, respectively, so that we can assume that all subformulas of the input formula of the form $\langle\pi\rangle\phi$ or $\langle\pi\rangle^{=1}\phi$ are such that $\phi=\top$.

Recall that  a $\uddl$ program $\pi$ defines a binary relation $\sem{\pi}_G$ on each graph $G$. 
It is easy to see that a formula of the form $\langle\pi\rangle\top$ is satisfied at a node $u$ precisely if $u\in dom(\sem{\pi}_G)$, where $dom(\sem{\pi}_G)=\{u \in V \mid (u,v) \in \sem{\pi}_G\}$. 

We prove by simultaneous induction the following claims:

\begin{enumerate}
    \item For each formula $\phi$ there is a GNN $\feature_\phi$
    such that, for all pointed keyed graphs $\apgraph$, $\feature_\phi(u)=1$ when $u\in \sem{\phi}_{G}$ and $0$ otherwise,
    \item For each program expression $\pi$, there is a GNN
    $\feature_\pi$ such that, for all pointed keyed graphs $\apgraph$, $\feature_\pi(u)=1$ when $u\in dom(\sem{\pi}_G)$ and 0 otherwise,
    \item For each program expression $\pi$, there are GNNs
    $\feature_\pi^{\min}$ and $\feature_\pi^{\max}$ such that,
    for all pointed keyed graphs $\apgraph$ with $u\in dom(\sem{\pi}_G)$, 
    \begin{itemize}
        \item $\feature_\pi^{\min}(u)=\min \{\val(v)\mid (u,v)\in\sem{\pi}_{G}\}$,
    and \item $\feature_\pi^{\max}(u)=\max \{\val(v)\mid (u,v)\in\sem{\pi}_{G}\}$.
    \end{itemize}
\end{enumerate}

The base cases and inductive cases for the above two claims are all listed in Table~\ref{tab:uddl-translation}.

\begin{table}
\[\begin{array}{ll}
\feature_{p_i} &= p_i \\
\feature_{\top} &= 1 \\
\feature_{\phi\land\psi} &= \text{ReLU}(\feature_\phi+\feature_\psi-1) \\
\feature_{\neg\phi} &= 1-\feature_\phi \\
\feature_{\langle\pi\rangle\top} &= \feature_{\pi} \\
\feature_{\langle\pi\rangle^{=1}\top} &= 
\text{IfZero}(\feature_{\pi},0, \text{IfZero}(\feature_{\pi}^{\max}-\feature_{\pi}^{\min},1,0)) \\
\\
\feature_{\stay} &= 1 \\
\feature_{\test(\phi);\pi} &= \text{IfZero}(\feature_\phi, 0, \feature_\pi) \\
\feature_{\step;\pi} &= \localmax(\feature_\pi)\\
\\
\feature^{\min}_{\stay} &= \val\\
\feature^{\max}_{\stay} &= \val \\
\\
\feature^{\min}_{\test(\phi);\pi} &= \text{IfZero}(\feature_\phi,\feature^{\min}_\pi,0) \\
\feature^{\max}_{\test(\phi);\pi} &= \text{IfZero}(\feature_\phi,\feature^{\max}_\pi,0) 
\\
\\
\feature^{\min}_{\step;\pi} &= \localmin(\text{IfZero}(\feature_{\pi},f,\feature^{\min}_{\pi}))  \\
 & \text{where $f=\localmax(\localmax(\feature_\pi^{\min}))$} \\ \\
\feature^{\max}_{\step;\pi} &= \localmax(\text{IfZero}(\feature_{\pi},f,\feature^{\max}_{\pi}) \\
& \text{where $f=\localmin(\localmin(\feature_\pi^{\max}))$} 
\end{array}\]
\caption{Translation from \uddl to \localmaxsemilinear feature expressions.}
\label{tab:uddl-translation}
\end{table}

All clauses above speak for themselves. except perhaps the
last two. We explain only the last clause, as the penultimate one is dual. Let $\apgraph$  be a pointed keyed graph such that $u\in dom(\sem{\step;\pi}_G)$. Let $v^*$ be a neighbor of $u$ maximizing the value of $\feature^{\max}_\pi(v^*)$. For each neighbor $v$ of $u$, one of the following cases applies:
\begin{enumerate}
    \item $v\in dom(\sem{\pi}_G)$. In this case, $v$ contributes $\feature^{\max}_\pi(v) \leq \feature^{\max}_\pi(v^*)$ to the \localmax aggregation.
    \item $v\not\in dom(\sem{\pi}_G)$. In this case, 
    the first branch of the $\text{IfZero}$ case distinction kicks in, and $v$ contributes to the \localmax aggregation the value constructed by the expression $f$, which is, by construction, some value upper bounded by $\feature^{\max}_\pi(v^*)$ (since it is computed by two iterations of \localmin, and $v^*$ is a two-step neighbor of $v$).
\end{enumerate}
Thus, each node $v$ contributes to the \localmax aggregation some value upper bounded by $\feature^{\max}_\pi(v^*)$, with the node $v^*$ itself contributing $\feature^{\max}_\pi(v^*)$.
It follows that the result of the \localmax aggregation is precisely 
$\feature^{\max}_\pi(v^*)$.
By induction hypothesis, this is equal to the intended output.

\end{proof}

\subsection{Proof of Theorem \ref{thm:diamond geq2 p not expressed by localmaxifpos}: $\Diamond^{\geq 2}p$ is not expressed by a \localmaxunarysemilinear \GNN}

The proof of Theorem~\ref{thm:diamond geq2 p not expressed by localmaxifpos} will involve \localmax--\FFN(\relu,\Heavi) \GNN node classifiers that run over graphs with a single
proposition letter. Such a \GNN node classifier can be identified with a \localmax--\FFN(\relu,\Heavi) feature expression built up from two atomic features, namely $\val$ and $p$. Towards this, it will be helpful to study \FFN(\relu,\Heavi)-functions $f(x,y)$, where the variable $y$ takes values from a finite set (think: $\{0,1\}$). 
The next two lemmas intuitively show that, for such functions, 
the arguments $x$ and $y$ can interact with each other only in a limited way.

\begin{lemma}
\label{lem:Univariate semilinear is eventually_linear}    
    Let $f:\reals \times \reals \to \reals^d \in \FFN(\relu,\Heavi)$.
    Then $f$ is eventually linear 
    in the first argument, where the offset can depend on the second argument but the slope does not.
 
    More precisely, for $f:(x,y)\mapsto \vec z$ 
    there exist $\{S_y \in \reals \,\mid\, y \in \reals\}, \vec a \in \reals^{d}$,
    $\{\vec b_y \in \reals^{d} \,\mid\, y \in \reals\}$ such that $f(x,y)=\vec a x + \vec b_{y}$ when $x \geq S_y$.   
\end{lemma}
\begin{proof}
    By induction on $f$. If $f$ is $I$, $S=0, \vec a = (1,0), \vec b_y=(0,y)$. If $g(x,y) \in \FFN(\relu, \Heavi)$ is eventually linear in the first argument, then this also holds for $A\,  g(x,y)+B$, where $A$ is a linear operator and $B$ a vector.     
    Suppose $g$ is eventually linear, there exist $\{S_y | y \in \reals\}, \vec a, \{\vec b_y | y \in \reals\}$ such that for $x \geq S_y$, $g(x,y)=\vec a x + \vec b_y$. Let $h$ apply $\relu$ to the $i$-th index.
    We show $h \circ g$ is eventually linear. Let $a_{i}, b_{y,i}$
    denote the $i$-th indices of $\vec a, \vec b_y$. There are four cases:
    \begin{align*}
    \begin{cases}
        g(S_y,y)_i \ge 0,\; a_i \ge 0, \text{ then } h(g(x,y))_i = g(x,y)_i \text{ when } x \ge S_y\\
        g(S_y,y)_i < 0,\; a_i > 0,
        \text{ then } h(g(x,y))_i = g(x,y)_i 
        \text{ when } x \ge S_y - \frac{g(S_y,y)_i}{a_i}\\
        g(S_y,y)_i < 0,\; a_i \le 0, \text{ then } h(g(x,y))_i = 0 \text{ when } x \ge S_y\\
        g(S_y,y)_i \ge 0,\; a_i < 0, \text{ then } h(g(x,y))_i = 0 \text{ when } x \ge S_y - \frac{g(S_y,y)_i}{a_i},
    \end{cases}
    \end{align*}
In all cases, $h \circ g$ is linear when $x$ is sufficiently large as required by these four cases. An analogous argument applies when $h$ applies $\Heavi$ to the $i'$th index.
\end{proof}

\begin{lemma}
\label{lem:collisions for ReLU FFNN}
      Let $f:\reals \times \reals \to \reals \in \FFN(\relu,\Heavi)$, and let $\Omega \subset \reals$ be finite. Then there exists $S_\Omega \in \reals$, $M_\Omega \in \reals$ and $\approx \in \{\leq, \geq\}$ such that for every $y \in \Omega$, one of the following holds:
      \begin{enumerate}
          \item $f$ is  constant over all inputs of the form $(x,y)$ with $x \geq S_\Omega$, or 
          \item for every $y' \in \Omega$ and all $x,x'\geq S_\Omega$ it holds that $x'\approx x+M$ implies $f(x,y)\leq f(x',y')$. 
      \end{enumerate}
      \end{lemma}
\begin{proof}
  By Lemma \ref{lem:Univariate semilinear is eventually_linear}, there exist $\{S_y \in \reals \mid y \in \reals\}, a \in \reals, \{b_y \in \reals | y \in \reals\}$ such that $x \geq S_y$ implies $f(x, y) = a\cdot x + b_{y}$.

    If $a=0$, $f$ is constant in $x$ when $x\geq S_y$. Suppose $a>0$, and given $y,y'$ let $M_{y,y'}=\frac{b_{y} - b_{y'}}{a}$. Then if $x,x'\geq \max(S_y,S_{y'})$ and $x' \geq x+M_{y,y'}$:
    \begin{align*}
        f(x', y') &= a \cdot x' + b_{y'}\\
        &\geq a \cdot (x+M_{y,y'}) + b_{y'}\\
        &= a \cdot x + b_{y} = f(x, y) 
    \end{align*}
    The same holds, if $a<0$ when $x' \leq x+M_{y,y'}$.

    Now let $S_\Omega = \max\{S_y \mid y \in \Omega\}$. If $a>0$ let $\approx$ be $\geq$ and let $M_\Omega= \max\{M_{y,y'} \mid y,y' \in \Omega\}$; If $a<0$ let $\approx$ be $\leq$ and let $M_\Omega = \min\{M_{y,y'} \mid y,y' \in \Omega\}$.
\end{proof}

\thmMaxContUpperBound*
\begin{proof}
    We proof the claim for \localmax \GNNs with $\FFN(\relu,\Heavi)$ combination functions. The theorem then follows by Lemma
    \ref{lem:unary semilinear = relu,H}. Given a \localmax--$\FFN(\relu,\Heavi)$ \GNN $\agnn$, we construct $\apgraph,\altpgraph$ inseparable by $\agnn$ such that $\apgraph \models \Diamond^{\geq 2}p$, $\altpgraph \not\models \Diamond^{\geq 2}p$. Let $l$ be the number of combination functions of $\agnn$, and let $\feature_\agnn$ be the feature expression corresponding to the output of $\agnn$. $\agraph$ has node $u$ with neighbors $u_0,u_1, \dots u_{l+1}$, and $\altgraph$ has node $v$ with neighbors $v_1, \dots v_{l+1}$. $G$ and $H$ have no further edges, $p$ is only true at $u_0 \, u_1 \, v_1$ and all other propositions are false at all nodes. We define keyings for $\agraph,\altgraph$ to be ``good'' if they satisfy the following:    
    \begin{enumerate}
        \item For every sub-expression $\feature$ of $\feature_\agnn$, $\sem{\feature}_G(u) = \sem{\feature}_H(v)$; 
        \item For every sub-expression $\feature$ of $\feature_\agnn$ and $i \geq 1$, $\sem{\feature}_G(u_i) = \sem{\feature}_H(v_i)$;
        \item For every sub-expression $\feature$ of $\feature_\agnn$ there exists $i \geq 1$ such that $\sem{\feature}_G(u_0) \leq \sem{\feature}_H(v_i)$.
    \end{enumerate}
    We construct a system of inequalities with the keys $k_u, k_{u,0}, \dots k_{u,l+1}, k_v, k_{v,1}, \dots k_{v,l+1}$ as variables such that every solution of this system is a ``good'' keying. If there exists a solution to this system the theorem follows by (1): with this keying $u$ and $v$ are not separated by $\agnn$.

    We first add constraints to ensure the three requirements are satisfied by the input features.
    \begin{align*}
        k_u&=k_v &\text{(for constraint $(1)$})\\
        \text{for $i \geq 1$:}\,\, k_{u,i}&=k_{v,i} &\text{(for constraint $(2)$})\\
        k_{u,0} & \leq k_{u,1} &\text{(for constraint $(3)$})
    \end{align*}
    We now iterate over $\agnn$ adding inequalities to ensure (1),(2),(3) are preserved during feature construction. Suppose the constraints hold for all features defined with at most $m$ combination functions
    and $m'$ applications of \localmax for some $l>m \geq 0, m'\geq 0$, 
    and let $\feature$ be constructed with \localmax or an \FFN(\relu, \Heavi) from such features.

    First, suppose $\feature=\localmax(\feature')$. By (1) and (2) every pair $(u,v)$ and $(u_i,v_i)$ with $i \geq 1$ shares the same $\feature'$ value and by (3) $\sem{\feature'}_G(u_0)\leq \sem{\feature'}_G(u_i)$ for some $i \geq 1$. Thus, applying \localmax gives the same value for $u,v$, as well as for each pair $u_i,v_i$ so that (1) and (2) are preserved. Further, by (1) $\localmax(\feature')$ is constant over all $u_i, v_j$ and (3) is preserved, so no new constraints need to be added.

    Secondly suppose $\feature$ is constructed with an $\FFN(\relu, \Heavi)$. This preserves (1) and (2). It remains to ensure (3) is preserved. We set requirements on the keys so that:
    \begin{align}
        \sem{\feature}_G(u_0) = \sem{\feature}_H(v_1) \text{ or }
        \sem{\feature}_G(u_0) &\leq \sem{\feature}_H(v_{m+2}) \tag{*}
    \end{align}
    Note that $\sem{\feature}_G(u_0), \sem{\feature}_H(v_1),\sem{\feature}_H(v_{m+2})$ are computed by a sequence of \FFNs~ and \localmax operations. Note further that by $(3)$, none of the outputs of these \localmax operations depend on features at $u_0$. Thus there exists a single $f \in$ \FFN(\relu, \Heavi) that computes the $\feature$ value for nodes $u_0, v_1$ and $v_{m+1}$ given their input features followed by the input features of all nodes in their graph ordered by node index with exception of $u_0$. %\michael{previous sentence tough to parse} 
    Letting $\vec x_{w}$ represent the input feature values for node $w$, we derive:
    \begin{align*}
       \sem{\feature}_G(u_0) &= f(\vec x_{u_0}, \vec x_u, \vec x_{u_1}, \dots \vec x_{u_{l+1}})\\
       \sem{\feature}_H(v_1) &= f(\vec x_{v_1}, \vec x_v, \vec x_{v_1}, \dots \vec x_{v_{l+1}})\\
       \sem{\feature}_H(v_{m+2}) &= f(\vec x_{v_{m+2}}, \vec x_v, \vec x_{v_1}, \dots \vec x_{v_{l+1}})   
    \end{align*}
    To see this, note that $f$ can apply $\FFNs$ locally at every node, and can apply \localmax operations by computing maxima and copying feature values. $f$ does not depend on the keying, but is fixed for $\agnn$.
    Now since $\vec x_u = \vec x_v$ and $\vec x_{u_i} = \vec x_{v_i}$ for all $i \geq 1$, and since all input features are constant over $G$ and $H$ except $\feature_p$ and $\feature_\val$ there exists another $\FFN$ $f'$ so that $(*)$ is satisfied if and only if:
    \begin{align*}
        f'(k_{u_0}, 1) = f'(k_{v_1}, 1) \text{ or }  f'(k_{u_0}, 1) \leq f'(k_{v_{m+2}}, 0)
        \end{align*}
    Using Lemma~\ref{lem:collisions for ReLU FFNN} there exist $S,M \in \reals$ and $\approx \in \{\leq, \geq\}$ so that for keys larger or equal than $S$, either $f'(k_{u_0}, 1) = f'(k_{v_{1}}, 1)$, or $k_{v_{m+2}} \approx k_{u_0}+M$ implies $f'(k_{u_0}, 1)\leq f'(k_{v_{m+2}}, 0)$. In both cases (3) is preserved, and we add constraints on $k_{u_0}, k_{v_1}, k_{v_{m+2}}$ accordingly.    

    Now taking stock of all constraints, there are $S_1, \dots, S_{l+1} \in \reals$ and $M_1, \dots, M_{l+1} \in \reals$ and $\approx_1, \dots, \approx_{l+1} \in \{\leq, \geq\}$ such that a  keying is good if the following system is satisfied:
    \begin{align*}
        k_u &= k_v\\
        \text{for }1 \leq i \leq l+1: k_{u_i} &= k_{v,i}\\
         \text{for }i \in \{0,1\}: k_{u_0} &\geq \max(S_1, \dots, S_{l+1})\\
        \text{for }1 \leq i \leq l+1: k_{v_i} &\geq S_i\\
        \text{for }1 \leq i \leq l+1: k_{v_i} &\approx_i k_{u_0}+M_i
    \end{align*}
    Here $S_1=0,M_1=0, \approx_1=\geq$ encode a strengthening of the input constraints, and $S_2, \dots S_{l+1}, M_2, \dots M_{l+1}, \approx_2, \dots \approx_{l+1}$ encode the constraints from the inductive step for each of the $l$ combination functions of $\agnn$.

    Let $k_{u,0} = \max(S_1, \dots, S_{l+1})+\max(|M_1|, \dots, |M_{l+1}|)+1$. Then for $i \geq 1$ we can choose keys $k_{u,i}=k_{v,i}$ above $\max(S_1, \dots, S_k)$ and $\approx_i k_{u,0} + M_i$ to satisfy all constraints, concluding the proof.
\end{proof}

\subsection{Proof of Theorem  \ref{thm:undecidableuddl}: Undecidability of $\uddl$}
\thmundecidable*

\begin{proof}
\begin{figure}
\[
\begin{tikzcd}[arrows={-}, column sep=tiny]
v^{Spy} \arrow[rd] \arrow[rrd] \arrow[rrrd] \arrow[rdd] \arrow[rddd] \\
& v_{11} \arrow[r] \arrow[d] & \cdot^{H_1} \arrow[r] & \cdot^{H_2} \arrow[r] & v_{12} \arrow[r] \arrow[d]  & \cdot^{H_1} \arrow[r] & \cdot^{H_2} \arrow[r]& v_{13} \arrow[d] \arrow[r] & \cdots \\
& \cdot^{V_1} \arrow[d] &&& \cdot^{V_1} \arrow[d] &&& \cdot^{V_1} \arrow[d] \\
& \cdot^{V_2} \arrow[d] &&& \cdot^{V_2} \arrow[d] &&& \cdot^{V_2} \arrow[d] \\
& v_{21} \arrow[r] \arrow[d] & \cdot^{H_1} \arrow[r] & \cdot^{H_2} \arrow[r] & v_{22} \arrow[r] \arrow[d]  & \cdot^{H_1} \arrow[r] & \cdot^{H_2} \arrow[r]  & v_{23} \arrow[d] \arrow[r] & \cdots \\
& \cdot^{V_1} \arrow[d] &&& \cdot^{V_1} \arrow[d] &&& \cdot^{V_1} \arrow[d] \\
& \cdot^{V_2} \arrow[d] &&& \cdot^{V_2} \arrow[d] &&& \cdot^{V_2} \arrow[d] \\
& v_{31} \arrow[r] \arrow[d]          & \cdot^{H_1} \arrow[r] & \cdot^{H_2} \arrow[r]  & v_{32} \arrow[r] \arrow[d]          & \cdot^{H_1} \arrow[r] & \cdot^{H_2} \arrow[r] & v_{33} \arrow[d] \arrow[r] & \cdots\\
& \vdots &&& \vdots &&& \vdots
\end{tikzcd}
\]
\caption{Grid-shaped graph used in the undecidability proof. We use unary predicates (i.e., binary node features) $Spy, H_1, H_2, V_1, V_2$, as well as others that are explained in  the proof. The spy node is connected to every other node (these edges are not all drawn, to avoid clutter).}
\label{fig:grid}
\end{figure}
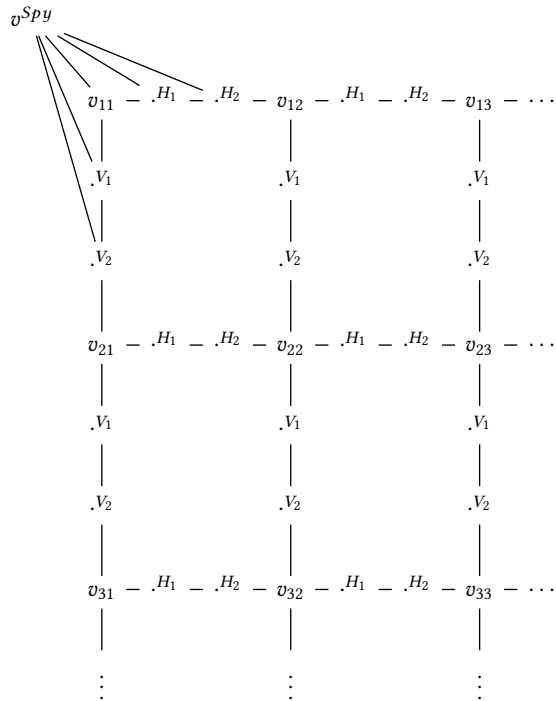
We prove this by reducing from the undecidable \emph{periodic tiling} problem \cite{periodictiling}. The input to this problem is
 a finite set $T$ of tile types and two binary relations $C_h, C_v\subseteq T\times T$ representing horizontal and vertical compatibility of tiles.
The problem is to decide whether there is a tiling  $t:\mathbb{N}\times\mathbb{N}\to T$ that respects the horizontal and vertical compatibility relations (i.e., $(t(x,y),t(x+1,y))\in C_h$ and $(t(x,y),t(x,y+1))\in C_h$ for all $x,y\in\mathbb{N}$)
and such that the tiling is periodic (i.e., there are $k,\ell>0$ such that, for all
$x,y\in\mathbb{N}$,
$t(x,y)=t(x+k)$ and $t(x,y)=t(x,y+\ell)$).

Let $T=\{t_1, \ldots, t_n\}$ and $C_h, C_v$ be given.
Figure~\ref{fig:grid} illustrates the encoding. The nodes that don't satisfy any of the unary predicates $Spy, H_i, V_i$  are the ``normal'' nodes that represented a grid point that is tiled, and they will satisfy exactly one of the unary predicates
$P_i$ for $i\leq n$ to indicate the tile type.

Let \[\textsf{hsucc} = \test(\bigvee_i P_i);\step;\test(H_1);\step;\test(H_2);\step;\test(\bigvee_i P_i)\]
and 
\[\textsf{vsucc} = \test(\bigvee_i P_i);\step;\test(V_1);\step;\test(V_2);\step;\test(\bigvee_i P_i)\]

Our formula $\phi$ will be the conjunction $\phi_1\land\phi_2\land\phi_3\land\phi_4$ where

\begin{itemize}
    \item $\phi_1 = \textsf{Spy} \land ( ~ [\step;\step]\langle \step\rangle \textsf{Spy} ~ ) \land \langle \stay\cup\step;\step\cup\step;\step;\step\rangle^{= 1}Spy$
    
``the spy point is unique  and has a direct edge to every other node (within the connected component)''

    The first conjunct states that the distinguished node satisfies spy.  The second states that every neighbor of a neighbor a spay node is also a neighbor of a spy node.
    Combined with the third conjunct, this will mean that there is a unique spy node within the connected component,  and that node will have a direct edge to every other node within the component.

    \item $\phi_2 = \langle\step\rangle(\bigvee_i P_i)$

    ``some node is tiled'' 

    \item $\phi_3 = \bigwedge_i [\step](P_i \to \big(\bigvee_{(t_i,t_j)\in C_h}\langle\textsf{hsucc}\rangle P_j \land \bigvee_{(t_i,t_j)\in C_v}\langle\textsf{vsucc}\rangle P_j\big))$

    ``every tiled node has successors satisfying the horizontal and vertical compatibility relations''

    \item $\phi_4 = [\step]((\bigvee_i P_i)\to\langle\textsf{hsucc};\textsf{vsucc}\cup \textsf{vsucc};\textsf{hsucc}\rangle^{= 1}\top)$

    ``$\textsf{hsucc}$ and $\textsf{vsucc}$  commute.'' 

    \item $\phi_5 = [\step] \bigwedge_{P \neq P'} \neg (P \wedge P')$, where $P,P'$ range over all distinct labels

    ``every node has at most one label''
\end{itemize}

Standard arguments show that $\phi$ is satisfiable (in the finite) if and only if a periodic tiling exists.
We sketch only one direction, assuming the $\uddl$ formula is satisfied in a finite pointed graph $\apgraph$, and arguing for a periodic tiling.
By $\phi_1$ the distinguished node $u$ is the unique node in the connected component
satisfying \textsf{Spy}.

By $\phi_2$ the distinguished node has a neighbor $v_{11}$, carrying one of the tiles.
By $\phi_3$ such a neighbor $v_{11}$ has a horizontal successor and vertical successor $3$-hop path, carrying
tiles that satisfy the appropriate compatibility relations:
call the targets of those paths $v_{12}$ and $v_{21}$, respectively. By $\phi_5$ and
the definitions of $\textsf{hsucc}$ and $\textsf{vsucc}$, these paths must lead to different nodes.
By $\phi_1$, $v_{12}$ and $v_{21}$ are also neighbors of a spy node,
which (by uniqueness) must be $u$.
Thus we can apply $\phi_3$ again, saying that $v_{12}$ and $v_{21}$ have
$3$-hop paths to compatible nodes. Note that $\phi_5$ and the definitions
of $\textsf{hsucc}$ and $\textsf{vsucc}$, will imply that these successor
nodes are not only different from each other, but different from the horizontal
and vertical predecessors.
Continuing indefinitely,  we get a tiling. 
But since our graph is finite, the tiling must be periodic.

\end{proof}
\subsection{Proof of Theorem \ref{thm:localmaxsemilinear recognizes up to isomorphism}: \localmaxsemilinear \GNNs express Connected Isomorphism Types}

\begin{lemma}
\label{lem:Get size and keys with local max}
    Let $N>0$. There exists a \localmaxsemilinear \GNN that computes $\feature_{N}, \feature_{k_1}, \dots \feature_{k_N}$, where for a connected pointed keyed graph $\apgraph_\key$, $\sem{\feature_{N}}_G(u)$ is $1$ if $\apgraph_\key$ has size $N$ and $0$ otherwise, and if $\apgraph$ has size $N$, $\sem{\feature_{k_1}}_G(u), \dots \sem{\feature_{k_N}}_G(u)$ are the keys of $\apgraph_\key$ in descending order.
\end{lemma}
\begin{proof}
We assume all keys are non-zero, which can be achieved with $\ifPos(-\val, \val, \val+1)$. We label each node with the smallest key in its $2N$ neighborhood (the subgraph induced by nodes at maximum distance $2N$):
 \begin{align*}
    \feature^0_{\min} &= \val\\
\text{for $1 \leq j \leq 2N$ }    \feature^j_{\min} &= \min \left(\feature^{j-1}_{\min}, -\localmin(-\feature^{j-1}_{\min})\right)\\
\feature_{\min} &= \feature^{2N}_{\min}
\end{align*}
where $\localmin(x) = -\localmax(-x)$, and  $\min(x,y)= x - \relu(x-y)$.
Now for $1 < i \leq N+1$ let $\feature_{\val,i}$ store keys that are smaller than $\feature_{k,i-1}$, we call these keys `active at step $i$'. Nodes with inactive keys store a value below the minimum key in their $2N$ neighborhood:
    \begin{align*}
        \feature_{\val,1} &= \val\\
        \feature_{\val,i} &= \ifPos(\feature_{k,i-1}-\val, \val, \feature_{\min}-1)
    \end{align*}
For $1 \leq i \leq N$, we compute the largest active key in its $2N$ neighborhood.
\begin{align*}
        \feature^0_{k,\max,i} &= \feature_{\val,i}\\
\text{for $1 \leq j \leq 2N$ }    \feature^j_{k,\max,i} &= \max \left(\feature^{j-1}_{k,\max,i},\localmax(\feature^{j-1}_{k,\max,i})\right)\\
\feature_{k_i} &= \feature^{2N}_{k,\max,i}
\end{align*}
Then at every step $i$ a key is deactivated if and only if it is the maximum of all active keys in its $2N$ neighborhood. Note that in the $N$ neighborhood of $u$ at most $1$ key is deactivated at each step (since two nodes at distance $\leq 2N$ from each other cannot both have the maximum active key). If $|V(\apgraph_\key)|\leq N$ exactly $1$ key is deactivated at each step (the largest active key). However, if $|V(\apgraph_\key)|>N$ it could be the case at some step that $0$ keys in the $N$ neighborhood of $u$ are deactivated. We thus check the following two conditions, which both hold if and only if $|V(\apgraph_\key)|=N$:
\begin{enumerate}
    \item A key in the $N$ neighborhood is deactivated at steps $1,\dots,N$.
    \item No key in the $N$ neighborhood is active at step $N+1$
\end{enumerate}
We represent both conditions as features. For $1 \leq i \leq N$ define a feature that is $1$ if a node in the $N$ neighborhood is deactivated at step $i$ and $0$ otherwise as:
\begin{align*}
\feature^0_{\text{deactivated},i} &= \ifPos(\feature_{\val,i}-\feature_{\val,i+1},1,0)\\
\text{for $1 \leq j \leq N$ }\feature^j_{\text{deactivated},i} &= \max\left(\feature^{j-1}_{\text{deactivated},i}, \localmax(\feature^{j-1}_{\text{deactivated},i})\right)
\end{align*}
Then define a feature that is $1$ if a node in the $N$ neighborhood is active at step $N+1$ and $0$ otherwise.
\begin{align*}
\feature^0_{\text{active},N+1} &= \ifPos(\val-\feature_{\val,N},0,1)\\
\text{for $1 \leq j \leq N$ }\feature^j_{\text{active},N+1} &= \max\left(\feature^{j-1}_{\text{active},N+1}, \localmax(\feature^{j-1}_{\text{active},N+1})\right)
\end{align*}
Now $\feature_N = \relu\left(\sum_{1 \leq i \leq N} \feature^N_{\text{deactivated},i} + (1-\feature^N_{\text{active},N}) - N\right)$ is $1$ if both conditions hold and $0$ otherwise. If $\sem{\feature^N}_G(u)=1$, then for $1 \leq i \leq N$, $\sem{\feature_{k_i}}_G(u)$ is the largest active key at step $i$, which is the $i$-th largest key of $\apgraph_\key$.
\end{proof}

\michael{This is pretty hard to read (and very notation-heavy); maybe I am missing something, but I am surprised it is this complicated. Also it would be nice to isolate exactly what is the significant GNN functionality required to do this.  I thought the idea was: we already know how to get all the keys in the graph, and we know the radius. Now what we need is a) an adjacency lookup GNN $G_r$  that takes
a pointed graph where the distinguished node is annotated with two (additional) key values, returning 1 if those two values are in 
radius $r$ and are adjacent, and zero otherwise; b) for each label p, a ``label lookup''' $GNN_p$ that takes a pointed graph with one additional key value, returning 1 if there is a node associated with that key value labeled with p}

\ThmMaxIsoTypes*
\begin{proof}
Let $N$ be the size of $\agraph$ and let $\feature_N, \feature_{k_1}, \dots \feature_{k_N}$ be as defined in Lemma \ref{lem:Get size and keys with local max}.
Given a pointed keyed graph $\altpgraph_\key$, note that for every $v' \in V(\altgraph)$, $\sem{\feature_N}_{\altgraph_\key}(v')=1$ if $\altgraph$ is of size $N$ and $0$ otherwise, and if $\altgraph$ is of size $N$, $\sem{\feature_{k_1}}_\altgraph(v'), \dots, \sem{\feature_{k_N}}_\altgraph(v')$ are the $N$ keys of $\altgraph$ in descending order. We construct a \GNN $\agnn$ that rejects $\altpgraph_\key$ if $\sem{\feature_N}_\altgraph(v)=0$, and that checks for each bijection between $V(\agraph)$ and $V(\altgraph)$ if it is an isomorphism from $\apgraph$ to $\altpgraph$.

Let $u_1, \dots u_N$ enumerate the nodes of $\agraph$ where $u_1$ is the distinguished node, and let $\pi \in S_N$ be a permutation of $(1,\dots,N)$. We check if the bijection that maps each $u_i \in V(\agraph)$ to the node in $\altgraph$ with the $\pi(i)'$th largest key is an isomorphism. For each $u_i$, $\feature^\pi_{u_i,\key} = \ifPos(|\feature_{k_{\pi(i)}}-\val|,0,1)$ checks if a node has the $\pi(i)$'th largest key. We define $\feature^\pi_{u_i}$ that is $1$ if a node has the $\pi(i)$'th largest key and the labeling and edges match that of $u_i$ and $0$ otherwise: 
\begin{align*}
    \feature_{u_i,\lab} &= \relu\left(\sum_{1 \leq l \leq D \,:\, \lab(u_i)_l=1} \left(p_l\right) + \sum_{1 \leq l \leq D \,:\, \lab(u_i)_l=0} \left(1-p_l\right) - (D - 1)\right)\\
    \feature^\pi_{u_i,\mathcal{N}} &= \relu\left(\sum_{u_j \in \mathcal{N}(u_i)} \left(\localmax(\feature^\pi_{u_j,\key})\right) + \sum_{u_j \in V(G)\setminus\mathcal{N}(u_i)} (1-\localmax(\feature^\pi_{u_j,\key}))-(N-1)\right)\\
    \feature^\pi_{u_i} &= \relu(\feature^\pi_{u_i,\key} + \feature_{u_i,\lab} + \feature^\pi_{u_i,\mathcal{N}} - 2)
    \end{align*}
    Here, $D$ is the labeling dimension, and $\feature_{u_i,\lab}$, $\feature^\pi_{u_i,\mathcal{N}}$ check if the labeling and edges respectively match that of $u_i$. Taking the local maximum $N$ times and checking that the designated nodes are related gives a feature that is $1$ at $v$ if the bijection defined by $\pi$ is an isomorphism between $\apgraph$ and $\altpgraph$ and $0$ otherwise, if $\altgraph$ is of size $N$. $\agnn$ checks if an isomorphism exists by summing over all $\pi \in S_N$, and outputs the following $\feature$:
    \begin{align*}
        \feature^\pi &= \relu(\sum^N_{i=1} \left(\localmax_N(\feature^\pi_{u_i})\right)+\feature^\pi_{u_1,\key}-N)\\
        \feature &= \ifPos(\feature_N, \sum_{\pi\in S_N}\feature^\pi, 0)
    \end{align*}
\end{proof}

\subsection{Proof of Theorem \ref{thm:orderinvariant}: $\localmax$ \GNNs are bounded by Order-invariant FO}

\thmorderinvariantfo*

\begin{proof} In the following argument, we assume that the distinguished node is not isolated: the isolated case can be detected in first-order logic, 
and dealt with separately. 
Fix $\localmax$ GNN $\agnn$ with arbitrary combination functions, and let $Q_\agnn$ be the corresponding node query.
We will translate $Q_\agnn$   to an equivalent order-invariant FO formula,  assuming that the distinguished node is not isolated.

We begin with the ``change in perspective'' mentioned in the proof sketch in the body
of the paper. Consider  the variant of keyed GNNs, where instead of
referring to the key, we simply take the nodes to be real numbers.
That is, there is a bijection that associates a keyed graph with a (finite) graph with vertices in the reals, identifying a node with its key.
Clearly any query, real-valued or Boolean-valued,  written in one data model can be applied to the other using this bijection. In particular, we can
interpret $Q_\agnn$ or any feature used to construct it in this model.

We explain the analog of first-order logic for finite labeled graphs with vertices in the reals.
Consider  the signature $S= \{G, C_1 \ldots C_k\}$ where $G$ is binary, $C_1 \ldots C_k$ are unary predicates (the labels in the GNN), and the signature $L_\combineclass$ of the combination functions. An embedded finite  graph is a  interpretation
of each $S$ relation by a finite set of tuples in the reals. The active domain of such a structure is the set of reals occurring in the interpretation of one of the $S$ predicates. An embedded finite pointed graph is an embedded finite graph along with a distinguished element of the active domain.
\emph{Active Domain First-Order Logic} (Active Domain FO) has the same syntax
as first-order logic over a signature decomposed into $S \cup L$, but where the quantified variables range only over the active domain.
We write the quantifiers as $\exists x \in \adom$, $\forall x \in \adom$ to emphasize this.
In particular, active domain formulas with one free variable define Boolean functions on embedded finite pointed graphs.

We first claim that any $\localmax$ GNN with combination functions in $\combineclass$  and any of our acceptance policies is equivalent, as a function of embedded finite pointed graphs, to an active domain FO formula over $S \cup L_\combineclass \cup \{<\}$, where $<$ is the usual ordering on
the reals. Note that in this argument we do not need that the GNN is key-invariant.

Inductively we need  to translate features. We cannot express a feature as an active domain formula, since it returns values that can be out of the active domain. We thus have a slightly more complex inductive invariant.

By an \emph{atomic type} we mean a specification of which labels $C_i$ hold of a real number in an embedded finite graph. In an embedded finite graph, every real in the active domain has an atomic type. We will sometimes abuse notation below and identify an atomic type with the conjunction of the atomic formulas it specifies.

We claim  that for each feature $\feature$ of our GNN,  and each atomic type $t$, there is a 
function $F^t(v_1 \ldots v_k)$ in the closure of our combination class under composition and an
active domain $S \cup L_\combineclass \cup \{<\}$ formula $\phi^t(v, w_1 \ldots w_k)$ such that the following hold in any embedded finite model:
\begin{itemize}
\item Each $\phi^t$ defines a  function. That is for every $v$ in the active domain, there is exactly one $\vec w$ such that $\phi^t(v, w_1 \ldots w_k)$. Letting $Q^t(v)$ denote the unique $\vec w$ such that  $\phi^t(v, \vec w)$, we refer to $\phi^t$ as the witness query formula and the $Q^t$ as \emph{witness query} for $t$ and $\feature$.
\item  for any $v$ in the active domain of the embedded finite model with atomic type $t$, the output of $\feature$ applied to $v$ is 
$F^t(Q^t(v))$. 
\end{itemize}
We refer to $F^t$ as the \emph{output function} for $t$, and $k$ as its \emph{dimension}. 

Note that if we can prove this inductive invariant for $\feature$, then it follows that given an acceptance policy, say $\feature>0$, 
we can produce an active domain formula that captures it: 
\[
\bigvee_t ~ t(v) \wedge \exists \vec w \in \adom ~ \phi^t(v,\vec w) \wedge (F^t(\vec w)>0)
\]
We prove this by induction on the structure of features.
For $\val$, regardless of the atomic type, we have only one witness query, the identity, and we apply the identity function to it.
For colors we use $1$ or $0$, depending on the atomic type, and have the identity as the witness query

Consider $F(\feature_1 \ldots \feature_k)$.
Fixing $t$, inductively , for each $\feature_j$ $j=1 \ldots k$, we have an output function $F^t_j(v^j_1 \ldots v^j_{o_j})$
of dimension $o_j$ and 
$\phi^t_j(v, \vec w_j)$, where $\vec w_j$ has arity $o_j$, inducing a witness query.
Our output function will have dimension the sum of the dimensions of each $F^t_j$. The function will map $\vec w_1 \ldots \vec w_j$ to $F(F^t_1(\vec w_1) \ldots F^t_k(\vec w_j))$.
Our witness query formulas will then just be the conjunction of $\phi^t_{j}(v, \vec w_j)$.

Now consider the case of $\feature= \localmax ~ \feature_0$.
Fixing $t$, we inductively have $F^t_0(v_1 \ldots v_k)$ and $\phi^{t,0}(v, \vec w)$.
We let our output function be just $F^t_0(v_1 \ldots v_k)$. We let $\phi^t(v, \vec w)$ hold
of $\vec w$ when it is the unique minimal tuple
in the lexicographic ordering on tuples such that:
there is  a neighbor $v'$ of $v$ such that $\phi^t_0(v', \vec w)$, and
$F^t_0(\vec w)$ is minimal among $F^t_0(\vec w')$ where $\vec w'$ is such that $\phi^t_0(v', \vec w')$ holds for some neighbor $v'$ of
$v$. 
Note that it is clear that the set of neighbors of $v$, and the set of $\vec w$ such that $\phi^t_0(v', \vec w)$ for such a neighbor $v'$ of $v$, is finite.
Thus there is always a minimal value within the set of values of $F^t_0(\vec w)$.  There may be several $\vec w$ that witness this minimum, but by breaking ties using the lexicographic ordering, we assure that we specify exactly one of these.
It is also straightforward to show that $\phi^t$ is expressible in active domain FO, using combination functions and the real order.

This completes the induction.

Now we are ready to use key-invariance of $Q_\agnn$.
This implies that the active domain formula $\phi$ is invariant under isomorphisms of the interpretation of $S$.
Then from the \emph{locally generic collapse}, see Corollary 1 of \cite{collapsejacm}, we can conclude that $\phi$ is equivalent (over finite interpretations of $S$ predicates in the reals)
to a first-order active domain formula  $\phi$ using only $S \cup \{<\}$. We note here that this step in fact requires less than key-invariance: it is enough if $\phi$ is invariant under order-preserving isomorphisms.

We now switch the perspective back, applying the inverse of the bijection from keyed graphs to embedded finite graphs. 
By interpreting $<$ as the ordering on keys, we can consider $\phi$ as a first-order sentence on keyed graphs.
Since our GNN was key-invariant, $\phi$ is order-invariant, as required.

\end{proof}

\section{Proofs for section \ref{sec:sum}}

\subsection{Proof of Theorem~\ref{thm:closed-under-coverings}: Combinatorial Upper Bound for \localsumcontinuous \GNNs}

The following lemma is derived by straightforward induction over feature expressions:
\begin{lemma}
\label{lem:valued coverings are GNN indistinguishable}
    Let $f$ be a covering from $\apgraph_\val$ to $\altpgraph_\val$ that preserves values and let $\agnn$ be a \localsum \GNN. Then $\agnn(\apgraph_\val) = \agnn(\altpgraph_\val)$.
\end{lemma}

\ThmSumClosedUnderCoverings*
\begin{proof}
     Let $\agnn$ be a \localsumcontinuous \GNN that decides $\calQ$. Let $\apgraph \in \calQ^c$ and $\apgraph \covers \altpgraph$ and assume $\altpgraph \not\in \calQ^c$. We derive a contradiction.

    Take a keyed extension $\altpgraph_\key$ and let $\apgraph_\val$ be the valued extension obtained by assigning to each node in $\apgraph$ the key of its image under the covering. $\apgraph_\val$ need not be keyed since the covering need not be injective. Since there is a covering from $\apgraph_\val$ to $\altpgraph_\key$ that preserves values, by Lemma~\ref{lem:valued coverings are GNN indistinguishable}
    both graphs obtain the same output value from $\agnn$. This must be a positive value since $\altpgraph \in \calQ$. By continuity, it follows that for sufficiently small $\epsilon$, $\agnn$ outputs a positive value for each $\apgraph_{\val'}$ where the values are preturbed by at most $\epsilon$. Among such $\apgraph_{\val'}$ we can find a keyed extension of $\apgraph$ that is accepted by $\agnn$, contradicting $\apgraph \in \calQ^c$.
\end{proof}

\subsection{Proofs of Theorems \ref{thm:>0/<0 localsumcontinuous collapses to key-oblivious}, \ref{thm:sum-isomorphism-collapse} and \ref{thm:localsumgapcollapse}: Upper Bounds for \localsumcontinuous}

Leighton~\cite{Leighton82} proved that any two graphs not
distinguished by color refinement have a common \emph{finite}
cover. We need the following version of this result, which easily
follows from Leighton's original proof,
and which can be seen as an analogue of Lemma~\ref{lem:fbisim both ways implies bisim} for color refinement.

\begin{theorem}[\cite{Leighton82}]\label{theo:leighton}
\label{thm:cr equivalence implies covering and inverse covering}
  Let $G,H$ be connected graphs and $v\in V(G),w\in V(H)$ such that
  $\colr_G(v)=\colr_H(w)$. Then there is a finite graph $F$, a
  covering map $g$ from $F$ to $G$, a covering map $h$ from $F$ to
  $H$, and a vertex $u\in V(F)$ such that $g(u)=v$ and $h(u)=w$;
\end{theorem}

Recall the statement of the Theorem:

\thmgreaterthanzeroslashlessthanzerolocalsumcontinuouscollapsestokeyoblivious*

\begin{proof}
    The proof is along the same lines as the proof of Theorem~\ref{thm:collapseyesnolocalmax}.
    By Theorem \ref{thm:localsumcontinuous expresses all local CR-invariant queries}, (3) and (4) are equivalent and they imply (1).
 The direction from (1) to (2) is immediate. Finally, we show that (2) implies (4). By Theorem \ref{thm:closed-under-coverings}, $\calQ$ and $\calQ^c$ are closed under coverings. By Theorem \ref{thm:cr equivalence implies covering and inverse covering},  $\calQ$ is CR-invariant. $\calQ$ is also strongly local since it is expressed by a key-invariant \localsumcontinuous \GNN. 
\end{proof}

\begin{definition}[Unravelling] \label{def:unravelling}
Given graph $\apgraph$ and $L\geq0$, the unravelling $U^L(\apgraph)$ is the graph with nodes $(u,u_1, \dots u_n)$ for every path $u,u_1 \dots u_n$ in $\agraph$ with $n\leq L$, of which $(u)$ is the distinguished node, edges $((u, u_1 \dots u_n),(u, u'_1 \dots u'_n,u'_{n+1}))$ whenever $u_i=u'_i$ for all $1 \leq i \leq n$ and labeling $\lab((u_1, \dots u_n)) = \lab_G(u_n)$.
\end{definition}

We use the following well-known property of this unravelling:
(see e.g. \cite{barceloetallogical}(Observation C.3), \cite{geerts2022expressiveness}): Graphs $\apgraph$, $\altpgraph$ are $L$-round-CR equivalent, that is, $CR^{(L)}(\apgraph)=CR^{(L)}(\altpgraph)$ if and only if $U^L(\apgraph)$ is isomorphic to $U^L(\altpgraph)$, where $U^L$ is
as in Definition \ref{def:unravelling}. 

\thmSumIsomorphismCollapse*
\begin{proof}
(2) to (1) is trivial. For (1) to (3), suppose $\apgraph$ has a cycle, i.e. a path $u_1, \dots u_{n}$ where $n\geq3$ and $(u_{n},u_1) \in E(\agraph)$. Construct $\altgraph$ by first taking $2$ disjoint copies of $\agraph$. If $v^1_1, \dots v^1_{n}$ and $v^2_1, \dots v^2_{n}$ are the two cycles obtained by copying $u_1, \dots u_n$, replace the edges $(v^1_{n},v^1_1)$ and $(v^2_{n},v^2_1)$ by the edges $(v^1_{n},v^2_1)$ and $(v^2_{n},v^1_1)$. It is easy to see that  $\altgraph$ is still connected.
Let $f: V(\altgraph) \to V(\agraph)$ map each of the two copies to the corresponding node in $\altgraph$, and let $v \in V(\altgraph)$ such that $f(v)=u$. Then $f$ is a covering from $\altpgraph$ to $\apgraph$. Hence, since $\altpgraph$ is not isomorphic to $\apgraph$, by Theorem \ref{thm:closed-under-coverings} $\calQ_{\apgraph}$ is not expressed by a key-invariant \localsumcontinuous \GNN. 

For (3) to (2), using Theorem \ref{thm:>0/<0 localsumcontinuous collapses to key-oblivious} it suffices to show that for every acyclic $\apgraph$, $\calQ_{\apgraph}$ is strongly local and CR-invariant. Clearly $\calQ_{\apgraph}$ is strongly local, where $r$ is the radius of $\apgraph$ plus $1$. To see $\calQ_{\apgraph}$ is CR-invariant, firstly suppose $\altpgraph$ is a tree and $CR(\apgraph)=CR(\altpgraph)$. By Lemma \ref{lem:same root color implies isomorphic pointed trees} $\apgraph$ and $\altpgraph$ are isomorphic. Conversely, suppose $\altpgraph$ has a cycle, we show $CR(\apgraph) \neq CR(\altpgraph)$.
Let $k$ be the length of the shortest path between $v$ and a node $v'$ in this cycle. For every $k'\geq 0$, in $U^{k+k'}(\altpgraph)$ there is a a node $(v, \dots v') \in V(U^{k+k'}(\altpgraph))$ that is itself the root of a tree that contains a perfect depth $k'$ binary tree as a subgraph, where every parent in this tree has a lower distance to the distinguished node $(v)$ than its children. Conversely, let $(u, \dots u_k)$ be a node in $U^{k+k'}(\apgraph)$ reachable by a path of length $k$ from $u$, and suppose $u_k$ is the root of a tree with a depth $k'$ perfect binary tree as subgraph, where every parent has lower distance to the distinguished $(u)$ than its children. Since $\agraph$ is acyclic, given node $(u, \dots, w)$ in this tree with children $(u, \dots, w, w_1), (u, \dots, w, w_2)$, if $w$ has distance $r$ to $u_k$ then at least one of $w_1,w_2$ has distance $r+1$. Hence $k'$ is at most the depth $r$ of $\apgraph$. Choosing $L>k+r$ then yields non-isomorphic unravellings $U^L(\apgraph)$ and $U^L(\altpgraph)$ so that $CR(\apgraph) \neq CR(\altpgraph)$.
\end{proof}

\thmlocalsumgapcollapse*
In fact, we can effectively transform a \localsumcontinuous GNN $\agnn$ to a key-oblivious one, in such a way that if the input GNN defines a key-invariant query under the policy, it is equivalent
to the output GNN.
We simply replace all references to keys in $\agnn$ by the constant $1$. Clearly this is key-oblivious.
Assume that $\agnn$ was key-invariant under the policy: we argue  for equivalence.  Consider a keyed pointed graph $\apgraph$ where the query associated
to $\agnn$ accepts. Let $\apgraph_n$ be  keyed graph
where all keys are transformed so that they converge to $1$ on every node. By key-invariance $\agnn$ has the same output on $\apgraph_n$ as on
$\apgraph$. Now let $\altpgraph$ be the valued graph where all the keys are replaced by $1$. By continuity and the fact that the acceptance condition
is closed, the output of $\agnn$ on $\apgraph_n$ converges to the output of $\altpgraph$, and thus the $\agnn$ run on $\altpgraph$ as a valued graph will accept.
But the output of $\agnn$ on the valued graph $\agraph'$ is the same as the output of our key-oblivious GNN on the underlying graph of $\agraph$.

We can reason analogously when $\agnn$ rejects on $\apgraph$, since the rejection condition is also closed.

\subsection{Proofs of Proposition~\ref{prop:sum continuous unique paths}, Proposition \ref{prop:sum semilinear unique paths} and Proposition \ref{prop:SumSemilinExpressesTriangleCoIso}: Unique Address Separating Queries}

The following lemma follows from a proof by Barcelo et al. \cite{barceloetallogical}(Proposition 4.1). The authors use truncated \relu activation $\min(\max(x,0),1)$, which equals $\relu(x) - \relu(x-1)$.
\begin{lemma}
\label{lem:GMLexpressedbySumReLU1/0}
    Let $\phi$ be a $\GML$ formula, then $\calQ_\phi$ is expressed by a key-oblivious \policy{\geq 1}{\leq 0} \localsumrelu \GNN.
\end{lemma}

\LocalSumContUniquePaths*
\begin{proof}
    One can express that there is a node uniquely addressable by $\phi_1, \dots \phi_n$ with a GML formula $\phi$:
    \begin{align*}
     \phi =& \Diamond(\phi_1 \wedge \Diamond(\phi_2 \wedge \dots \wedge \Diamond(\phi_n))) \, \wedge\\
     &\neg \left(\Diamond^{\geq 2}(\phi_1 \wedge \Diamond(\phi_2 \wedge \dots \wedge \Diamond(\phi_n))) \bigvee \Diamond(\phi_1 \wedge \Diamond^{\geq 2}(\phi_2 \wedge \dots \wedge \Diamond(\phi_n))) \, \dots \bigvee \Diamond(\phi_1 \wedge \Diamond(\phi_2 \wedge \dots \wedge \Diamond^{\geq 2}(\phi_n))\right)
    \end{align*}
    and similarly for $\psi_1, \dots \psi_m$ with a \GML formula $\psi$. Using Lemma \ref{lem:GMLexpressedbySumReLU1/0} there exists a feature $\feature_{\phi\wedge\psi}$ computable by a key-oblivious \localsumrelu \GNN such that $\sem{\feature_{\phi\wedge\psi}}_\agraph(u)$ is $1$ if $G^u \models \phi \wedge \psi$. We further use features $\feature_{\phi_i}$, $\feature_{\psi_i}$ that similarly express $\phi_i,\psi_i$. We retrieve the key at the end of the unique walk over $\phi_1, \dots \phi_n$:
    \begin{align*}
    \feature_{\phi_n,\key} &= \relu(\sigmoid(\val)+\feature_{\phi_n} - 1)\\
        \text{for $1 \leq i < n$ } \feature_{\phi_i,\key} &= \relu(\localsum(\feature_{\phi_{i+1},\key}) + \feature_{\phi_i}-1)\\
        \feature_{\phi,\key} &= \localsum(\feature_{\phi_1,\key})
    \end{align*}
    and similarly for the walk over $\psi_1, \dots \psi_m$. We then define a feature that is positive if and only if $\sem{\feature_{\phi \wedge \psi}}_\agraph(u)=1$ and these keys are different.
\begin{align*}    
    \feature &= \underbrace{|\feature_{\phi,\key}-\feature_{\psi,\key}|}_{\text{bounded in }(0,1)} + \feature_{\phi \wedge \psi} - 1
\end{align*}
\end{proof}

\LocalSumSemiUniquePaths*
\begin{proof}
The proof very closely matches that of Proposition~\ref{prop:sum continuous unique paths}, with a single difference. To let a node send its key if it satisfies $\phi_n$, and $0$ otherwise, we now use:
\begin{align*}    
    \feature_{\phi_n,\key} &= \ifPos(\feature_{\phi_{n}}, \val,0)
\end{align*}
Note that, unlike for key-invariant \localsumcontinuous \GNNs, the complement of this query is also expressible by applying $\ifPos(x,0,1)$ to the output feature.
\end{proof}

\propSumSemilinExpressesTriangleIso*
\begin{proof}
    Example \ref{ex:localsum} showed that $\calQ^c_{\apgraph_{\triangle p}}$ is expressed by a key-invariant \localsumsigmoid \GNN, where $\relu$ and $\sigmoid$ are used to let each node send its key if it satisfies $p$ and $0$ otherwise. We do the same with a semilinear function, as in the proof above:
    \begin{align*}
        \feature_\key &= \ifPos(\feature_p, \val, 0)
    \end{align*}
    Since $\calQ^c_{\apgraph_{\triangle p}}$ is expressed by a key-invariant \localsumsemilinear \GNN, the same holds for $\calQ_{\apgraph_{\triangle p}}$.
\end{proof}

\subsection{Proof of Theorem~\ref{thm:SumReluinOrdInvFOC}: \localsumrelu \GNNs are bounded by Order-invariant FO+C}
\ThmSumReluinOrdInvFOC*
\begin{proof}
  As noted in the body, this follows from \cite[Theorem~5.1,
  Remark~5.3]{grohedescriptivegnn}. The precise argument comes with a
  considerable notational overhead, which we avoid here at the price
  not going into all details. Also, to simplify the presentation, we will consider keyed graphs without node labels (i.e., $\Pi=\emptyset$), but the argument extends immediately to the general case.

In \cite{grohedescriptivegnn}, valued
  graphs (with rational values) are modeled as 2-sorted structures,
  where the first sort is the vertex set of a graph and the second
  sort is the set $\mathbb N$ of nonegative integers. The edge
  relation $E$ is a binary
  relation on the vertex set, and the values are described by
  additional relations between vertices and natural numbers. Here, we
  only need to consider valued graphs using only nonnegative integers
  as values; let us call them $\mathbb N$-valued graphs. We can describe
  such $\mathbb N$-valued graphs using a single unary function symbol $f$ from vertices to numbers. \cite[Theorem~5.1]{grohedescriptivegnn} implies that for every
  \localsumrelu \GNN\ $\agnn$ there is a tuple $\FOC$-terms
  $\boldsymbol\theta_\agnn(x)$, in the format used to describe rational numbers,
  \footnote{According to the definition of FO+C we gave in the paper, terms denote natural numbers. However, the logic can be equivalently defined to have terms denoting rational numbers (one can think of such terms as pairs of terms), which is what is assumed here.}
  such that for every $\mathbb N$-valued pointed graph $\apgraph_\val$
  it holds that $\agnn(\apgraph_\val)$ is precisely the rational
  number described by $\boldsymbol\theta_\agnn(u)$ in $(G,\val)$. Actually,
  \cite[Theorem~5.1]{grohedescriptivegnn} only states that the value
  of the terms approximates $\agnn(\apgraph_\val)$, but then
  \cite[Remark~5.3]{grohedescriptivegnn} points out that for
  \localsumrelu \GNNs no approximation is needed and the value is
  exact.
  Using the terms $\boldsymbol\theta_\agnn(x)$, we can define an
  $\FOC$ formula $\phi_\agnn(x)$ such that $\agnn(\apgraph_\val)>0$ if
  and only if $(\apgraph,\val)\models\phi_\agnn(u)$. The
  vocabulary of this formula is $\{E,f\}$, where $E$ is the binary
  edge relation symbol and $f$ is the unary function symbol for the
  value function.

  Now suppose that $\agnn$ is key-invariant. For a pointed graph
  $\apgraph$ of order $n$, we only consider keyings
  $k:V(G)\to\{1,\ldots,n\}$; let us call such keyings
  \emph{numberings}. As $\agnn$ is key-invariant, for every pointed
  graph $\apgraph$ and every two numberings $\num,\num'$ of
  $\apgraph$ it holds that
  $\agnn(\apgraph_\num)=\agnn(\apgraph_{\num'})$ and hence
  \begin{equation}
    \label{eq:1}
    (\apgraph,\num)\models\phi_\agnn(u)\iff (\apgraph,{\num'})\models\phi_\agnn(u) 
  \end{equation}
  This means that the formula $\phi_\agnn$ is invariant under
  numberings. Moreover, $\phi_\agnn$ defines the query $\calQ_\agnn$, in
  the sense that for every pointed graph $\agnn$ and every numbering
  $\num$ of $\apgraph$ it holds that
  \begin{equation}
    \label{eq:2}
    \apgraph\in\calQ_\agnn\iff \apgraph_\num\models\phi_\agnn(u). 
  \end{equation}
  The last observation we need to make is that numberings and linear
  orders of a graph $G$ are in one-to-one correspondence and that this
  correspondence is definable in $\FOC$. For a linear order $\le$ of
  $V(G)$, we let $\num_\le:V(G)\to\{1,\ldots,n\}$, where $n\coloneqq
  |V(G)|$, be the numbering mapping $v\in V(G)$ to the number of $w\in
  V(G)$ such that $w\le v$. Let
  \[
    \theta(x)\coloneqq\#(x').x'\le x,
  \]
  viewed as an $\FOC$ term of vocabulary $\{E,\le\}$. Then for all
  ordered graphs $(G,\le)$ and all vertices $v\in V(G)$, the value of
  $\theta(v)$ in $(G,\le)$ is $\num_\le(v)$.

  Now let $\psi(x)$ be the
  formula obtained from $\phi(x)$ by replacing each subterm $f(z)$ by
  $\theta(z)$. Then $\psi(z)$ is a formula of vocabulary $\{E,\le\}$ and for every pointed graph $\apgraph$ and every
  linear order $\le$ on $V(G)$ it holds that
  \[
    (G,\le)\models\psi(u)\iff(G,\num_\le)\models\phi(u).
  \]
  Then it follows from \eqref{eq:1} that $\psi(x)$ is order-invariant,
  and it follows from \eqref{eq:2} that $\psi$ defines $\calQ_\agnn$.
\end{proof}

\subsection{Proof of Theorem \ref{thm:localsumarbunlimited}: Expressive Completeness of \localsum}

\thmlocalsumarbunlimited*
\begin{proof}
Fix  a $Q$ that is strongly local with radius $r$.

Doing a summation over a multiset of reals, like the keys of neighbors,  will
lose information, mapping many multisets to the same real number. The main idea, deriving
from \cite{amir2023neural}, is that  we can preserve information about all of our neighbor's keys by pre-processing the keys (similar as what we did in the proof of Theorem~\ref{thm:localsumcontinuous expresses all local CR-invariant queries}).

For a function $f$ from reals to reals, we let $\Sigma_f$ be the function on $\multisets{\reals}$ 
that
maps $\multiset{r_1 \ldots r_j}$ to $\Sigma_i f(r_i)$. 
That is, it first applies $f$ elementwise, and then does a sum.
A function $f$  over the reals is then moment-injective if the derived function $\Sigma_f$ is injective.
We note that
\begin{claim} \label{clm:injective} There is a function $f$ that is moment-injective over the reals 
\end{claim} 

We explain how the claim implies the theorem. We first explain in the case where there
are no label predicates.
Our  GNN first applies $f$ and then does a local sum. Thus each node stores a feature $\feature_1$ encoding the keys of its neighbors.
We then apply a pairing function -- any injective function from $\reals^2$ to $\reals$ and apply it to the key of a node along with the output of $\feature_1$: now each node stores a real number encoding its own key and the keys of its neighbors.
We can repeat the process to get a new feature $\feature_2$:  at each node this will store an encoding of the multiset of pairs, where each pair is a key and a $\feature_1$ values.
Iterating this way $r$ times, we arrive at a feature which stores an encoding of the $k$-neighborhood of a node.
We then apply a function $\textbf{decode}_Q$ which holds true on such an encoding if the corresponding neighborhood satisfies $Q$.
In the presence of labels, we proceed in the same way, except we first apply a pairing function
to store at each node a pair consisting of its key and a binary representation of its labels.

The proof of the claim
 follows from the fact \cite{hamel}
 that the reals, seen as a vector space  over the rationals has a Hamel basis: a set such
 that every real can be uniquely decomposed as a linear combination (with rational coefficients) of elements of the set. This basis must be of the same cardinality as the reals, and thus
 there is a bijection $f$ from the reals to the basis.

It is easily seen that  $\Sigma_f$ is injective.
Consider distinct multisets $\multiset{r_1 \ldots r_k}$ $\multiset{r'_1 \ldots r'_n}$, possibly of different sizes, and suppose the corresponding sums $\Sigma_i f(r_i)$ and $\Sigma_j f(r_j)$ are equal.
Then looking at the difference, we get a nontrivial sum of $f(r_i)$'s with integer coefficients that is zero, contradicting the property of a basis.
\end{proof}

The standard construction of a Hamel basis is simple, but proceeds by ordinal induction. Thus it does not give us a function in a natural class.
It follows from results in \cite{limitinjective}  that it is impossible to use a continuous function for this task.

\subsection{Proof of Proposition \ref{prop:sumarbitrarybeyondcontandsemilin}: separating Arbitrary Functions from Continuous and Semilinear Functions for \localsum}
\propSumArbitraryBeyondContandSemilin*
\begin{proof}
    For each set of natural numbers $S \subset \mathbb{N}\setminus\{2\}$ let $\calQ_S$ be the query that contains pointed graph $\apgraph$ if and only if $u$ has $S$ neighbors, or $u$ has $2$ neighbors and is in a triangle. $\calQ^c_S$ is not closed under coverings for any $S$, since $C^u_6 \not\in \calQ_S$ and $C^u_3 \in \calQ_S$ but $C^u_6\covers C^u_3$ (where $u$ is any node in the cycle, and with uniform labeling). Since there are continuum many queries $\calQ^c_S$ but only countably many queries expressed by \localsumsemilinear \GNNs the proposition follows.
\end{proof}

\subsection{Proof of Theorem \ref{thm:iso test with local sum}: \localsumsemilinear with $(\cdot)^2$ express Connected Isomorphism Types}
We show in three lemmas that \localsum \GNNs with semilinear combination functions and squaring that for a given $N \in \mathbb{N}$ can decide if a graph has size $N$, and compute the $N$ keys if this is the case. We first show that these functions are sufficient to distinguish bounded size multisets with at most $1$ non-$0$ value. This is the only place where the squaring operation is used.
\begin{lemma}
\label{lem:recognize multiple nonzeros in multiset using semilinear}
Let $k\in \mathbb{N}$, and let $F$ be the composition closed class containing the semilinear functions and the squaring function $(\cdot)^2$. There exist functions $g,f_1,f_2,f_3: \reals \to \reals \in F$ such that given $m \in \mathcal{M}(\reals)$ of size at most $k$, the function:
\begin{align*}
    g\left(\sum_{x \in m} f_1(x), \sum_{x \in m} f_2(x), \sum_{x \in m} f_3(x)\right) \tag{*}
\end{align*} 
outputs $0$ if $|\{x \in m | x \neq 0\}| \leq 1$ and $1$ otherwise.
\end{lemma}
\begin{proof}
    Let $f_1(x)= \ifPos(|x|,1,0), f_2(x)=x, f_3(x)=x^2$. 
    Here $f_1$ outputs $1$ for each non-zero value in $m$, $f_2$ outputs the values in $m$ and $f_3$ outputs their squares. Let:
\begin{align*}
    h(a,b,c) &= \sum_{1 \leq i \leq k} \left(\ifPos(1-|a-i|, i\cdot c-b^2,0\right) = \begin{cases}
        a\cdot c-b^2 \text{ if } 1\leq a\leq k\\
        0 \text{ otherwise }
    \end{cases}\\
    g(a,b,c) &= \ifPos(h(a,b,c),1,0) 
\end{align*}
We show that the function (*) has the desired property. If $|\{x \in m | x \neq 0\}|=0$, then $\sum_{x \in m}f_1(x) = 0$ so that $h$ outputs $0$. Conversely, if $|\{x \in m | x \neq 0\}|>0$, let $y_1, \dots y_{l}$ be the non-zero elements of $m$ for some $l\geq1$. We show $g$ outputs $0$ if they are all equal and $1$ otherwise. By Cauchy-Schwartz:
\begin{align*}
    \sum^l_{i=1} 1^2 \cdot \sum^l_{i=1} y_i^2 &\geq (\sum^{l}_{i=1} y_i)^2 
\end{align*}
with equality if and only if all $y_i$ are equal. Hence, using that $1\leq \sum_{x \in m} f_1(x) \leq k$:
\begin{align*}
(*)&= \ifPos\left((\sum_{x \in m} f_1(x)) \cdot (\sum_{x \in m}f_3(x))-(\sum_{x \in m} f_2(x))^2,1,0\right)\\
&= \ifPos\left( l\cdot(\sum^l_{i=1} y_i^2)-(\sum^{l}_{i=1} y_i)^2,1,0\right)\\
&= \begin{cases}
    0 \text{ if all $y_i$ are equal }\\
    1 \text{ otherwise}
\end{cases}
\end{align*}
\end{proof}

\begin{lemma}
\label{lem:Local Sum recognizes unique local value}
    Ket $k,r \in \mathbb{N}$ and let $\feature_\agnn$ describe a \GNN $\agnn$. There exists a \localsum \GNN node-classifier $\altgnn$ with semilinear combination functions and squaring that accepts a pointed keyed graph $\apgraph_\key$ of degree at most $k$ if and only if $|\{\sem{\feature_\agnn}_\agraph(u') \mid u' \in V(\apgraph \rneighborhood)\}\setminus\{0\}|\leq 1$.
\end{lemma}
\begin{proof}
$\altgnn$ computes and sends the local average of $\feature_\agnn$ $r$ times. For $0 \leq i < r$:
\begin{align*}
    \feature_0 &= \feature_\agnn\\
    \feature_{i,\text{local avg}} &= \sum^k_{j=1}\ifPos(1-|j-\localsum(\ifPos(|\feature_i|,1,0))|, \localsum(\feature_i)/j,0)\\
    \feature_{i+1} &= \ifPos(\feature_{i,\text{local avg}}, \feature_{i,\text{local avg}}, \feature_{i})
\end{align*}
Here $\feature_{i,\text{local avg}}$ computes the average over the non-zero $\feature_i$ values of neighbors. $\feature_{i+1}$ is set to this average, or to $\feature_{i}$ if the average is $0$. In this process, using Lemma~\ref{lem:recognize multiple nonzeros in multiset using semilinear}, we raise a flag whenever a node has multiple non-zero values in its local neighborhood. For $0 \leq i < r$ let $\feature_{\text{multiple values},i}$ be $1$ if there are multiple non-zero $\feature_i$ values at the neighbors and $0$ otherwise, and:
\begin{align*}
    \feature_{\text{flag},0} =& \,\, 0\\
    \feature_{\text{flag},i+1} =& \,\, \feature_{\text{flag}, i}+\localsum(\feature_{\text{flag}, i}) + \feature_{\text{multiple values},i}\\
    &+\ifPos(\min(|\feature_{i+1}|,|\feature_{i}|), |\feature_{i+1}-\feature_{i}|,0)
\end{align*}
Then $\feature_{\text{flag},i}$ is positive if and only if there are multiple non-zero $\feature$ values in the $i$ neighborhood. $\altgnn$ has output feature $\feature_{\text{flag},r}$.
\end{proof}

\begin{lemma}
\label{lem:Get size and keys with local sum}
    Let $N>0$. There exists a \localsum \GNN with semilinear combination functions and squaring that computes $\feature_{N}, \feature_{k_1}, \dots \feature_{k_N}$, where for a connected pointed keyed graph $\apgraph_\key$, $\sem{\feature_{N}}_\agraph(u)$ is $1$ if $\apgraph_\key$ has size $N$ and $0$ otherwise, and if $\apgraph$ has size $N$, $\sem{\feature_{k_1}}_\agraph(u), \dots \sem{\feature_{k_N}}_\agraph(u)$ are the keys of $\apgraph_\key$ in descending order.
\end{lemma}
\begin{proof}
    We assume all keys are non-zero, which can be achieved with $\ifPos(-\val, \val, \val+1)$. Now, for $1 \leq j \leq 3N$:
    \begin{align*}    \feature^0_{\text{flag,degree}}&= \relu(\localsum(1)-N)\\
    \feature^j_{\text{flag, degree}} &= \feature^{j-1}_{\text{flag, degree}} + \localsum(\feature^{j-1}_{\text{flag, degree}})\\
    \feature_{\text{flag, degree}} &= \feature^{3N}_{\text{flag, degree}} 
    \end{align*}
    $\feature_{\text{flag, degree}}$ is $0$ if all nodes reachable in $3N$ steps have degree bounded by $N$, and at least $1$ otherwise. For $1 < i \leq N+1$ let $\feature_{\val,i}$ store keys that are smaller than $\feature_{k,i-1}$, we call these keys `active at step $i$':
    \begin{align*}
        \feature_{\val,1} &= \val\\
        \feature_{\val,i} &= \ifPos(\feature_{k,i-1}-\val, \val, 0)
    \end{align*}    
    Given the active keys at step $i$ we perform $N$ iterations of computing the average of active keys reachable in at most $2N$ steps (weighted by reaching walks) and setting keys smaller than the average to $0$, with the aim of finding the largest active key. For $1 \leq j \leq N$:
    \begin{align*}
     \feature^1_{\val,i} &= \feature_{\val,i} \\
     \feature^j_{\val,i, \text{sum}} &= \sum^{2N}_{l=1}\left(\localsum_l(\feature^j_{\val,i})\right)\\
     \feature^j_{\val,i, \text{weight}} &= \sum^{2N}_{l=1}\left(\localsum_l(\ifPos(|\feature^j_{\val,i}|,1,0))\right)\\
     \feature^j_{\val,i, \text{avg}} &= \sum^{{(2N)}^{N+1}}_{l=1}\ifPos(1-|l-\feature^j_{\val,i, \text{weight}}|, \feature^j_{\val,i, \text{sum}}/l, 0)\\
    \feature^{j+1}_{\val,i} &= \ifPos(\feature^{j}_{\val,i,\text{avg}}-\feature_{\val,i},0, \feature_{\val,i})    \end{align*}
    Here $\feature^j_{\val,i,\text{avg}}$ computes the average over active keys within distance at most $2N$, provided that there are at most $(2N)^{N+1}$ reached keys, which holds if the degree is bounded by $N$. Let $\feature_{k,i} = \feature^N_{\val,i, \text{avg}}$. If $\apgraph_\key$ has size at most $N$ this is the average over a single active key, which is the $i'$th largest key in $\apgraph_\key$.
    
    However, if $\apgraph_\key$ is larger than $N$ it could be the case that $\feature^N_{\val,i}$ is $0$ everywhere in the $N$ neighborhood of $u$, or non-$0$ at more than one node in the $N$ neighborhood, or not the largest active key in the $N$ neighborhood. We raise a flag in all these cases:
    \begin{align*}
    \feature_{\text{flag, all-}0, i} &=\ifPos(\feature^N_{\val,i, \text{weight}},0,1)\\
    \feature_{\text{flag, multiple non-}0, i} &= \begin{cases}
            1 \text{ if there is more than 1 non-zero $\feature^N_{\val,i}$ value in the $3N$ neighborhood of $\apgraph$}\\
            0 \text{ otherwise}
        \end{cases}\\
    \feature^0_{\text{flag}, \text{no max}, i} &= \ifPos(\min(|\feature_{\val,i}|,\feature_{\val,i}-\feature^N_{\val,i, \text{avg}}), 1, 0) + \ifPos(|\feature^N_{\val,i, \text{avg}}|,0,1)\\
    \text{for $1 \leq j \leq N$ }\feature^j_{\text{flag}, \text{ no max}, i} &= \feature^{j-1}_{\text{flag}, \text{ no max}, i}\\
  \feature_{\text{flag}, \text{ no max}, i} &= \feature^N_{\text{flag}, \text{ no max}, i}
    \end{align*}
    Here, Lemma \ref{lem:Local Sum recognizes unique local value} is used to construct $\feature_{\text{flag, multiple non-}0, i}$. If the first two flags are $0$ at node $u$, this ensures that the corresponding weighted average $\feature^N_{\val,i, \text{avg}}$ has a single non-zero value in the $N$ neighborhood. Then if $\sem{\feature^N_{\text{flag}, \text{ no max}, i}}_\agraph(u)=0$, this ensures this is also the \emph{largest} active key.

    If these flag features are all $0$, then $\feature_{k,i}$ is the $i'$th largest key in the $N$ neighborhood as intended. We let $\feature_N$ be $1$ if this is true at every step $i\leq N$, and if further $\feature_{\text{flag, degree}}=0$, the $N$'th largest key $\feature_{k_N}$ is non-zero and no key in the $N$ neighborhood is active at step $N+1$. If one of these checks fails $\sem{\feature_N}_G(u)=0$. 
\end{proof}

\ThmSumSemilinearSquareIsotypes*
\begin{proof}
    The proof is very similar to that of Theorem ~\ref{thm:localmaxsemilinear recognizes up to isomorphism}. Now, $\feature_N, \feature_{k_1},\dots\feature_{k_N}$ are as defined in \ref{lem:Get size and keys with local sum}, and for each permutation $\pi \in S_N$, and each $u_i \in V(G)$, $\localsum$ is used instead of $\localmax$ to define $\feature^\pi_{u_i}$ which is $1$ if a node has the $\pi(i)$'th largest key and a local neighborhood isomorphic to that of $u_i$ with the isomorphism defined by $\pi$, and $0$ otherwise. Finally, $\agnn$ computes:
    \begin{align*}
        \feature^\pi &= \relu(\sum^N_{i=1} \max(\localsum_N(\feature^\pi_{u_i}),1)+\feature^\pi_{u_1,\key}-N)\\
        \feature &= \ifPos(\feature_N, \sum_{\pi\in S_N}\feature^\pi, 0)
    \end{align*}
\end{proof}

\section{Proofs for Section~\ref{sec:discussion}}

\subsection{Proof of Theorem \ref{thm:epsilonapart}: $\epsilon$-apart Keys}

By a \emph{discrete coloring} of a graph, we will mean a coloring such that each node is assigned a different color.

\begin{lemma}
\label{lem:cr equiv is isomorphism on individualized graphs}
    Let $\apgraph, \altpgraph$ be pointed graphs and let $\col_G, \col_H$ be discrete node colorings that distinguish labels (i.e. if $\lab_G(u_1)\neq\lab_H(v_1)$) then $\col_G(u_1)\neq\col_H(v_1)$). Suppose $\refine(\agraph,\col_G,r+1)(u) = \refine(\altgraph,\col_H,r+1)(v)$, where $\refine$ is as in Subsection \ref{subsec:colorrefbisimulation}. Then $\apgraph \rneighborhood$ and $\altpgraph \rneighborhood$ are isomorphic. 
\end{lemma}
\begin{proof}
    $\refine(\agraph,\col_G,1)$ and $\refine(\altgraph,\col_H,1)$ map onto the same colors within the $r$ neighborhoods. To see this, suppose that for some $u_1$ at distance $\leq r$ from $u$, $\refine(\agraph,\col_G,1)(u_1)$ is not in the image of $\refine(\altgraph,\col_H,1)$ restricted to the $r$ neighborhood of $v$. Then $\refine(\agraph,\col_G,r+1)(u) \neq \refine(\altgraph,\col_H,r+1)(v)$. The same argument applies in the other direction.

    Hence, let $f$ be the bijection between the $r$ neighborhoods of $u$ and $v$ where $f(u_1) = v_1$ if and only if $\refine(\agraph,\col_G,1)(u_1)=\refine(\altgraph,\col_H,1)(v_1)$. $f$ preserves the initial colors, and hence the labels in both directions. In addition, if $(u_1,u_2) \in E_G$ where both $u_1$ and $u_2$ are in the $r$ neighborhood of $u$, then $f(u_1)$ is connected to a node $v_2$ with the same initial color as $u_2$. Hence $f(u_2)=v_2$. The same argument shows that $f$ preserves edges in the other direction. Thus, $f$ is an isomorphism between $\apgraph \rneighborhood$ and $\altpgraph \rneighborhood$.
\end{proof}
\epsilonapart*
\begin{proof}
    Let $K\subset \mathbb{R}$ be the space of keys where all elements are at least $\epsilon$ apart for some $\epsilon >0$. Let $f: \Pi \times K \to \mathbb{N}$ be a bijection from labels and keys to colors. Let $\calQ$ be a strongly local query with radius $r$. By Lemma \ref{lem:cr equiv is isomorphism on individualized graphs}, $r+1$ iterations of color refinement separate discretely colored graphs in $\calQ$ from discretely colored graphs not in $\calQ$. Thus, there  is a function $g: \mathbb{N} \to \{0,1\}$
    such that on keyed pointed graph $\apgraph_\key$, $g(\refine(G,f(\lab_{G},\key),r+1)(u))=1$ if and only if $G^{u} \in \calQ$.  

    Since $f$ and $g$ can be extended to continuous functions from $\reals$ to $\reals$, and since by the proof of Proposition \ref{prop:Oblivious continuous sum implements WL}, $r+1$ rounds of color refinement are expressed with a continuous feature, there exists a key-invariant \localsumcontinuous \GNN that expresses $\calQ$.
\end{proof}

\subsection{Proof of Theorem \ref{thm:iso test with global sum}: \globallocalsum--(Semilinear) \GNNs express All Isomorphism Types}

Lemma \ref{lem:Get size and keys with local sum} showed that \localsum \GNNs with semilinear combination functions and squaring can check if the size of a connected graph is $N$ and store all the keys in descending order at each node. We now show that \globallocalsum \GNNs with only semilinear combination functions can do the same, also for non-connected graphs.
\begin{lemma}
\label{lem:Get size and keys with global sum}
    Let $N>0$. There exists a \globalsumsemilinear \GNN that computes $\feature_{N}, \feature_{k_1}, \dots \feature_{k_N}$, where for any pointed keyed graph $\apgraph_\key$, $\sem{\feature_{N}}_\agraph(u)$ is $1$ if $\apgraph_\key$ has size $N$ and $0$ otherwise, and if $\apgraph$ has size $N$, $\sem{\feature_{k_1}}_\agraph(u), \dots \sem{\feature_{k_N}}_\agraph(u)$ are the keys of $\apgraph_\key$ in descending order.
\end{lemma}
\begin{proof}
    We assume all keys are non-zero, which can be achieved with $\ifPos(-\val, \val, \val+1)$. Let $\feature_N = \relu(1-|\globalsum(1)-N|)$. For $1 < i \leq N$ let $\feature_{\val,i}$ store keys that are smaller than $\feature_{k,i-1}$, we call these keys `active at step $i$':
    \begin{align*}
        \feature_{\val,1} &= \val\\
        \feature_{\val,i} &= \ifPos(\feature_{k,i-1}-\val, \val, 0)
    \end{align*}
    Given the active keys at step $i$ we perform $N$ iterations of computing the average of active keys and setting keys smaller than this average to $0$ to obtain the largest active key.     \begin{align*}
     \feature^1_{\val,i} &= \feature_{\val,i} \\
     \feature^j_{\val,i, \text{sum}} &= \globalsum(\feature^j_{\val,i}))\\
     \feature^j_{\val,i, \text{weight}} &= \globalsum(\ifPos(|\feature^j_{\val,i}|,1,0)))\\
     \feature^j_{\val,i, \text{avg}} &= \sum^{N}_{l=1}\ifPos(1-|l-\feature^j_{\val,i, \text{weight}}|, \feature^j_{\val,i, \text{sum}}/l, 0)\\
    \feature^j_{\val,i} &= \ifPos(\feature^{j-1}_{\val,i,\text{avg}}-\feature_{\val,i},0, \feature_{\val,i})
    \end{align*}
Then let $\feature_{k_i} = \feature^N_{\val,i}$. $\sem{\feature_{k_i}}(u)$ is the $i'$th largest key.
\end{proof}

\ThmIsoWithGlobalSum*
\begin{proof}
This proof closely follows that of Theorem \ref{thm:iso test with local sum}, now using Lemma \ref{lem:Get size and keys with global sum} instead of Lemma \ref{lem:Get size and keys with local sum}. The only further distinction is that now instead of $N$ applications of \localsum, a single application of \globalsum is used to check that each node is locally isomorphic to its pre-image as defined by $\pi \in S_N$.
\end{proof}

\subsection{Proof of Theorem \ref{thm:undecidable-invariance}: Undecidability of \localmaxsemilinear \GNNs}
\thmUndecidabilityKeyInvariance*

\begin{proof}
We make use of Theorem \ref{thm:undecidableuddl}, which states that satisfiability for \uddl is undecidable, and Theorem~\ref{thm:uddllower} which states
that there is a computable procedure that takes a \uddl formula $\phi$ and constructs a semilinear $\localmax$ GNN $\agnn_\phi$ that performs the same node query as $\phi$. Note that $\agnn_\phi$
 must be key-invariant, since $\phi$ does not use keys. 
We will reduce satisfiability of \uddl to key-invariance.

Given a \uddl formula $\phi$, we construct a  semilinear $\localmax$ GNN $\agnn_{\textbf{test} \phi}$.that returns true on a pointed graph if and only if  $\agnn_\phi$  
returns true and the key of the distinguished node is above zero.
Clearly we can construct such a GNN from $\agnn_\phi$ by just adding a feature $\feature_0$ that is $1$ exactly when the key is above zero, and applying a final combination function 
that takes the minimum of the final feature of $\agnn_\phi$ and $\feature_0$.
If $\phi$ is unsatisfiable, then $\agnn_{\textbf{test} \phi}$ 
always returns $\false$, and hence is key-invariant. If $\phi$ is satisfiable, let $\apgraph$ be the witness. Then $\agnn_{\textbf{test} \phi}$  will return
true on keyings of $\apgraph$ where the key of the distinguished node is above zero, and false on other keyings: so it is not key-invariant.
\end{proof}

\end{document}